\documentclass[reqno,fleqn]{amsart}
\usepackage[T1]{fontenc}
\usepackage[latin9]{inputenc}
\usepackage{verbatim}
\usepackage{amstext}
\usepackage{amsthm}
\usepackage{amssymb}
\usepackage{marvosym}
\usepackage{stmaryrd}
\usepackage{graphicx}

\makeatletter
\numberwithin{equation}{section}
\numberwithin{figure}{section}
\theoremstyle{plain}
\newtheorem{thm}{\protect\theoremname}[section]
  \theoremstyle{plain}
  \newtheorem{fact}[thm]{\protect\factname}
  \theoremstyle{plain}
  \newtheorem{cor}[thm]{\protect\corollaryname}
  \theoremstyle{plain}
  \newtheorem{prop}[thm]{\protect\propositionname}
  \theoremstyle{plain}
  \newtheorem{lem}[thm]{\protect\lemmaname}
  \theoremstyle{remark}
  \newtheorem{claim}[thm]{\protect\claimname}
  \theoremstyle{remark}
  \newtheorem{rem}[thm]{\protect\remarkname}
  \theoremstyle{plain}
  \newtheorem*{thm*}{\protect\theoremname}
  \theoremstyle{plain}
  \newtheorem*{fact*}{\protect\factname}

\usepackage{amsmath}
\usepackage{sherstov}
\usepackage{small-caption}
\usepackage{xcolor}
\usepackage{setspace}
\usepackage{centernot}
\usepackage{bbm}

\renewcommand{\moo}{\{-1,+1\}}

\usepackage{hyperref}
\hypersetup{
pdfauthor = {A. A. Sherstov},
plainpages = false,
bookmarks = true,
bookmarksopen = true,
pdftex,
colorlinks = true,
citecolor = teal,
linkcolor = teal,
urlcolor  = teal,
}

\newtheoremstyle{myplain}      {10pt}{10pt}{\itshape}{}{\scshape}{.}{.5em}{}
\newtheoremstyle{mydefinition} {10pt}{10pt}{}{}{\scshape}{.}{.5em}{}
\newtheoremstyle{myremark} {10pt}{10pt}{}{}{\itshape}{.}{.5em}{}

\@namedef{thm}{\@thm{\let \thm@swap \@gobble \th@myplain }{thm}{Theorem}}
\@namedef{lem}{\@thm{\let \thm@swap \@gobble \th@myplain }{thm}{Lemma}}
\@namedef{cor}{\@thm{\let \thm@swap \@gobble \th@myplain }{thm}{Corollary}}
\@namedef{prop}{\@thm{\let \thm@swap \@gobble \th@myplain }{thm}{Proposition}}
\@namedef{claim}{\@thm{\let \thm@swap \@gobble \th@myplain }{thm}{Claim}}
\@namedef{fact}{\@thm{\let \thm@swap \@gobble \th@myplain }{thm}{Fact}}
\@namedef{defn}{\@thm{\let \thm@swap \@gobble \th@mydefinition }{thm}{Definition}}
\@namedef{rem}{\@thm{\let \thm@swap \@gobble \th@mydefinition }{thm}{Remark}}
\@namedef{example}{\@thm{\let \thm@swap \@gobble \th@mydefinition }{thm}{Example}}

\@namedef{thm*}{\@thm {\th@myplain }{}{Theorem}}
\@namedef{lem*}{\@thm {\th@myplain }{}{Lemma}}
\@namedef{cor*}{\@thm {\th@myplain }{}{Corollary}}
\@namedef{prop*}{\@thm {\th@myplain }{}{Proposition}}
\@namedef{claim*}{\@thm {\th@myplain }{}{Claim}}
\@namedef{fact*}{\@thm {\th@myplain }{}{Fact}}
\@namedef{defn*}{\@thm {\th@mydefinition }{}{Definition}}
\@namedef{rem*}{\@thm {\th@mydefinition }{}{Remark}}
\@namedef{example*}{\@thm {\th@mydefinition }{}{Example}}

\renewcommand{\mathcal}[1]{\mathscr{#1}}

\makeatletter
\def\@seccntformat#1{%
  \protect\textup{%
    \protect\@secnumfont
    \expandafter\protect\csname format#1\endcsname 
    \csname the#1\endcsname
    \protect\@secnumpunct
  }%
}

\makeatletter
\newcommand \SparseDotfill {\leavevmode \leaders \hb@xt@ .7em{\hss .\hss }\hfill \kern \z@}
\makeatother	

\makeatletter
\def\@tocline#1#2#3#4#5#6#7{\relax
  \ifnum #1>\c@tocdepth 
  \else
    \par \addpenalty\@secpenalty\addvspace{\ifnum #1=1 2mm \else #2\fi}%
    \begingroup \hyphenpenalty\@M
    \@ifempty{#4}{%
      \@tempdima\csname r@tocindent\number#1\endcsname\relax
    }{%
      \@tempdima#4\relax
    }%
    \parindent\z@ \leftskip#3\relax \advance\leftskip\@tempdima\relax
    \rightskip\@pnumwidth plus4em \parfillskip-\@pnumwidth
          \ifnum #1=1 \bfseries #5\else #5\fi 
   \leavevmode\hskip-\@tempdima
      \ifcase #1
       \or\or \hskip 1em \or \hskip 2em \else \hskip 3em \fi%
#6     \nobreak\relax
{\ifnum #1=1\hfill \else \SparseDotfill\fi}
 \hbox to\@pnumwidth{\@tocpagenum{
    \ifnum #1=1 \bfseries \fi #7}}\par
    \nobreak
    \endgroup
  \fi}
\makeatother

\DeclareMathOperator{\circulant}{circ}
\DeclareMathOperator{\diag}{diag}
\newcommand{\iu}{\mathbf{i}}
\DeclareMathOperator{\realpart}{Re}
\DeclareMathOperator{\imagpart}{Im}

\newcommand{\false}{\operatorname{\it false}}
\newcommand{\true}{\operatorname{\it true}}
\DeclareMathOperator{\fr}{frac}

\DeclareMathOperator{\upp}{UPP}
\DeclareMathOperator{\pp}{PP}
\DeclareMathOperator{\srank}{\rk_{\pm}}

\DeclareMathOperator{\Sgn}{\widetilde{sgn}}

\AtBeginDocument{
  
}

\makeatother

  \providecommand{\claimname}{Claim}
  \providecommand{\corollaryname}{Corollary}
  \providecommand{\factname}{Fact}
  \providecommand{\lemmaname}{Lemma}
  \providecommand{\propositionname}{Proposition}
  \providecommand{\remarkname}{Remark}
  \providecommand{\theoremname}{Theorem}
\providecommand{\theoremname}{Theorem}

\begin{document}

\title[The Hardest Halfspace]{The Hardest Halfspace}

\author{Alexander A. Sherstov}

\thanks{$^{*}$ Computer Science Department, UCLA, Los Angeles, CA~90095.
{\large{}\Letter ~}\texttt{sherstov@cs.ucla.edu }Supported by NSF
CAREER award CCF-1149018 and an Alfred P. Sloan Foundation Research
Fellowship.}
\begin{abstract}
We study the approximation of halfspaces $h:\zoon\to\zoo$ in the
infinity norm by polynomials and rational functions of any given degree.
Our main result is an explicit construction of the ``hardest'' halfspace,
for which we prove polynomial and rational approximation lower bounds
that match the trivial upper bounds achievable for all halfspaces.
This completes a lengthy line of work started by Myhill and Kautz
(1961).

As an application, we construct a communication problem
that achieves essentially the largest possible separation, of
$O(n)$ versus $2^{-\Omega(n)},$ between the sign-rank
and discrepancy. Equivalently, our problem exhibits
a gap of $\log n$ versus $\Omega(n)$ between the communication complexity
with \emph{unbounded} versus \emph{weakly unbounded} error, improving
quadratically on previous constructions and completing a line of work
started by Babai, Frankl, and Simon~(FOCS~1986). Our results further
generalize to the $k$-party number-on-the-forehead model, where we
obtain an explicit separation of $\log n$ versus $\Omega(n/4^{n})$
for communication with unbounded versus weakly unbounded error. This
gap is a quadratic improvement on previous work and matches the state
of the art for number-on-the-forehead lower bounds.
\end{abstract}

\maketitle
\belowdisplayskip=14pt plus 2pt minus 5pt 
\abovedisplayskip=14pt plus 2pt minus 5pt 
\thispagestyle{empty}

\newpage{}\thispagestyle{empty}
\hypersetup{linkcolor=black} 
~
\vspace{-15mm}\tableofcontents{}\newpage{}

\hypersetup{linkcolor=teal} 
\thispagestyle{empty}

\section{Introduction}

Representations of Boolean functions by real polynomials play a central
role in theoretical computer science. The notion of \emph{approximating}
a Boolean function $f\colon\zoon\to\moo$ pointwise by polynomials
of given degree has been particularly fruitful. Formally, let $E(f,d)$
denote the minimum error in an infinity-norm approximation of $f$
by a real polynomial of degree at most $d$: 
\[
E(f,d)=\min_{p}\{\|f-p\|_{\infty}:\deg p\leq d\}.
\]
This quantity clearly ranges between $0$ and $1$ for any function
$f\colon\zoon\to\moo$. In more detail, we have $0=E(f,n)\leq E(f,n-1)\leq\cdots\leq E(f,0)\leq1$,
where the first equality holds because any such $f$ is representable
exactly by a polynomial of degree at most $n$. The study of the polynomial
approximation of Boolean functions dates back to the pioneering work
in the 1960s by Myhill and Kautz~\cite{myhill-kautz61} and Minsky
and Papert~\cite{minsky88perceptrons}. This line of research has
grown remarkably over the decades, with numerous connections discovered
to other subjects in theoretical computer science. Lower bounds for
polynomial approximation have complexity-theoretic applications, whereas
upper bounds are a tool in algorithm design. In the former category,
polynomial approximation has enabled significant progress in circuit
complexity~\cite{beigel91rational,aspnes91voting,krause94depth2mod,KP98threshold,sherstov07ac-majmaj,beame-huyn-ngoc09multiparty-focs},
quantum query complexity~\cite{beals-et-al01quantum-by-polynomials,aaronson-shi04distinctness,ambainis05collision,BKT17poly-strikes-back},
and communication complexity~\cite{buhrman-dewolf01polynomials,razborov02quantum,buhrman07pp-upp,sherstov07ac-majmaj,sherstov07quantum,RS07dc-dnf,lee-shraibman08disjointness,chatt-ada08disjointness,dual-survey,beame-huyn-ngoc09multiparty-focs,sherstov12mdisj,sherstov13directional}.
On the algorithmic side, polynomial approximation underlies many of
the strongest results obtained to date in computational learning~\cite{tt99DNF-incl-excl,KS01dnf,KOS:02,KKMS,odonnell03degree,ACRSZ07nand},
differentially private data release~\cite{tuv12releasing-marginals,ctuw14release-of-marginals},
and algorithm design in general~\cite{linial-nisan90incl-excl,kahn96incl-excl,sherstov07inclexcl-ccc}.

\subsection{The hardest halfspace}

Myhill and Kautz's work~\cite{myhill-kautz61} six decades ago, and
many of the papers that followed~\cite{myhill-kautz61,muroga71threshold,siu91small-weights,paturi92approx,beigel94perceptrons,hastad94weights,sherstov09hshs,sherstov09opthshs,thaler14omb},
focused on \emph{halfspaces}. Also known as a linear threshold function,
a halfspace is any function $h\colon\zoon\to\moo$ representable as
$h(x)=\sign(\sum_{i=1}^{n}z_{i}x_{i}-\theta)$ for some fixed reals
$z_{1},z_{2},\ldots,z_{n},\theta.$ The fundamental question taken
up in this line of research is: how well can halfspaces be approximated
by polynomials of given degree? An early finding, due to Muroga~\cite{muroga71threshold},
was the upper bound
\begin{equation}
E(h,1)\leq1-\frac{1}{n^{\Theta(n)}}\label{eq:muroga}
\end{equation}
for every halfspace $h$ in $n$ variables. In words, every halfspace
can be approximated pointwise by a linear polynomial to error just
barely smaller than the trivial bound of~$1$. Many authors pursued
matching lower bounds on $E(h,1)$ for specific halfspaces $h$, culminating
in an explicit construction by H\aa stad~\cite{hastad94weights}
that matches Muroga's bound~(\ref{eq:muroga}).

The study of $E(h,d)$ for $d\geq2$ proved to be challenging. For
a long time, essentially the only result was the lower bound $E(h,d)\geq1-2^{-\Theta(n/d^{2})+1}$
due to Beigel~\cite{beigel94perceptrons}, where $h$ is the so-called
\emph{odd-max-bit} halfspace\emph{. }Paturi~\cite{paturi92approx}
proved the incomparable lower bound $E(h,\Theta(n))\ge1/3$, where
$h$ is the majority function on $n$ bits. Much later, the bound
$E(h,\Theta(\sqrt{n}))\geq1-2^{-\Theta(\sqrt{n})}$ was obtained in~\cite{sherstov09hshs}
for an explicit halfspace. This fragmented state of affairs persisted
until the question was resolved completely in~\cite{sherstov09opthshs},
with an \emph{existence proof} of a halfspace $h$ such that $E(h,d)\ge1-2^{-\Theta(n)}$
for $d=1,2,\ldots,\Theta(n).$ This result is clearly as strong as
one could hope for, since it essentially matches Muroga's upper bound
for approximation by \emph{linear} polynomials\emph{.} The work in~\cite{sherstov09opthshs}
further determined the minimum error, denoted $R(h,d)$, to which
this $h$ can be approximated by a degree-$d$ rational function,
showing that this quantity too is as large for $h$ as it can be for
any halfspace. Explicitly constructing a halfspace with these properties
is our main technical contribution:
\begin{thm}
\label{thm:MAIN-approx}There is an algorithm that takes as input
an integer $n\geq1,$ runs in time polynomial in $n,$ and outputs
a halfspace $h_{n}\colon\zoon\to\moo$ with
\begin{align*}
E(h_{n},d) & \geq1-2^{-\Omega(n)}, &  & d=1,2,\ldots,\lfloor cn\rfloor,\\
R(h_{n},d) & \geq1-2^{-\Omega(n/d)}, &  & d=1,2,\ldots,\lfloor cn\rfloor,
\end{align*}
where $c>0$ is an absolute constant.
\end{thm}

\noindent Classic bounds for the approximation of the sign function
imply that for any $d,$ the lower bounds in Theorem~\ref{thm:MAIN-approx}
are essentially the best possible for any halfspace on $n$ variables
(see Sections~\ref{subsec:Polynomial-approximation} and~\ref{subsec:Rational-approximation}
for details). Thus, the construction of Theorem~\ref{thm:MAIN-approx}
is the ``hardest'' halfspace from the point of view of approximation
by polynomials and rational functions.

Theorem~\ref{thm:MAIN-approx} is not a de-randomization of the existence
proof in~\cite{sherstov09opthshs}, which incidentally we are still
unable to de-randomize. Rather, it is based on a new and simpler approach,
presented in detail at the end of this section. Given the role that
halfspaces play in theoretical computer science, we see Theorem~\ref{thm:MAIN-approx}
as answering a basic question of independent interest. In addition,
Theorem~\ref{thm:MAIN-approx} has applications to communication
complexity and computational learning, which we now discuss.

\subsection{\label{sec:disc-vs-signrank}Discrepancy vs. sign-rank}

Consider the standard model of randomized communication~\cite{ccbook},
which features players Alice and Bob and a Boolean function $F\colon X\times Y\to\moo.$
On input $(x,y)\in X\times Y,$ Alice and Bob receive the arguments
$x$ and $y,$ respectively. Their objective is to compute $F$ on
any given input with minimal communication. To this end, each player
privately holds an unlimited supply of uniformly random bits which
he or she can use in deciding what message to send at any given point
in the protocol. The \emph{cost} of a protocol is the total number
of bits exchanged by Alice and Bob in a worst-case execution. The\emph{
$\epsilon$-error randomized communication complexity of $F$}, denoted
\emph{$R_{\epsilon}(F)$,} is the least cost of a protocol that computes
$F$ with probability of error at most $\epsilon$ on every input.

Our interest in this paper is in communication protocols with error
probability close to that of random guessing, $1/2.$ There are two
standard ways to define the complexity of a function $F$ in this
setting, both inspired by probabilistic polynomial time for Turing
machines~\cite{gill77pp}: 
\[
\upp(F)=\inf_{0\leq\epsilon<1/2}R_{\epsilon}(F)
\]
and 
\[
\pp(F)=\inf_{0\leq\epsilon<1/2}\left\{ R_{\epsilon}(F)+\log_{2}\left(\frac{1}{\frac{1}{2}-\epsilon}\right)\right\} .
\]
The former quantity, introduced by Paturi and Simon~\cite{paturi86cc},
is called the communication complexity of $F$ with \emph{unbounded
error}, in reference to the fact that the error probability can be
arbitrarily close to $1/2.$ The latter quantity, proposed by Babai
et al.~\cite{BFS86cc}, includes an additional penalty term that
depends on the error probability. We refer to $\pp(F)$ as the communication
complexity of $F$ with \emph{weakly unbounded error}. For all functions
$F\colon\zoon\times\zoon\to\moo,$ one has the trivial bounds $\upp(F)\leq\pp(F)\leq n+2.$
These two complexity measures give rise to corresponding \emph{complexity
classes} in communication complexity theory, defined in the seminal
paper of Babai et al.~\cite{BFS86cc}. Formally, $\UPP$ is the class
of families $\{F_{n}\}_{n=1}^{\infty}$ of communication problems
$F_{n}\colon\zoon\times\zoon\to\moo$ whose unbounded-error communication
complexity is at most polylogarithmic in $n.$ Its counterpart $\PP$
is defined analogously for the complexity measure $\pp$. 

These two models of large-error communication are synonymous with
two central notions in communication complexity: \emph{sign-rank}
and \emph{discrepancy}, defined formally in Sections~\ref{subsec:Communication-complexity}
and~\ref{subsec:Discrepancy}. In more detail, Paturi and Simon~\cite{paturi86cc}
proved that the communication complexity of any problem with unbounded
error is characterized up to an additive constant by the sign-rank
of its communication matrix, $[F(x,y)]_{x,y}.$ Analogously, Klauck~\cite{klauck01quantum,klauck01quantum-journal}
showed that the communication complexity of any problem $F\colon\zoon\times\zoon\to\moo$
with weakly unbounded error is essentially characterized in terms
of the discrepancy of $F$. Discrepancy and sign-rank enjoy a rich
mathematical life~\cite{linial04sign,sherstov07halfspace-mat,sherstov07cc-prod-nonprod,LS08learning-cc}
outside communication complexity, which further motivates the study
of $\PP$ and $\UPP$ as fundamental complexity classes.

Communication with weakly unbounded error is by definition no more
powerful than unbounded-error communication, and for twenty years
after the paper of Babai et al.~\cite{BFS86cc} it was unknown whether
this containment is proper. Buhrman et al.~\cite{buhrman07pp-upp}
and the author~\cite{sherstov07halfspace-mat} answered this question
in the affirmative, independently and with unrelated techniques. These
papers exhibited functions $F\colon\zoon\times\zoon\to\moo$ with
an exponential gap between communication complexity with unbounded
error versus weakly unbounded error: $\upp(F)=O(\log n)$ in both
works, versus $\pp(F)=\Omega(n^{1/3})$ in \cite{buhrman07pp-upp}
and $\pp(F)=\Omega(\sqrt{n})$ in~\cite{sherstov07halfspace-mat}.
In complexity-theoretic notation, these results show that $\PP\subsetneq\UPP$.
A simpler alternate proof of the result of Buhrman et al.~\cite{buhrman07pp-upp}
was given in~\cite{sherstov07quantum} using the pattern matrix method.
More recently, Thaler~\cite{thaler14omb} exhibited another, remarkably
simple communication problem $F\colon\zoon\times\zoon\to\moo,$ with
communication complexity $\upp(F)=O(\log n)$ and $\pp(F)=\Omega(n/\log n)^{2/5}.$ 

To summarize, the strongest explicit separation of communication complexity
with unbounded versus weakly unbounded error prior to our work was
the separation of $O(\log n)$~versus~$\Omega(\sqrt{n})$ from twelve
years ago~\cite{sherstov07halfspace-mat}. The \emph{existence} of
a communication problem with a quadratically larger gap, of $O(\log n)$~versus~$\Omega(n)$,
follows from the work in~\cite{sherstov09opthshs}. This state of
affairs parallels other instances in communication complexity, such
as the $\mathsf{{P}}$ versus $\BPP$ question in multiparty communication~\cite{BDPW07p-rp},
where the best existential separations are much stronger than the
best explicit ones. There is considerable interest in communication
complexity in explicit separations because they provide a deeper and
more complete understanding of the complexity classes, whereas the
lack of a strong explicit separation indicates a basic gap in our
knowledge. As an application of Theorem~\ref{thm:MAIN-approx}, we
obtain:
\begin{thm}
\label{thm:MAIN-pp-upp}There is a communication problem $F_{n}\colon\zoon\times\zoon\to\moo,$
defined by
\begin{align}
F_{n}(x,y)=\sign\left(w_{0}+\sum_{i=1}^{n}w_{i}x_{i}y_{i}\right)\label{eq:main-pp-upp-form}
\end{align}
for some explicitly given reals $w_{0},w_{1},\dots,w_{n},$ such that
\begin{align*}
\upp(F_{n}) & \leq\log n+O(1),\\
\pp(F_{n}) & =\Omega(n).
\end{align*}
Moreover, 
\begin{align*}
\srank(F_{n}) & \leq n+1,\\
\disc(F_{n}) & =2^{-\Omega(n)}.
\end{align*}
\end{thm}

\noindent Theorem~\ref{thm:MAIN-pp-upp} gives essentially the strongest
possible separation of the communication classes $\PP$ and $\UPP$,
improving quadratically on previous constructions and matching the
previous nonconstructive separation. Another compelling aspect of
the theorem is the simple form~(\ref{eq:main-pp-upp-form}) of the
communication problem in question. The last two bounds in Theorem~\ref{thm:MAIN-pp-upp}
state that $F_{n}$ has sign-rank at most $n+1$ and discrepancy $2^{-\Omega(n)}$,
which is essentially the strongest possible separation. The best previous
construction~\cite{sherstov07halfspace-mat} achieved sign-rank $O(n)$
and discrepancy $2^{-\Omega(\sqrt{n})}$.

We further generalize Theorem~\ref{thm:MAIN-pp-upp} to the \emph{number-on-the-forehead
$k$-party model}, the standard formalism of multiparty communication.
Analogous to two-party communication, the $k$-party model has its
own classes $\UPP_{k}$ and $\PP_{k}$ of problems solvable efficiently
by protocols with unbounded error and weakly unbounded error, respectively.
Their formal definitions can be found in Section~\ref{subsec:Communication-complexity}.
In this setting, we prove:
\begin{thm}
\label{thm:MAIN-pp-upp-multiparty}There is a $k$-party communication
problem $F_{n}\colon(\zoon)^{k}\to\moo,$ defined by
\[
F_{n}(x_{1},x_{2},\ldots,x_{k})=\sign\left(w_{0}+\sum_{i=1}^{n}w_{i}x_{1,i}x_{2,i}\cdots x_{k,i}\right)
\]
for some explicitly given reals $w_{0},w_{1},\dots,w_{n},$ such that
\begin{align*}
\upp(F_{n}) & \leq\log n+O(1),\\
\pp(F_{n}) & =\Omega\left(\frac{n}{4^{k}}\right),\\
\disc(F_{n}) & =\exp\left(-\Omega\left(\frac{n}{4^{k}}\right)\right).
\end{align*}
\end{thm}

\noindent Theorem~\ref{thm:MAIN-pp-upp-multiparty} gives essentially
the strongest possible explicit separation of the $k$-party communication
complexity classes $\UPP_{k}$ and $\PP_{k}$ for up to $k\leq(0.5-\epsilon)\log n$
parties, where $\epsilon>0$ is an arbitrary constant. The previous
best explicit separation~\cite{chattopadhyay-mande16multiparty-pp-upp,sherstov16multiparty-pp-upp}
of these classes was quadratically weaker, with communication complexity
$\Omega(\sqrt{n}/4^{k})$ for unbounded error and $O(\log n)$ for
weakly unbounded error. The communication lower bound in Theorem~\ref{thm:MAIN-pp-upp-multiparty}
reflects the state of the art in the area, in that the strongest lower
bound for any explicit communication problem $F\colon(\zoon)^{k}\to\moo$
to date is $\Omega(n/2^{k})$ due to Babai et al.~\cite{bns92}.

\subsection{Computational learning}

A \emph{sign-representing polynomial} for a given function $f\colon\zoon\to\moo$
is any real polynomial $p$ such that $f(x)=\sign p(x)$ for all $x.$
The minimum degree of a sign-representing polynomial for $f$ is called
the \emph{threshold degree} of $f,$ denoted $\degthr(f).$ Clearly
$0\leq\degthr(f)\leq n$ for every Boolean function $f$ on $n$ variables.
The reader can further verify that sign-representation is equivalent
to pointwise approximation with error strictly less than, but arbitrarily
close to, the trivial error of~$1$. Sign-representing polynomials
are appealing from a learning standpoint because they immediately
lead to efficient learning algorithms. Indeed, any function of threshold
degree $d$ is by definition a linear combination of $N={n \choose 0}+{n \choose 1}+\cdots+{n \choose d}$
monomials and can thus be viewed as a halfspace in $N$ dimensions.
As a result, $f$ can be PAC learned~\cite{valiant84pac} under arbitrary
distributions in time polynomial in $N,$ using a variety of halfspace
learning algorithms.

The study of sign-representing polynomials started fifty years ago
with the seminal monograph of Minsky and Papert~\cite{minsky88perceptrons},
who examined the threshold degree of several common functions. Since
then, the threshold degree approach has yielded the fastest known
PAC learning algorithms for notoriously hard concept classes, including
DNF formulas~\cite{KS01dnf} and AND-OR trees~\cite{ACRSZ07nand}.
Conspicuously absent from this list of success stories is the concept
class of \emph{intersections of halfspaces}. While solutions are known
to several restrictions of this learning problem~\cite{blum-kannan97intersection-of-halfspaces,KwekPitt:98,Vempala:97,arriaga98proj,KOS:02,KlivansServedio:04coltmargin,KLT09intersections-of-halfspaces},
no algorithm has been discovered for PAC learning the intersection
of even two halfspaces in time faster than $2^{\Theta(n)}.$ Known
hardness results, on the other hand, only apply to polynomially many
halfspaces or to proper learning, e.g.,~\cite{blum92trainingNN,ABFKP:04,focs06hardness,khot-saket08hs-and-hs}. 

This state of affairs has motivated a quest to determine the threshold
degree of the intersection of two halfspaces~\cite{minsky88perceptrons,odonnell03degree,klivans-thesis,sherstov09hshs,sherstov09opthshs}.
Prior to our work, the best lower bound was $\Omega(\sqrt{n})$ for
an explicit intersection of two halfspaces~\cite{sherstov09hshs},
complemented by a tight but highly nonconstructive $\Omega(n)$ lower
bound~\cite{sherstov09opthshs}. Using Theorem~\ref{thm:MAIN-approx},
we prove: 
\begin{thm}
\label{thm:MAIN-hshs} There is an $($explicitly given$)$ halfspace
$h_{n}\colon\zoon\to\moo$ such that
\[
\degthr(h_{n}\wedge h_{n})=\Omega(n).
\]
\end{thm}

\noindent The symbol $h_{n}\wedge h_{n}$ above stands for the intersection
of two copies of $h_{n}$ on disjoint sets of variables. In other
words, Theorem~\ref{thm:MAIN-hshs} constructs an explicit intersection
of two halfspaces whose threshold degree is asymptotically maximal,
$\Omega(n).$ While the nonconstructive $\Omega(n)$ lower bound of~\cite{sherstov09opthshs}
already ruled out the threshold degree approach as a way to learn
intersections of halfspaces, we see Theorem~\ref{thm:MAIN-hshs}
as contributing a key qualitative piece of the puzzle. Specifically,
it constructs a small and simple family of intersections of two halfspaces
that are off-limits to all known algorithmic approaches (namely, the
family obtained by applying $h_{n}\wedge h_{n}$ to different subsets
of the variables $x_{1},x_{2},\ldots,x_{4n}$).

\subsection{\label{subsec:Proof-overview}Proof overview}

Our solution has two main components: the construction of a sparse
set of integers that appear random modulo $m,$ and the univariatization
of a multivariate Boolean function. We describe each of these components
in detail.

\subsubsection*{Discrepancy of integer sets.}

Let $m\geq2$ be a given integer. Key to our work is the notion of
\emph{$m$-discrepancy}, which quantifies the pseudorandomness or
aperiodicity modulo $m$ of any given multiset of integers. It is
largely unrelated to the notion of discrepancy in communication complexity
(Section~\ref{sec:disc-vs-signrank}). Formally, the $m$-discrepancy
of a nonempty multiset $Z=\{z_{1},z_{2},\ldots,z_{n}\}$ is defined
as
\[
\disc(Z,m)=\max_{k=1,2,\ldots,m-1}\left|\frac{1}{n}\sum_{j=1}^{n}\omega^{kz_{j}}\right|,
\]
where $\omega$ is a primitive $m$-th root of unity. This fundamental
quantity arises in combinatorics and theoretical computer science,
e.g.,~\cite{gks86k-page-graphs-and-nondeterministic-TMs,ruzsa87essential-components,AIKPS90aperiodic-set,katz89character-sums,rsw93sets-uniform-in-arithmetic-progressions,alon-roichman94rando-cayley-graphs-and-expanders}.
The identity $1+\omega+\omega^{2}+\cdots+\omega^{m-1}=0$ for any
$m$-th root of unity $\omega\ne1$ implies that the set $Z=\{0,1,2,\ldots,m-1\}$
achieves the smallest possible $m$-discrepancy: $\disc(Z,m)=0.$
Much sparser sets with small $m$-discrepancy can be shown to exist
using the probabilistic method (Fact~\ref{fact:small-fourier-set-existence}
and Corollary~\ref{cor:small-fourier-set-existence}). Specifically,
one easily verifies for any constant $\epsilon>0$ the existence of
a set $Z\subseteq\{0,1,2,\ldots,m-1\}$ with $m$-discrepancy at most
$\epsilon$ and cardinality $O(\log m),$ an exponential improvement
in sparsity compared to the trivial set $\{0,1,2,\ldots,m-1\}.$ We
are aware of two efficient constructions of sparse sets with small
$m$-discrepancy, due to Ajtai et al.~\cite{AIKPS90aperiodic-set}
and Katz~\cite{katz89character-sums}. The approach of Ajtai et al.~is
elementary except for an appeal to the prime number theorem, whereas
Katz's construction relies on deep results in number theory. Neither
work appears to directly imply the kind of optimal de-randomization
that we require, namely, an algorithm that runs in time polynomial
in $\log m$ and produces a multiset of cardinality $O(\log m)$ with
$m$-discrepancy bounded away from~1. We obtain such an algorithm
by adapting the approach of Ajtai et al.~\cite{AIKPS90aperiodic-set}.

The centerpiece of the construction of Ajtai et al.~\cite{AIKPS90aperiodic-set}
is what the authors call the \emph{iteration lemma}, stated in this
paper as Theorem~\ref{thm:ajtai-iteration}. Its role is to reduce
the construction of a sparse set with small $m$-discrepancy to the
construction of sparse sets with small $p$-discrepancy, for primes
$p\ll m.$ Ajtai et al.~\cite{AIKPS90aperiodic-set} proved their
iteration lemma for $m$ prime, but we show that their argument readily
generalizes to arbitrary moduli $m$. By applying the iteration lemma
in a recursive manner, one reaches smaller and smaller primes. The
authors of~\cite{AIKPS90aperiodic-set}~continue this recursive
process until they reach primes $p$ so small that the trivial construction
$\{0,1,2,\ldots,p-1\}$ can be considered sparse. We proceed differently
and terminate the recursion after just two stages, at which point
the input size is small enough for brute force search based on the
probabilistic method. The final set that we construct has size logarithmic
in $m$ and $m$-discrepancy a small constant, as opposed to the superlogarithmic
size and $o(1)$ discrepancy in the work of Ajtai et al.~\cite{AIKPS90aperiodic-set}. 

We note that this modified approach additionally gives the first explicit
circulant expander on $n$ vertices of degree $O(\log n),$ which
is optimal and improves on the previous best degree bound of $(\log^{*}n)^{O(\log^{*}n)}\cdot O(\log n)$
due to Ajtai et al.~\cite{AIKPS90aperiodic-set}. Background on circulant
expanders, and the details of our expander construction, can be found
in Section~\ref{subsec:A-circulant-expander}.

\subsubsection*{Univariatization.}

We now describe the second major component of our proof. Consider
a halfspace $h_{n}(x)=\sign(\sum z_{i}x_{i}-\theta)$ in Boolean variables
$x_{1},x_{2},\ldots,x_{n},$ where the coefficients can be assumed
without loss of generality to be integers. Then the linear form $\sum z_{i}x_{i}-\theta$
ranges in the discrete set $\{\pm1,\pm2,\ldots,\pm N\}$, for some
integer $N$ proportionate to the magnitude of the coefficients. As
a result, one can approximate $h_{n}$ to any given error $\epsilon$
by approximating the sign function to $\epsilon$ on $\{\pm1,\pm2,\ldots,\pm N\}.$
This approach works for both rational approximation and polynomial
approximation. We think of it as the \emph{black-box} approach to
the approximation of $h_{n}$ because it uses the linear form $\sum z_{i}x_{i}-\theta$
rather than the individual bits. There is no reason to expect that
the black-box construction is anywhere close to optimal. Indeed, there
are halfspaces~\cite[Section~1.3]{sherstov09hshs} that can be approximated
to arbitrarily small error by a rational function of degree~$1$
but require a black-box approximant of degree $\Omega(n)$. Surprisingly,
we are able to construct a halfspace $h_{n}$ with exponentially large
coefficients for which the black-box approximant is essentially optimal.
As a result, tight lower bounds for the rational and polynomial approximation
of $h_{n}$ follow immediately from the univariate lower bounds for
approximating the sign function on $\{\pm1,\pm2,\pm3,\ldots,\pm2^{\Theta(n)}\}$.
The role of $h_{n}$ is to reduce the multivariate problem taken up
in this work to a well-understood univariate question, hence the term
\emph{univariatization}.

The construction of $h_{n}$ involves several steps. First, we study
the probability distribution of the weighted sum $z_{1}X_{1}+z_{2}X_{2}+\cdots+z_{n}X_{n}$
modulo $m$, where $z_{1},z_{2},\ldots,z_{n}$ are given integers
and the bits $X_{1},X_{2},\ldots,X_{n}\in\zoo$ are chosen uniformly
at random. We show that the distribution is exponentially close to
uniform whenever the multiset $\{z_{1},z_{2},\ldots,z_{n}\}$ has
$m$-discrepancy bounded away from~$1$. For the next step, fix any
multiset $\{z_{1},z_{2},\ldots,z_{n}\}$ with small $m$-discrepancy
and consider the linear map $L\colon\zoon\to\ZZ_{m}$ given by $L(x)=\sum z_{i}x_{i}.$
At this point in the proof, we know that for uniformly random $X\in\zoon$,
the probability distribution of $L(X)$ is exponentially close to
uniform. This implies that the characteristic functions of $L^{-1}(0),L^{-1}(1),\ldots,L^{-1}(m-1)$
have approximately the same Fourier spectrum up to degree $cn$, for
some constant $c>0$. We substantially strengthen this conclusion
by proving that there are probability distributions $\mu_{0},\mu_{1},\ldots,\mu_{m-1}$,
supported on $L^{-1}(0),L^{-1}(1),\ldots,L^{-1}(m-1)$, respectively,
such that the Fourier spectra of $\mu_{0},\mu_{1},\ldots,\mu_{m-1}$
are \emph{exactly} the same up to degree $cn.$ Our proof relies on
a general tool from~\cite[Theorem~4.1]{sherstov09opthshs}, proved
there using the Gershgorin circle theorem.

As our final step, we use $\mu_{0},\mu_{1},\ldots,\mu_{m-1}$ to construct
a halfspace in terms of $z_{1},z_{2},\ldots,z_{n}$ whose approximation
by rational functions and polynomials gives corresponding approximants
for the sign function on the discrete set $\{\pm1,\pm2,\ldots,\pm m\}$.
More generally, for any tuple $z_{1},z_{2},\ldots,z_{n}$, we define
an associated halfspace and prove a lower bound on $m$ in terms of
the discrepancy of the multiset $\{z_{1},z_{2},\ldots,z_{n}\}.$ Combining
this result with the efficient construction of an integer set with
small $m$-discrepancy for $m=2^{\Theta(n)}$, we obtain an explicit
halfspace $h_{n}\colon\zoon\to\moo$ whose approximation by polynomials
and rational functions is equivalent to the univariate approximation
of the sign function on $\{\pm1,\pm2,\pm3,\ldots,\pm2^{\Theta(n)}\}$.
Theorem~\ref{thm:MAIN-approx} now follows by appealing to known
lower bounds for the polynomial and rational approximation of the
sign function. To obtain the exponential separation of communication
complexity with unbounded versus weakly unbounded error (Theorem~\ref{thm:MAIN-pp-upp}),
we use the \emph{pattern matrix method}~\cite{sherstov07ac-majmaj,sherstov07quantum}
to ``lift'' the lower bound of Theorem~\ref{thm:MAIN-approx} to
a discrepancy bound. Finally, our result on the threshold degree of
the intersection of two halfspaces (Theorem~\ref{thm:MAIN-hshs})
works by combining the rational approximation lower bound of Theorem~\ref{thm:MAIN-approx}
with a structural result from~\cite{sherstov09hshs} on the sign-representation
of arbitrary functions of the form $f\wedge f.$

A key technical contribution of this paper is the identification of
$m$-discrepancy as a pseudorandom property that is weak enough to
admit efficient de-randomization and strong enough to allow the univariatization
of the corresponding halfspace. The previous, existential result in~\cite{sherstov09opthshs}
used a completely different and more complicated pseudorandom property
based on affine shifts of the Fourier transform on $\zoon,$ which
we have not been able to de-randomize. Apart from the construction
of a low-discrepancy set, our proof is simpler and more intuitive
than the existential proof in~\cite{sherstov09opthshs}. 

\section{\label{sec:prelim}Preliminaries}

We start with a review of the technical preliminaries. The purpose
of this section is to make the paper as self-contained as possible,
and comfortably readable by a broad audience. The expert reader should
therefore skim this section for notation or skip it altogether.

\subsection{Notation}

There are two common arithmetic encodings for the Boolean values:
the traditional encoding $\false\leftrightarrow0,\;\true\leftrightarrow1,$
and the Fourier-motivated encoding $\false\leftrightarrow1,\;\true\leftrightarrow-1.$
Throughout this manuscript, we use the former encoding for the domain
of a Boolean function and the latter for the range. With this convention,
Boolean functions are mappings $\zoon\to\moo$ for some $n.$ For
Boolean functions $f\colon\zoon\to\moo$ and $g\colon\zoom\to\moo,$
we let $f\circ g$ denote the coordinatewise composition of $f$ with
$g.$ Formally, $f\circ g\colon(\zoom)^{n}\to\moo$ is given by 
\begin{align}
(f\circ g)(x_{1},x_{2},\dots,x_{n})=f\left(\frac{1-g(x_{1})}{2},\frac{1-g(x_{2})}{2},\ldots,\frac{1-g(x_{n})}{2}\right),\label{eq:coordinatewise-composition}
\end{align}
where the linear map on the right-hand side serves the purpose of
switching between the distinct arithmetizations for the domain versus~range.
A \emph{partial function} $f$ on a set $X$ is a function whose domain
of definition, denoted $\dom f,$ is a nonempty proper subset of $X.$
We generalize coordinatewise composition $f\circ g$ to partial Boolean
functions $f$ and $g$ in the natural way. Specifically, $f\circ g$
is the Boolean function given by (\ref{eq:coordinatewise-composition}),
with domain the set of all inputs $(\ldots,x_{i},\dots)\in(\dom g)^{n}$
for which $(\ldots,(1-g(x_{i}))/2,\dots)\in\dom f.$ 

We use the following two versions of the sign function:
\begin{align*}
\sign x=\begin{cases}
-1 & \text{if }x<0,\\
0 & \text{if }x=0,\\
1 & \text{if }x>0,
\end{cases}\qquad\qquad\qquad & \Sgn x=\begin{cases}
-1 & \text{if }x<0,\\
1 & \text{if }x\geq0.
\end{cases}
\end{align*}
For a subset $\Xcal\subseteq\Re,$ we let $\sign|_{\Xcal}$ denote
the restriction of the sign function to $\Xcal.$ A \emph{halfspace
}for us is any Boolean function $h\colon\zoon\to\moo$ given by 
\[
h(x)=\sign\left(\sum_{i=1}^{n}w_{i}x_{i}-\theta\right)
\]
for some reals $w_{1},w_{2},\ldots,w_{n},\theta.$ The \emph{majority
function} $\MAJ_{n}\colon\zoon\to\moo$ is the halfspace defined by
\begin{align*}
\MAJ_{n}(x) & =-\sign\left(\sum_{i=1}^{n}x_{i}-\frac{n}{2}-\frac{1}{4}\right)\\
 & =\begin{cases}
-1 & \text{if }x_{1}+x_{2}+\cdots+x_{n}>n/2,\\
1 & \text{otherwise}.
\end{cases}
\end{align*}
Some authors define $\MAJ_{n}$ only for $n$ odd, in which case the
tiebreaker term $1/4$ can be omitted.

The complement and the power set of a set $S$ are denoted as usual
by $\overline{S}$ and $\Pcal(S)$, respectively. The symmetric difference
of sets $S$ and $T$ is $S\oplus T=(S\cap\overline{T})\cup(\overline{S}\cap T).$
Throughout this manuscript, we use brace notation as in $\{z_{1},z_{2},\ldots,z_{n}\}$
to specify \emph{multisets} rather than sets. The \emph{cardinality}
$|Z|$ of a finite multiset $Z$ is defined as the total number of
element occurrences in $Z$, with each element counted as many times
as it occurs. The equality and subset relations on multisets are defined
analogously, with the number of element occurrences taken into account.
For example, $\{1,1,2\}=\{1,2,1\}$ but $\{1,1,2\}\ne\{1,2\}$. Similarly,
$\{1,2\}\subseteq\{1,1,2\}$ but $\{1,1,2\}\nsubseteq\{1,2\}.$ 

The infinity norm of a function $f\colon\Xcal\to\Re$ is denoted $\|f\|_{\infty}=\sup_{x\in\Xcal}|f(x)|.$
For real-valued functions $f$ and $g$ and a nonempty finite subset
$\Xcal$ of their domain, we write
\[
\langle f,g\rangle_{\Xcal}=\frac{1}{|\Xcal|}\sum_{x\in\Xcal}f(x)g(x).
\]
We will often use this notation with $\Xcal$ a nonempty \emph{proper}
subset of the domain of $f$ and $g.$ We let $\ln x$ and $\log x$
stand for the natural logarithm of $x$ and the logarithm of $x$
to base $2,$ respectively. The binary entropy function $H\colon[0,1]\to[0,1]$
is given by $H(p)=-p\log p-(1-p)\log(1-p)$ and is strictly increasing
on $[0,1/2].$ The following bound is well known~\cite[p.~283]{jukna01extremal}:
\begin{align}
\sum_{i=0}^{k}{n \choose i}\leq2^{H(k/n)n}, &  & k=0,1,2,\dots,\left\lfloor \frac{n}{2}\right\rfloor .\label{eqn:entropy-bound}
\end{align}
For a complex number $x,$ we denote the real part, imaginary part,
and complex conjugate of $x$ as usual by $\realpart(x),$ $\imagpart(x),$
and $\overline{x},$ respectively. We typeset the imaginary unit $\iu$
in boldface to distinguish it from the index variable $i$. 

For an arbitrary integer $a$ and a positive integer $m$, recall
that $a\bmod m$ denotes the unique element of $\{0,1,2,\ldots,m-1\}$
that is congruent to $a$ modulo $m.$ For an integer $m\geq2,$ the
symbols $\mathbb{Z}_{m}$ and $\mathbb{Z}_{m}^{*}$ refer to the ring
of integers modulo $m$ and the multiplicative group of integers modulo
$m,$ respectively. For a multiset $Z=\{z_{1},z_{2},\ldots,z_{n}\}$
of integers, we adopt the standard notation
\begin{align}
-Z & =\{-z_{1},\ldots,-z_{n}\},\label{eq:multi1}\\
aZ & =\{az_{1},\ldots,az_{n}\},\\
Z+b & =\{z_{1}+b,\ldots,z_{n}+b\},\\
Z\bmod m & =\{z_{1}\bmod m,\ldots,z_{n}\bmod m\}.\label{eq:multi4}
\end{align}
Note that the multisets in (\ref{eq:multi1})\textendash (\ref{eq:multi4})
each have cardinality $n,$ the same as the original set $Z$. We
often use these shorthands in combination, as in $\text{\ensuremath{(aZ+b)\bmod m}}=\{(az_{1}+b)\bmod m,\ldots,(az_{n}+b)\bmod m\}.$

For a logical condition $C,$ we use the Iverson bracket
\[
\I[C]=\begin{cases}
1 & \text{if \ensuremath{C} holds,}\\
0 & \text{otherwise.}
\end{cases}
\]
The following concentration inequality, due to Hoeffding~\cite{hoeffding-bound},
is well-known.
\begin{fact}[Hoeffding's Inequality]
\label{fact:hoeffding}Let $X_{1},X_{2},\ldots,X_{n}$ be independent
random variables with $X_{i}\in[a_{i},b_{i}].$ Let
\[
p=\sum_{i=1}^{n}\Exp X_{i}.
\]
Then
\[
\Prob\left[\left|\sum_{i=1}^{n}X_{i}-p\right|\geq\delta\right]\leq2\exp\left(-\frac{2\delta^{2}}{\sum_{i=1}^{n}(b_{i}-a_{i})^{2}}\right).
\]
\end{fact}

\noindent In Fact~\ref{fact:hoeffding} and throughout this paper,
we typeset random variables using capital letters. %

\subsection{\label{subsec:Number-theoretic-preliminaries}Number-theoretic preliminaries}

For positive integers $a$ and $b$ that are relatively prime, $(a^{-1})_{b}\in\{1,2,\ldots,b-1\}$
denotes the multiplicative inverse of $a$ modulo $b.$ The following
fact is well-known and straightforward to verify;~cf.~\cite{AIKPS90aperiodic-set}.
\begin{fact}
\label{fact:rel-prime}For any positive integers $a$ and $b$ that
are relatively prime,
\begin{equation}
\frac{(a^{-1})_{b}}{b}+\frac{(b^{-1})_{a}}{a}-\frac{1}{ab}\in\mathbb{Z}.\label{eq:rel-prime}
\end{equation}
\end{fact}

\begin{proof}
We have $a(a^{-1})_{b}+b(b^{-1})_{a}\equiv b(b^{-1})_{a}\equiv1\pmod a,$
and analogously $a(a^{-1})_{b}+b(b^{-1})_{a}\equiv a(a^{-1})_{b}\equiv1\pmod b.$
Thus, $a(a^{-1})_{b}+b(b^{-1})_{a}-1$ is divisible by both $a$ and
$b.$ Since $a$ and $b$ are relatively prime, we conclude that $a(a^{-1})_{b}+b(b^{-1})_{a}-1$
is divisible by $ab,$ which is equivalent to~(\ref{eq:rel-prime}).
\end{proof}
Recall that the \emph{prime counting function} $\pi(x)$ for a real
argument $x\geq0$ evaluates to the number of prime numbers less than
or equal to $x.$ In what follows, it will be clear from the context
whether $\pi$ refers to $3.14159\ldots$ or the prime counting function.
The asymptotic growth of the latter is given by the \emph{prime number
theorem}, which states that $\pi(n)\sim n/\ln n.$ Many explicit bounds
on $\pi(n)$ are known, such as the following theorem of Rosser~\cite{rosser41primes}.
\begin{fact}[Rosser]
\label{fact:PNT}For $n\geq55,$
\[
\frac{n}{\ln n+2}<\pi(n)<\frac{n}{\ln n-4}.
\]
\end{fact}

\noindent 

\noindent The number of distinct prime divisors of a natural number
$n$ is denoted $\nu(n)$. We will need the following first-principles
bound on $\nu(n)$, which is asymptotically tight for infinitely many
$n.$
\begin{fact}
\label{fact:num-prime-factors}The number of distinct prime divisors
of $n$ obeys
\begin{equation}
(\nu(n)+1)!\leq n.\label{eq:nu-n-upper-bound}
\end{equation}
In particular,
\begin{equation}
\nu(n)\leq(1+o(1))\frac{\ln n}{\ln\ln n}.\label{eq:nu-n-asymptotic-bound}
\end{equation}
\end{fact}

\begin{proof}
An integer $n\geq1$ has by definition $\nu(n)$ distinct prime divisors.
Letting $p_{k}$ denote the $k$-th prime, we have
\begin{align*}
\ln n & \geq\ln p_{1}p_{2}\ldots p_{\nu(n)}\\
 & \geq\sum_{k=1}^{\nu(n)}\ln(k+1)\\
 & \geq\int_{1}^{\nu(n)}\ln x\;dx\\
 & =\nu(n)\ln\nu(n)-\nu(n)+1,
\end{align*}
where the second step uses the trivial estimate $p_{k}\geq k+1.$
The second step in this derivation settles~(\ref{eq:nu-n-upper-bound}),
whereas the last step settles~(\ref{eq:nu-n-asymptotic-bound}). 
\end{proof}
\noindent 

\subsection{Matrix analysis}

For an arbitrary set $X$ such as $X=\mathbb{C}$ or $X=\{-1,1\},$
the symbol $X^{n\times m}$ denotes the family of $n\times m$ matrices
with entries in $X$. The symbols $I_{n}$ and $J_{n,m}$ stand for
the order-$n$ identity matrix and the $n\times m$ matrix of all
ones, respectively. When the dimensions of the matrix are clear from
the context, we omit the subscripts and write simply $I$ or $J.$
The shorthand $\diag(d_{1},d_{2},\ldots,d_{n})$ refers to the diagonal
matrix with entries $d_{1},d_{2},\ldots,d_{n}$ on the diagonal:
\[
\diag(d_{1},d_{2},\ldots,d_{n})=\begin{bmatrix}d_{1} & 0 & \cdots & 0\\
0 & d_{2} & \cdots & 0\\
\vdots & \vdots & \ddots & \vdots\\
0 & 0 & \cdots & d_{n}
\end{bmatrix}.
\]
For a matrix $M=[M_{i,j}],$ recall that its complex conjugate is
given by $\overline{M}=[\overline{M_{i,j}}]$. The transpose and conjugate
transpose of $M$ are denoted $M^{T}$ and $M^{*}=\overline{M}{}^{T},$
respectively. The conjugation, transpose, and conjugate transpose
operations apply as a special case to vectors, which we view as matrices
with a single column. We use the familiar matrix norms $\|M\|_{\infty}=\max|M_{ij}|$
and $\|M\|_{1}=\sum|M_{ij}|.$ Again, these definitions carry over
to vectors as a special case. A matrix $M\in\mathbb{C}^{n\times n}$
is called \emph{unitary} if $MM^{*}=M^{*}M=I.$

A \emph{circulant matrix} is any matrix $C\in\CC^{m\times m}$ of
the form
\begin{align}
C & =\begin{bmatrix}c_{0} & c_{1} & c_{2} & \cdots & c_{m-2} & c_{m-1}\\
c_{m-1} & c_{0} & c_{1} & \cdots & c_{m-3} & c_{m-2}\\
c_{m-2} & c_{m-1} & c_{0} & \cdots & c_{m-4} & c_{m-3}\\
\vdots & \vdots & \vdots & \ddots & \vdots & \vdots\\
c_{2} & c_{3} & c_{4} & \cdots & c_{0} & c_{1}\\
c_{1} & c_{2} & c_{3} & \cdots & c_{m-1} & c_{0}
\end{bmatrix}\label{eq:circulant}
\end{align}
for some $c_{0},c_{1},\ldots,c_{m-1}\in\CC.$ Thus, every row of $C$
is obtained by a circular shift of the previous row one entry to the
right. We let $\circulant(c_{0},c_{1},\ldots,c_{m-1})$ denote the
right-hand side of~(\ref{eq:circulant}). In this notation, $\circulant(1,0,\ldots,0)=I$
and $\circulant(1,1,\ldots,1)=J.$ The eigenvalues and eigenvectors
of a circulant matrix are well-known and straightforward to determine.
For the reader's convenience, we include the short derivation below
in Fact~\ref{fact:circulant-eigen} and Corollary~\ref{cor:circulant-diagonalization}.
\begin{fact}
\label{fact:circulant-eigen}Let $C=\circulant(c_{0},c_{1},\ldots,c_{m-1})$
be a circulant matrix. Then for every $m$-th root of unity $\omega,$
the vector
\begin{equation}
\begin{bmatrix}1\\
\omega\\
\omega^{2}\\
\vdots\\
\omega^{m-1}
\end{bmatrix}\label{eq:eigenvector}
\end{equation}
is an eigenvector of $C$ with eigenvalue $\sum_{j=0}^{m-1}c_{j}\omega^{j}.$
\end{fact}

\begin{proof}
Let $v$ denote the vector in~(\ref{eq:eigenvector}). Then for $k=1,2,3,\ldots,m,$
\begin{align*}
(Cv)_{k} & =\sum_{j=0}^{m-1}c_{(j-k+1)\bmod m}\;\omega^{j}\\
 & =\left(\sum_{j=0}^{m-1}c_{(j-k+1)\bmod m}\;\omega^{j-k+1}\right)v_{k}\\
 & =\left(\sum_{j=0}^{m-1}c_{(j-k+1)\bmod m}\;\omega^{(j-k+1)\bmod m}\right)v_{k}\\
 & =\left(\sum_{j=0}^{m-1}c_{j}\omega^{j}\right)v_{k},
\end{align*}
where the third step uses $\omega^{m}=1.$
\end{proof}
\noindent As a corollary to Fact~\ref{fact:circulant-eigen}, one
recovers the full complement of eigenvalues for any circulant matrix
$C$ and furthermore learns that $C$ is unitarily similar to a diagonal
matrix. In the statement below, recall that a \emph{primitive $m$-th
root of unity} is any generator, such as $\exp(2\pi\iu/m),$ for the
multiplicative group of the roots of $x^{m}-1\in\mathbb{Q}[x]$.
\begin{cor}
\label{cor:circulant-diagonalization}Let $C=\circulant(c_{0},c_{1},\ldots,c_{m-1})$
be a circulant matrix. Let $\omega$ be a primitive $m$-th root of
unity. Then the matrix
\[
W=[\omega^{jk}/\sqrt{m}]_{j,k=0,1,\ldots,m-1}
\]
is unitary and satisfies
\begin{equation}
W^{*}CW=\diag\left(\sum_{j=0}^{m-1}c_{j},\sum_{j=0}^{m-1}c_{j}\omega^{j},\sum_{j=0}^{m-1}c_{j}\omega^{2j},\ldots,\sum_{j=0}^{m-1}c_{j}\omega^{(m-1)j}\right).\label{eq:C-unitarily-diagonalizable}
\end{equation}
In particular, the eigenvalues of $C,$ counting multiplicities, are
\[
\sum_{j=0}^{m-1}c_{j}\omega^{kj},\qquad\qquad k=0,1,2,\ldots,m-1.
\]
\end{cor}

\begin{proof}
For $k,k'=0,1,\ldots,m-1$, we have
\begin{align*}
\sum_{j=0}^{m-1}\frac{\omega^{jk}}{\sqrt{m}}\cdot\frac{\overline{\omega^{jk'}}}{\sqrt{m}} & =\frac{1}{m}\sum_{j=0}^{m-1}\omega^{j(k-k')}\\
 & =\begin{cases}
1 & \text{if \ensuremath{k=k',}}\\
0 & \text{otherwise,}
\end{cases}
\end{align*}
where the second step is valid because $\omega$ is primitive and
in particular $\omega^{k}\ne\omega^{k'}$. We conclude that
\begin{equation}
WW^{*}=W^{*}W=I.\label{eq:W-unitary}
\end{equation}
Fact~\ref{fact:circulant-eigen} implies that
\[
CW=W\diag\left(\sum_{j=0}^{m-1}c_{j},\sum_{j=0}^{m-1}c_{j}\omega^{j},\sum_{j=0}^{m-1}c_{j}\omega^{2j},\ldots,\sum_{j=0}^{m-1}c_{j}\omega^{(m-1)j}\right),
\]
which in light of~(\ref{eq:W-unitary}) is equivalent to~(\ref{eq:C-unitarily-diagonalizable}).
\end{proof}

\subsection{Polynomial approximation}

Recall that the \emph{total degree} of a multivariate real polynomial
$p\colon\Re^{n}\to\Re$, denoted $\deg p,$ is the largest degree
of any monomial of $p.$ We use the terms ``degree'' and ``total
degree'' interchangeably in this paper. Let $f\colon\Xcal\to\Re$
be a given function with domain $\Xcal\subseteq\Re^{n}.$ For any
$d\geq0,$ define 
\[
E(f,d)=\inf_{p}\|f-p\|_{\infty},
\]
where the infimum is over real polynomials $p$ of degree at most
$d.$ In words, $E(f,d)$ is the least error in a pointwise approximation
of $f$ by a polynomial of degree no greater than $d.$ The \emph{$\epsilon$-approximate
degree of $f$} is the minimum degree of a polynomial $p$ that approximates
$f$ pointwise within $\epsilon$:
\[
\|f-p\|_{\infty}\leq\epsilon.
\]
 In this overview, we focus on the polynomial approximation of the
sign function. We start with an elementary construction of an approximant
due to Buhrman et al.~\cite{BNRW05robust}.
\begin{fact}[Buhrman et al.]
\label{fact:polynomial-approx-SGN-upper}For any $N>1$ and $0<\epsilon<1,$
the sign function can be approximated on $[-N,-1]\cup[1,N]$ pointwise
to within $\epsilon$ by a polynomial of degree 
\[
O\left(N^{2}\log\frac{2}{\epsilon}\right).
\]
\end{fact}

\noindent The degree upper bound in Fact~\ref{fact:polynomial-approx-SGN-upper}
is not tight. Indeed, a quadratically stronger bound of $O(N\log(2/\epsilon))$
follows in a straightforward manner from Jackson's theorem in approximation
theory~\cite[Theorem~1.4]{rivlin-book}.  Our applications do not
benefit from this improvement, however, and we opt for the construction
of Buhrman et al.~~\cite{BNRW05robust} because of its striking
simplicity. For the reader's convenience, we provide their short proof
below.
\begin{proof}[Proof \emph{(adapted from Buhrman et al.)}]
 For a positive integer $d,$ consider the degree-$d$ univariate
polynomial 
\begin{align*}
B_{d}(t)=\sum_{i=\lceil d/2\rceil}^{d}{d \choose i}t^{i}(1-t)^{d-i}.
\end{align*}
In words, $B_{d}(t)$ is the probability of observing at least as
many heads as tails in a sequence of $d$ independent coin flips,
each coming up heads with probability $t.$ By Hoeffding's inequality
(Fact~\ref{fact:hoeffding}) for sufficiently large $d=O(N^{2}\log(2/\epsilon)),$
the polynomial $B_{d}$ sends $[0,\frac{1}{2}-\frac{1}{2N}]\to[0,\frac{\epsilon}{2}]$
and similarly $[\frac{1}{2}+\frac{1}{2N},1]\to[1-\frac{\epsilon}{2},1].$
As a result, the shifted and scaled polynomial $2B_{d}\left(\frac{1}{2N}\cdot t+\frac{1}{2}\right)-1$
approximates the sign function pointwise on $[-N,-1]\cup[1,N]$ within
$\epsilon.$
\end{proof}
On the lower bounds side, Paturi proved that low-degree polynomials
cannot approximate the majority function well. He in fact obtained
analogous results for all symmetric functions, but the special case
of majority will be sufficient for our purposes.
\begin{thm}[Paturi]
\label{thm:paturi-maj}For some constant $c>0$ and all integers
$n\geq1,$
\[
E(\MAJ_{n},cn)\geq\frac{1}{3}.
\]
\end{thm}

\noindent The constant $1/3$ in Paturi's theorem can be replaced
by any other in $(0,1).$ His result is of interest to us because
along with Fact~\ref{fact:polynomial-approx-SGN-upper}, it implies
a lower bound for the approximation of the sign function on the discrete
set of points $\{\pm1,\pm2,\dots,\pm N\}$ for any $N.$
\begin{prop}
\label{prop:polynomial-approx-SGN-lower}For all positive integers
$N$ and $d,$
\[
E(\sign|_{\{\pm1,\pm2,\ldots,\pm N\}},d)\geq1-O\left(\frac{d}{N}\right)^{1/2}.
\]
\end{prop}

\begin{proof}
Abbreviate $\epsilon=E(\sign|_{\{\pm1,\pm2,\ldots,\pm N\}},d)$ and
fix a polynomial $p$ of degree at most $d$ that approximates the
sign function on $\{\pm1,\pm2,\ldots,\pm N\}$ within $\epsilon$.
Fact~\ref{fact:polynomial-approx-SGN-upper} gives a polynomial $s$
of degree $O(1/(1-\epsilon)^{2})$ that sends $[-1-\epsilon,-1+\epsilon]\to[-4/3,-2/3]$
and $[1-\epsilon,1+\epsilon]\to[2/3,4/3].$ Then the composition of
these two approximants obeys
\[
\max_{t=\pm1,\pm2,\ldots,\pm N}|\sign(t)-s(p(t))|\leq\frac{1}{3}.
\]
This in turn gives an approximant for the majority function on $n=\lfloor(N-1)/2\rfloor$
bits:
\begin{align*}
 & \max_{x\in\zoon}\left|\MAJ_{n}(x)-s\left(p\left(2\sum_{j=1}^{n}(-1)^{x_{j}}+1\right)\right)\right|\\
 & \qquad=\max_{x\in\zoon}\left|\sign\left(2\sum_{j=1}^{n}(-1)^{x_{j}}+1\right)-s\left(p\left(2\sum_{j=1}^{n}(-1)^{x_{j}}+1\right)\right)\right|\\
 & \qquad\leq\max_{t=\pm1,\pm2,\ldots,\pm N}|\sign(t)-s(p(t))|\\
 & \qquad\leq\frac{1}{3}.
\end{align*}
In view of Paturi's lower bound for the majority function (Theorem~\ref{thm:paturi-maj}),
the approximant $s(p(2\sum(-1)^{x_{j}}+1))$ must have degree $\Omega(n)=\Omega(N).$
But this composition is a polynomial in $x\in\zoon$ of degree $\deg s\cdot\deg p=O(d/(1-\epsilon)^{2}).$
We conclude that $d/(1-\epsilon)^{2}\geq\Omega(N),$ whence $\epsilon\geq1-O(d/N)^{1/2}.$
\end{proof}

\subsection{Rational approximation}

Consider a rational function $r(x)=p(x)/q(x),$ where $p$ and $q$
are polynomials on $\Re^{n}.$ We refer to the degrees of $p$ and
$q$ as the \emph{numerator degree} and \emph{denominator degree},
respectively, of $r$. The \emph{degree} of $r$ is, then, the maximum
of the numerator and denominator degrees. For a function $f\colon X\to\Re$
with domain $X\subseteq\Re^{n},$ we define 
\begin{align}
R(f,d_{0},d_{1})\,=\,\inf_{p,q}\,\sup_{x\in X}\left\lvert f(x)-\frac{p(x)}{q(x)}\right\rvert ,\label{eq:R-d0-d1-defined}
\end{align}
where the infimum is over multivariate polynomials $p$ and $q$ of
degree at most $d_{0}$ and $d_{1}$, respectively, such that $q$
does not vanish on $X.$ In words, $R(f,d_{0},d_{1})$ is the least
error in an approximation of $f$ by a multivariate rational function
with numerator degree and denominator degree at most $d_{0}$ and
$d_{1},$ respectively. We will be mostly working with $R(f,d_{0},d_{1})$
in the regimes $d_{0}=d_{1}$ and $d_{0}\gg d_{1}$. In the former
regime, we use the shorthand
\[
R(f,d)=R(f,d,d).
\]
As a limiting case of the latter regime, we have 
\[
E(f,d)=R(f,d,0).
\]

The study of the rational approximation of the sign function dates
back to the seminal work by Zolotarev~\cite{zolotarev1877rational}
in the 1870s. The problem was revisited almost a century later by
Newman~\cite{newman64rational}, who proved the following result.
\begin{fact}[Newman]
\label{fact:newman}For any $N>1$ and any integer $d\geq1,$ 
\[
R(\sign|_{[-N,-1]\cup[1,N]},d)\leq1-\frac{1}{N^{1/d}}.
\]
\end{fact}

\noindent For a recent exposition of Newman's construction, we refer
the reader to~\cite[Theorem~2.4]{sherstov09hshs}. As an important
special case, Newman's work gives upper bounds for the rational approximation
of the sign function on the discrete set $\{\pm1,\pm2,\ldots,\pm N\}.$
Newman's upper bounds were sharpened and complemented with matching
lower bounds in~\cite[Eq.~(2.2) and Theorem~5.1]{sherstov09hshs},
to the following effect.
\begin{thm}[Sherstov]
\label{thm:rational-approx-SGN}For any positive integers $N$ and
$d,$ 
\[
R(\sign|_{\{\pm1,\pm2,\ldots,\pm N\}},d)=\begin{cases}
1-N^{-\Theta(1/d)} & \text{if }1\leq d\leq\log N,\\
2^{-\Theta(d/\log(N/d))} & \text{if }\log N<d<N/2.
\end{cases}
\]
\end{thm}

\noindent Among other things, Theorem~\ref{thm:rational-approx-SGN}
implies the following result on the rational approximation of the
majority function~\cite[Eq.~(2.2) and Theorems~5.1,~5.9]{sherstov09hshs}.
\begin{thm}[Sherstov]
\label{thm:rational-approx-MAJ}For any positive integers $n$ and
$d,$
\[
R(\MAJ_{n},d)=\begin{cases}
1-n^{-\Theta(1/d)} & \text{if }\ensuremath{1\leq d\leq\log n,}\\
2^{-\Theta(d/\log(n/d))} & \text{if }\ensuremath{\log n\leq d<\lfloor n/4\rfloor.}
\end{cases}
\]
\end{thm}

\subsection{Sign-representation}

Let $f\colon X\to\moo$ be a given function, where $X\subset\Re^{n}$
is finite. The \emph{threshold degree} of $f,$ denoted $\degthr(f),$
is the least degree of a polynomial $p(x)$ such that $f(x)\equiv\sign p(x).$
 For functions $f\colon X\to\moo$ and $g\colon Y\to\moo,$ we let
the symbol $f\wedge g$ stand for the function $X\times Y\to\moo$
given by $(f\wedge g)(x,y)=f(x)\wedge g(y).$ Note that in this notation,
$f$ and $f\wedge f$ are completely different functions, the former
having domain $X$ and the latter $X\times X.$ The following ingenious
observation, due to Beigel et al.~\cite{beigel91rational}, relates
the notions of sign-representation and rational approximation for
conjunctions of Boolean functions.
\begin{thm}[Beigel et al.]
\label{thm:beigel-degthr-rational} Let $f\colon X\to\moo$ and $g\colon Y\to\moo$
be given functions, where $X,Y\subseteq\Re^{n}.$ Let $d$ be any
integer with 
\[
R(f,d)+R(g,d)<1.
\]
 Then 
\begin{align*}
\degthr(f\wedge g)\leq4d.
\end{align*}
\end{thm}

\begin{proof}[Proof \emph{(adapted from Beigel et al.).}]
 Fix arbitrary rational functions $p_{1}(x)/q_{1}(x)$ and $p_{2}(y)/q_{2}(y)$
of degree at most $d$ such that 
\begin{align*}
\sup_{X}\left|f(x)-\frac{p_{1}(x)}{q_{1}(x)}\right|+\sup_{Y}\left|g(y)-\frac{p_{2}(y)}{q_{2}(y)}\right|<1.
\end{align*}
Then 
\begin{align*}
f(x)\wedge g(y) & \equiv\sign(1+f(x)+g(y))\\
 & \equiv\sign\left(1+\frac{p_{1}(x)}{q_{1}(x)}+\frac{p_{2}(y)}{q_{2}(y)}\right).
\end{align*}
Multiplying through by the positive quantity $q_{1}(x)^{2}q_{2}(y)^{2}$
gives the desired sign-representing polynomial: $f(x)\wedge g(y)\equiv\sign\{q_{1}(x)^{2}q_{2}(y)^{2}+p_{1}(x)q_{1}(x)q_{2}(y)^{2}+p_{2}(y)q_{2}(y)q_{1}(x)^{2}\}.$
\end{proof}
The construction of Theorem~\ref{thm:beigel-degthr-rational} is
somewhat ad hoc, and there is no particular reason to believe that
it gives a sign-representing polynomial of asymptotically optimal
degree. Remarkably, it does. The following converse to the theorem
of Beigel et al.~was established in~\cite[Theorem~3.16]{sherstov09hshs}.
\begin{thm}[Sherstov]
\label{thm:sherstov-degthr-R} Let $f\colon X\to\moo$ and $g\colon Y\to\moo$
be given functions, where $X,Y\subset\Re^{n}$ are arbitrary finite
sets. Assume that $f$ and $g$ are not identically false. Let $d=\degthr(f\wedge g).$
Then 
\begin{align*}
R(f,4d)+R(g,2d)<1.
\end{align*}
\end{thm}

\subsection{Symmetrization}

Let $S_{n}$ denote the symmetric group on $n$ elements. For $\sigma\in S_{n}$
and $x\in\zoon$, we denote $\sigma x=(x_{\sigma(1)},\ldots,x_{\sigma(n)})\in\zoon.$
For $x\in\zoon,$ we define $|x|=x_{1}+x_{2}+\cdots+x_{n}.$ A function
$\phi\colon\zoon\to\Re$ is called \emph{symmetric} if $\phi(x)=\phi(\sigma x)$
for every $x\in\zoon$ and every $\sigma\in S_{n}.$ Equivalently,
$\phi$ is symmetric if $\phi(x)$ is uniquely determined by $|x|.$
Symmetric functions on $\zoon$ are intimately related to univariate
polynomials, as borne out by Minsky and Papert's \emph{symmetrization
argument}~\cite{minsky88perceptrons}.
\begin{prop}[Minsky and Papert]
\label{prop:minsky-papert}Let $p\colon\zoon\to\Re$ be a polynomial
of degree $d.$ Then there is a univariate polynomial $p^{*}$ of
degree at most $d$ such that for all $x\in\zoon,$
\begin{align*}
\Exp_{\sigma\in S_{n}}p(\sigma x)=p^{*}(|x|).
\end{align*}
\end{prop}

\noindent Minsky and Papert's result generalizes to block-symmetric
functions, as pointed out in~\cite[Proposition~2.3]{RS07dc-dnf}:
\begin{prop}[Razborov and Sherstov]
\label{prop:symmetrization} Let $n_{1},\dots,n_{k}$ be positive
integers. Let $p\colon\zoo^{n_{1}}\times\cdots\times\zoo^{n_{k}}\to\Re$
be a polynomial of degree $d.$ Then there is a polynomial $p^{*}\colon\Re^{k}\to\Re$
of degree at most $d$ such that for all $x_{1}\in\zoo^{n_{1}},\ldots,x_{k}\in\zoo^{n_{k}},$
\begin{align*}
\Exp_{\sigma_{1}\in S_{n_{1}},\dots,\sigma_{k}\in S_{n_{k}}}p(\sigma_{1}x_{1},\dots,\sigma_{k}x_{k}) & =p^{*}(|x_{1}|,\ldots,|x_{k}|).
\end{align*}
\end{prop}

\noindent Proposition~\ref{prop:symmetrization} follows in a straightforward
manner from Minsky and Papert's Proposition~\ref{prop:minsky-papert}
by induction on the number of blocks $k.$

\subsection{\label{subsec:Communication-complexity}Communication complexity}

An excellent reference on communication complexity is the monograph
by Kushilevitz and Nisan~\cite{ccbook}. In this overview, we will
limit ourselves to key definitions and notation. We adopt the \emph{randomized
number-on-the-forehead model}, due to Chandra et al.~\cite{cfl83multiparty}.
The model features $k$ communicating players, tasked with computing
a (possibly partial) Boolean function $F$ on the Cartesian product
$X_{1}\times X_{2}\times\cdots\times X_{k}$ of some finite sets $X_{1},X_{2},\dots,X_{k}$.
A given input $(x_{1},x_{2},\dots,x_{k})\in X_{1}\times X_{2}\times\cdots\times X_{k}$
is distributed among the players by placing $x_{i}$, figuratively
speaking, on the forehead of the $i$-th player (for $i=1,2,\dots,k$).
In other words, the $i$-th player knows the arguments $x_{1},\dots,x_{i-1},x_{i+1},\dots,x_{k}$
but not $x_{i}$. The players communicate by sending broadcast messages,
taking turns according to a protocol agreed upon in advance. Each
of them privately holds an unlimited supply of uniformly random bits,
which he can use along with his available arguments when deciding
what message to send at any given point in the protocol. The protocol's
purpose is to allow accurate computation of $F$ everywhere on the
domain of $F$. An \emph{$\epsilon$-error protocol} for $F$ is one
which, on every input $(x_{1},x_{2},\dots,x_{k})\in\dom F,$ produces
the correct answer $F(x_{1},x_{2},\dots,x_{k})$ with probability
at least $1-\epsilon$. The \emph{cost} of a protocol is the total
bit length of the messages broadcast by all the players in the worst
case.\footnote{~The contribution of a $b$-bit broadcast to the protocol cost is
$b$ rather than $k\cdot b$.} The \emph{$\epsilon$-error randomized communication complexity}
of $F,$ denoted $R_{\epsilon}(F),$ is the least cost of an $\epsilon$-error
randomized protocol for $F$. As a special case of this model for
$k=2,$ one recovers the original two-party model of Yao~\cite{yao79cc}
reviewed in the introduction.

We focus on randomized protocols with probability of error close to
that of random guessing, $1/2$. There are two natural ways to define
the communication complexity of a multiparty problem $F$ in this
setting. The \emph{communication complexity of $F$ with unbounded
error}, introduced by Paturi and Simon~\cite{paturi86cc}, is the
quantity 
\[
\upp(F)=\inf_{0\leq\epsilon<1/2}R_{\epsilon}(F).
\]
The error probability in this formalism is ``unbounded'' in the
sense that it can be arbitrarily close to $1/2$. Babai et al.~\cite{BFS86cc}
proposed an alternate quantity, which includes an additive penalty
term that depends on the error probability: 
\[
\pp(F)=\inf_{0\leq\epsilon<1/2}\left\{ R_{\epsilon}(F)+\log\frac{1}{\frac{1}{2}-\epsilon}\right\} .
\]
We refer to $\pp(F)$ as the \emph{communication complexity of $F$
with weakly unbounded error.} These two complexity measures naturally
give rise to corresponding complexity classes $\UPP_{k}$ and $\PP_{k}$
in multiparty communication complexity~\cite{BFS86cc}, both inspired
by Gill's probabilistic polynomial time for Turing machines~\cite{gill77pp}.
Formally, let $\{F_{n,k}\}_{n=1}^{\infty}$ be a family of $k$-party
communication problems $F_{n,k}\colon(\zoon)^{k}\to\moo$, where $k=k(n)$
is either a constant or a function. Then $\{F_{n,k}\}_{n=1}^{\infty}\in\UPP_{k}$
if and only if $\upp(F_{n,k})\leq\log^{c}n$ for some constant $c$
and all $n\geq c$. Analogously, $\{F_{n,k}\}_{n=1}^{\infty}\in\PP_{k}$
if and only if $\pp(F_{n,k})\leq\log^{c}n$ for some constant $c$
and all $n\geq c$. By definition, 
\[
\PP_{k}\subseteq\UPP_{k}.
\]
It is standard practice to abbreviate $\PP=\PP_{2}$ and $\UPP=\UPP_{2}$.
The following well-known fact, whose proof in the stated generality
is available in~\cite[Fact~2.4]{sherstov16multiparty-pp-upp}, gives
a large class of communication problems that are efficiently computable
with unbounded error.
\begin{fact}
\label{fact:upp-upper-bound} Let $F\colon(\zoon)^{k}\to\moo$ be
a $k$-party communication problem such that $F(x)=\sign p(x)$ for
some polynomial $p$ with $\ell$ monomials. Then 
\[
\upp(F)\leq\lceil\log\ell\rceil+2.
\]
\end{fact}

In the setting of $k=2$ parties, Paturi and Simon~\cite{paturi86cc}
showed that unbounded-error communication complexity has a natural
matrix-analytic characterization. For a matrix $M$ without zero entries,
the \emph{sign-rank} of $M$ is denoted $\srank(M)$ and defined as
the minimum rank of a real matrix $R$ such that $\sign R_{i,j}=\sign M_{i,j}$
for all $i,j.$ In words, the sign-rank of $M$ is the minimum rank
of a real matrix that has the same sign pattern as $M.$ We extend
the notion of sign-rank to communication problems $F\colon X\times Y\to\moo$
by defining $\srank(F)=\srank(M_{F}),$ where $M_{F}=[F(x,y)]_{x\in X,y\in Y}$
is the characteristic matrix of $F.$ The following classic result
due to Paturi and Simon~\cite[Theorem~3]{paturi86cc} relates two-party
unbounded-error communication complexity to sign-rank. 
\begin{thm}[Paturi and Simon]
\label{thm:srank-upp} Let $F\colon X\times Y\to\moo$ be a two-party
communication problem. Then
\[
\log\srank(F)\leq\upp(F)\leq\log\srank(F)+2.
\]
\end{thm}

\subsection{\label{subsec:Discrepancy}Discrepancy}

A $k$-dimensional \emph{cylinder intersection} is a function $\chi\colon X_{1}\times X_{2}\times\cdots\times X_{k}\to\zoo$
of the form 
\begin{align*}
\chi(x_{1},x_{2},\dots,x_{k})=\prod_{i=1}^{k}\chi_{i}(x_{1},\dots,x_{i-1},x_{i+1},\dots,x_{k}),
\end{align*}
where $\chi_{i}\colon X_{1}\times\cdots\times X_{i-1}\times X_{i+1}\times\cdots\times X_{k}\to\zoo$.
In other words, a $k$-dimensional cylinder intersection is the product
of $k$ functions with range $\zoo,$ where the $i$-th function does
not depend on the $i$-th coordinate but may depend arbitrarily on
the other $k-1$ coordinates. Introduced by Babai et al.~\cite{bns92},
cylinder intersections are the fundamental building blocks of communication
protocols and for that reason play a central role in the theory. For
a (possibly partial) Boolean function $F$ on $X_{1}\times X_{2}\times\cdots\times X_{k}$
and a probability distribution $P$ on $X_{1}\times X_{2}\times\cdots\times X_{k},$
the \emph{discrepancy} \emph{of $F$ with respect to $P$} is given
by 
\begin{align*}
\disc_{P}(F)=\sum_{x\notin\dom F}P(x)+\max_{\chi}\left|\sum_{x\in\dom F}F(x)P(x)\chi(x)\right|,
\end{align*}
where the maximum is over cylinder intersections $\chi$. The minimum
discrepancy over all distributions is denoted 
\[
\disc(F)=\min_{P}\disc_{P}(F).
\]
Upper bounds on a function's discrepancy give lower bounds on its
randomized communication complexity, a classic technique known as
the \emph{discrepancy method}~\cite{chor-goldreich88ip,bns92,ccbook}.
\begin{thm}
\label{thm:dm} Let $F$ be a $($possibly partial$)$ Boolean function
on $X_{1}\times X_{2}\times\cdots\times X_{k}$. Then for $0\leq\epsilon\leq1/2,$
\begin{align*}
2^{R_{\epsilon}(F)}\geq\frac{1-2\epsilon}{\disc(F)}.
\end{align*}
\end{thm}

\noindent A proof of Theorem~\ref{thm:dm} in the stated generality
is available in~\cite[Theorem 2.9]{sherstov12mdisj}. Combining this
theorem with the definition of $\pp(F)$ gives the following corollary.
\begin{cor}
\label{cor:dm}Let $F$ be a $($possibly partial$)$ Boolean function
on $X_{1}\times X_{2}\times\cdots\times X_{k}$. Then 
\[
\pp(F)\geq\log\frac{2}{\disc(F)}.
\]
\end{cor}

\subsection{\label{subsec:Pattern-matrix-method}Pattern matrix method}

Theorem~\ref{thm:dm} and Corollary~\ref{cor:dm} highlight the
role of discrepancy in proving lower bounds on randomized communication
complexity. Apart from a few canonical examples~\cite{ccbook}, discrepancy
is a challenging quantity to analyze. The \emph{pattern matrix method}
is a technique that gives tight bounds on the discrepancy and communication
complexity for a large class of communication problems. The technique
was developed in~\cite{sherstov07ac-majmaj,sherstov07quantum} for
two-party communication complexity and has since been generalized
by several authors to the multiparty setting.  We now review the
strongest form~\cite{sherstov12mdisj,sherstov13directional} of the
pattern matrix method, focusing our discussion on discrepancy bounds.

\emph{Set disjointness} is the $k$-party communication problem of
determining whether $k$ given subsets of the universe $\{1,2,\dots,n\}$
have empty intersection, where, as usual, the $i$-th party knows
all the sets except for the $i$-th. Identifying the sets with their
characteristic vectors, set disjointness corresponds to the Boolean
function $\DISJ_{n,k}\colon(\zoon)^{k}\to\moo$ given by 
\begin{align}
\DISJ_{n,k}(x_{1},x_{2},\dots,x_{k})=\neg\bigvee_{i=1}^{n}x_{1,i}\wedge x_{2,i}\wedge\cdots\wedge x_{k,i}\,.\label{eq:disj-def}
\end{align}
The partial function $\UDISJ_{n,k}$ on $(\zoon)^{k}$, called \emph{unique
set disjointness}, is defined as $\DISJ_{n,k}$ with domain restricted
to inputs $x\in(\zoon)^{k}$ such that \emph{$x_{1,i}\wedge x_{2,i}\wedge\cdots\wedge x_{k,i}=1$}
for at most one coordinate $i$. In set-theoretic terms, this restriction
corresponds to requiring that the $k$ sets either have empty intersection
or intersect in a unique element.

The pattern matrix method pertains to the communication complexity
of \emph{composed} communication problems. Specifically, let $G$
be a (possibly partial) Boolean function on $X_{1}\times X_{2}\times\cdots\times X_{k},$
representing a $k$-party communication problem, and let $f\colon\zoon\to\moo$
be given. The coordinatewise composition $f\circ G$ is then a $k$-party
communication problem on $X_{1}^{n}\times X_{2}^{n}\times\cdots\times X_{k}^{n}$.
We are now in a position to state the pattern matrix method for discrepancy
bounds~\cite[Theorem~5.7]{sherstov12mdisj}.
\begin{thm}[Sherstov]
\label{thm:pm-large-adeg}For every Boolean function $f\colon\zoon\to\{-1,+1\},$
all positive integers $m$ and $k,$ and all reals $0<\gamma<1,$
\[
\disc(f\circ\UDISJ_{m,k})\leq\left(\frac{\e\cdot2^{k}n}{\deg_{1-\gamma}(f)\sqrt{m}}\right)^{\deg_{1-\gamma}(f)}+\gamma\,.
\]
\end{thm}

\noindent This theorem makes it possible to prove communication lower
bounds by leveraging the existing literature on polynomial approximation.
In follow-up work, the author improved Theorem~\ref{thm:pm-large-adeg}
to an essentially tight upper bound~\cite[Theorem~5.7]{sherstov13directional}.
However, we will not need this sharper version.

\section{\label{sec:Discrepancy}Discrepancy of integer sets}

Let $m\geq2$ be an integer modulus. Key to our work is the notion
of \emph{$m$-discrepancy}, which quantifies the pseudorandomness
or aperiodicity of any given multiset of integers modulo $m.$ The
$m$-discrepancy of a nonempty multiset $Z=\{z_{1},z_{2},\ldots,z_{n}\}$
of arbitrary integers is defined as
\[
\disc(Z,m)=\max_{k=1,2,\ldots,m-1}\left|\frac{1}{n}\sum_{j=1}^{n}\omega^{kz_{j}}\right|,
\]
where $\omega$ is a primitive $m$-th root of unity; the right-hand
side is obviously the same for any such $\omega$. By way of terminology,
we emphasize that the notion of $m$-discrepancy just defined is unrelated
to the notion of \emph{discrepancy} from Section~\ref{subsec:Discrepancy}.
As a matter of convenience, we define
\begin{equation}
\disc(\varnothing,m)=0.\label{eq:disc-empty}
\end{equation}
The notion of $m$-discrepancy has a long history in combinatorics
and theoretical computer science, e.g.,~\cite{gks86k-page-graphs-and-nondeterministic-TMs,ruzsa87essential-components,AIKPS90aperiodic-set,katz89character-sums,rsw93sets-uniform-in-arithmetic-progressions,alon-roichman94rando-cayley-graphs-and-expanders}.
The $m$-discrepancy of an integer multiset $Z$ has a natural interpretation
in terms of the discrete Fourier transform on $\ZZ_{m}.$ Specifically,
consider the \emph{frequency vector $(f_{0},f_{1},\ldots,f_{m-1})$}
of $Z$, where $f_{j}$ is the total number of element occurrences
in $Z$ that are congruent to $j$ modulo $m.$ Applying the discrete
Fourier transform to $(f_{j})_{j=0}^{m-1}$ produces the sequence
$(\sum_{j=0}^{m-1}f_{j}\exp(-2\pi\iu kj/m))_{k=0}^{m-1}=(\sum_{j=1}^{n}\exp(-2\pi\iu kz_{j}/m))_{k=0}^{m-1},$
which is a permutation of $(n,\sum_{j=1}^{n}\omega^{z_{j}},\ldots,\sum_{j=1}^{n}\omega^{(m-1)z_{j}}).$
Summarizing, the $m$-discrepancy of $Z$ coincides up to a normalizing
factor with the largest absolute value of a nonconstant Fourier coefficient
of the frequency vector of $Z.$

\subsection{Basic properties}

We collect a few elementary properties of $m$-discrepancy. To start
with, we quantify the ``continuity'' of $\disc(Z,m)$ in the first
argument. By way of notation, we remind the reader that the cardinality
$|Z|$ of a multiset $Z$ is found by summing, for each distinct element
$z\in Z,$ the number of times $z$ occurs in~$Z.$
\begin{prop}
\label{prop:norm-Z-subset-Z}Fix a natural number $m\geq2.$ Then
for any nonempty finite multisets $Z,Z'$ of integers with $Z'\subseteq Z,$
\begin{equation}
1+\disc(Z',m)\leq(1+\disc(Z,m))\cdot\frac{|Z|}{|Z'|}.\label{eq:Z'-norm}
\end{equation}
\end{prop}

\begin{proof}
Abbreviate $n=|Z|$ and $n'=|Z'|,$ and fix an enumeration $z_{1},z_{2},\ldots,z_{n}$
of the elements of $Z$ such that $Z'=\{z_{1},z_{2},\ldots,z_{n'}\}.$
Then for a primitive $m$-th root of unity $\omega,$
\begin{align*}
n\disc(Z,m) & =\max_{k=1,2,\ldots,m-1}\left|\sum_{j=1}^{n}\omega^{kz_{j}}\right|\\
 & \geq\max_{k=1,2,\ldots,m-1}\left\{ \left|\sum_{j=1}^{n'}\omega^{kz_{j}}\right|-\sum_{j=n'+1}^{n}\left|\omega^{kz_{j}}\right|\right\} \\
 & =\max_{k=1,2,\ldots,m-1}\left|\sum_{j=1}^{n'}\omega^{kz_{j}}\right|-(n-n')\\
 & =n'\disc(Z',m)-(n-n'),
\end{align*}
which directly implies~(\ref{eq:Z'-norm}).
\end{proof}
The $m$-discrepancy of $Z$ is invariant under a variety of operations
on $Z$, such as  shifting the elements of $Z$ by any given integer
or multiplying the elements of $Z$ by an integer relatively prime
to $m.$ For our purposes, the following observation will be sufficient.
\begin{prop}
\label{prop:normZ-norm-minusZ}Fix a natural number $m\geq2$ and
a nonempty finite multiset $Z$ of integers. Then 
\[
\disc(-Z,m)=\disc(Z,m).
\]
\end{prop}

\begin{proof}
The claim is immediate from the definition of $m$-discrepancy because
$\omega$ is a primitive $m$-th root of unity if and only if $\omega^{-1}$
is.
\end{proof}

\subsection{Existential bounds}

Since the $m$-discrepancy of a multiset remains unchanged when one
reduces its elements modulo $m,$ we can focus without loss of generality
on multisets with elements in $\{0,1,2,\ldots,m-1\}.$ The identity
$1+\omega+\omega^{2}+\cdots+\omega^{m-1}=0$ for any $m$-th root
of unity $\omega\ne1$ implies that $Z=\{0,1,2,\ldots,m-1\}$ achieves
the smallest possible $m$-discrepancy: $\disc(Z,m)=0.$ The problem
of constructing \emph{sparse} nonempty multisets with small discrepancy
has seen considerable work. Their existence is straightforward to
verify, as follows.
\begin{fact}
\label{fact:small-fourier-set-existence}Fix $0<\epsilon<1$ and an
integer $m\geq2$. Let $Z$ be a random multiset of size $n$ whose
elements are chosen independently and uniformly at random from $\{0,1,2,\ldots,m-1\}.$
Then
\[
\Prob\left[\disc(Z,m)\geq\epsilon\right]\leq4m\exp\left(-\frac{n\epsilon^{2}}{8}\right).
\]
\end{fact}

\noindent Fact~\ref{fact:small-fourier-set-existence} has been proved
in one form or another by many authors, e.g.,~\cite{gks86k-page-graphs-and-nondeterministic-TMs,ruzsa87essential-components,alon-roichman94rando-cayley-graphs-and-expanders}.
For the reader's convenience, we include a short proof below.
\begin{proof}[Proof of Fact~\emph{\ref{fact:small-fourier-set-existence}}.]
 Let $Z_{1},Z_{2},\ldots,Z_{n}$ be independent random variables,
each distributed uniformly in $\{0,1,2,\ldots,m-1\}.$ For any $m$-th
root of unity $\omega\ne1,$ we have $|\omega^{Z_{j}}|=1$ and $\Exp\omega^{Z_{j}}=0$
for $j=1,2,\ldots,n.$ Hence, $\operatorname{Re}(\omega^{Z_{1}}),\operatorname{Re}(\omega^{Z_{2}}),\ldots,\operatorname{Re}(\omega^{Z_{n}})$
are independent random variables with range in $[-1,1]$ and expectation~$0,$
and likewise for $\imagpart(\omega^{Z_{1}}),\imagpart(\omega^{Z_{2}}),\ldots,\imagpart(\omega^{Z_{n}})$.
As a result, 
\begin{align*}
\Prob\left[\left|\frac{1}{n}\sum_{j=1}^{n}\omega^{Z_{j}}\right|\geq\epsilon\right] & \leq\Prob\left[\left|\realpart\left(\frac{1}{n}\sum_{j=1}^{n}\omega^{Z_{j}}\right)\right|\geq\frac{\epsilon}{2}\right]\\
 & \qquad\qquad+\Prob\left[\left|\imagpart\left(\frac{1}{n}\sum_{j=1}^{n}\omega^{Z_{j}}\right)\right|\geq\frac{\epsilon}{2}\right]\\
 & \leq4\exp\left(-\frac{n\epsilon^{2}}{8}\right),
\end{align*}
where the second step uses Hoeffding's inequality (Fact~\ref{fact:hoeffding}).
Applying the union bound across all $m$-th roots of unity $\omega\ne1,$
we conclude that the probability that $\disc(\{Z_{1},Z_{2},\ldots,Z_{n}\},m)\geq\epsilon$
is at most $4(m-1)\exp(-n\epsilon^{2}/8)$.
\end{proof}
\noindent In some applications, one is restricted to working with
\emph{subsets} of $\{0,1,2,\ldots,m-1\}$ as opposed to arbitrary
\emph{multisets} with possibly repeated elements. We record a version
of Fact~\ref{fact:small-fourier-set-existence} for this setting.
\begin{cor}
\label{cor:small-fourier-set-existence}Fix $0<\epsilon<1$ and an
integer $m\geq2$. Let $Z$ be a random multiset of size $n\leq m$
whose elements are chosen independently and uniformly at random from
$\{0,1,2,\ldots,m-1\}.$ Then with probability at least
\begin{equation}
\left(1-\frac{n}{m}\right)^{n}-4m\exp\left(-\frac{n\epsilon^{2}}{8}\right),\label{eq:small-fourier-set-distinct-elements}
\end{equation}
the elements of $Z$ are nonzero and pairwise distinct, and obey $\disc(Z,m)\leq\epsilon.$
\end{cor}

\begin{proof}
The probability that $Z$ does not contain $0$ or repeated elements
is easily seen to be $\prod_{i=1}^{n}\frac{m-i}{m}\geq(1-\frac{n}{m})^{n}.$
As a result, the claim follows from Fact~\ref{fact:small-fourier-set-existence}.
\end{proof}
In all of our applications, the error parameter $\epsilon>0$ will
be a small constant. In this regime, Corollary~\ref{cor:small-fourier-set-existence}
guarantees the existence of a set $Z\subseteq\{1,2,\ldots,m-1\}$
with $m$-discrepancy at most $\epsilon$ and cardinality $O(\log m),$
an exponential improvement in sparsity compared to the trivial set
$\{0,1,2,\ldots,m-1\}.$ No further improvement is possible: it is
well known that any nonempty multiset with $m$-discrepancy bounded
away from $1$ has cardinality $\Omega(\log m)$. This classical lower
bound has a remarkable variety of proofs, e.g., using random walks~\cite{alon-roichman94rando-cayley-graphs-and-expanders},
sphere packing arguments~\cite{fmt06spectral-estimates-for-cayley-graphs},
and diophantine approximation~\cite{lns11nonexistence-circular-expander}.
We include here a particularly simple and self-contained proof, adapted
from Leung et al.~\cite{lns11nonexistence-circular-expander}. Unlike
all other technical statements in this paper, Fact~\ref{fact:disc-lower}
is not used in the proof of our main result and is provided solely
for completeness.
\begin{fact}[Leung et al.]
\label{fact:disc-lower}Fix a natural number $m\geq2.$ Let $Z=\{z_{1},z_{2},\ldots,z_{n}\}$
be a multiset of integers. Then
\[
\disc(Z,m)\geq1-\frac{2\pi}{\lfloor(m-1)^{1/n}\rfloor}.
\]
\end{fact}

\begin{proof}[Proof \emph{(adapted from~\cite{lns11nonexistence-circular-expander}).}]
 The proof is based on a classic technique from simultaneous diophantine
approximation. For a nonnegative real number $x,$ let $\fr(x)$ denote
the fractional part of $x.$ Abbreviate $q=\lfloor(m-1)^{1/n}\rfloor$
and consider the $q$ intervals
\begin{equation}
\left[0,\frac{1}{q}\right),\left[\frac{1}{q},\frac{2}{q}\right),\left[\frac{2}{q},\frac{3}{q}\right),\ldots,\left[\frac{q-1}{q},1\right).\label{eq:boxes}
\end{equation}
By the pigeonhole principle, there must be a pair of distinct integers
$k',k''\in\{0,1,2,\ldots,q^{n}\}$ such that
\[
\fr\left(\frac{z_{1}k'}{m}\right),\fr\left(\frac{z_{2}k'}{m}\right),\ldots,\fr\left(\frac{z_{n}k'}{m}\right)
\]
are in the same intervals of~(\ref{eq:boxes}) as
\[
\fr\left(\frac{z_{1}k''}{m}\right),\fr\left(\frac{z_{2}k''}{m}\right),\ldots,\fr\left(\frac{z_{n}k''}{m}\right),
\]
respectively. Without loss of generality, $k'>k''.$ Then the integer
$k=k'-k''$ obeys
\begin{align}
 & k\in\{1,2,\ldots,m-1\},\label{eq:k-legal}\\
 & \left|\frac{z_{j}k}{m}-u_{j}\right|\leq\frac{1}{q},\qquad\qquad\qquad j=1,2,\ldots,n\label{eq:all-near-int}
\end{align}
for some $u_{1},u_{2},\ldots,u_{n}\in\ZZ.$ Now
\begin{align*}
\disc(Z,m) & \geq\frac{1}{n}\left|\sum_{j=1}^{n}\exp\left(2\pi\iu\cdot\frac{kz_{j}}{m}\right)\right|\\
 & \geq1-\frac{1}{n}\sum_{j=1}^{n}\left|1-\exp\left(2\pi\iu\cdot\frac{kz_{j}}{m}\right)\right|\\
 & =1-\frac{1}{n}\sum_{j=1}^{n}\left|1-\exp\left(2\pi\iu\cdot\left(\frac{kz_{j}}{m}-u_{j}\right)\right)\right|\\
 & \geq1-\frac{1}{n}\sum_{j=1}^{n}2\pi\left|\frac{kz_{j}}{m}-u_{j}\right|\\
 & \geq1-\frac{2\pi}{q},
\end{align*}
where the first step uses the definition of $m$-discrepancy; the
second step applies the triangle inequality; the third step is valid
by periodicity; the fourth step uses the bound $|1-\exp(2\pi x\iu)|=\sqrt{2-2\cos(2\pi x)}\leq2\pi|x|$
for all real $x$; and the final step is immediate from~(\ref{eq:all-near-int}).
\end{proof}

\subsection{\label{subsec:An-explicit-construction}An explicit construction}

We now turn to the problem of efficiently constructing sparse sets
with small $m$-discrepancy. Two such constructions are known to date,
due to Ajtai et al.~\cite{AIKPS90aperiodic-set} and Katz~\cite{katz89character-sums}.
The approach of Ajtai et al.~is elementary except for an appeal to
the prime number theorem. Katz's construction, on the other hand,
relies on deep results in number theory. Neither work appears to directly
imply the kind of optimal de-randomization that we require, namely,
an algorithm that runs in time polynomial in $\log m$ and produces
a multiset of cardinality $O(\log m)$ with $m$-discrepancy bounded
away from~1. We obtain such an algorithm by adapting the approach
of Ajtai et al.~\cite{AIKPS90aperiodic-set}. The following technical
result plays a central role.
\begin{thm}[cf.~Ajtai et al.]
\label{thm:ajtai-iteration}Fix an integer $R\geq1$ and a real number
$P\geq2$. Let $m$ be an integer with $m\geq P^{2}(R+1).$ Fix a
set $S_{p}\subseteq\{1,2,\ldots,p-1\}$ for each prime $p\in(P/2,P]$
with $p\nmid m,$ such that all $S_{p}$ have the same cardinality.
Consider the multiset
\begin{multline*}
S=\{(r+s\cdot(p^{-1})_{m})\bmod m:\\
\qquad r=1,\ldots,R;\quad p\in(P/2,P]\text{ prime with }p\nmid m;\quad s\in S_{p}\}.
\end{multline*}
Then the elements of $S$ are pairwise distinct and nonzero. Moreover,
\[
\disc(S,m)\leq\frac{c}{\sqrt{R}}+\frac{c\log m}{\log\log m}\cdot\frac{\log P}{P}+\max_{p}\{\disc(S_{p},p)\}
\]
for some $($explicitly given$)$ constant $c\geq1$ independent of
$P,R,m.$
\end{thm}

\noindent Ajtai et al.~\cite{AIKPS90aperiodic-set} proved a special
case of Theorem~\ref{thm:ajtai-iteration} for $m$ prime, but their
argument readily generalizes to arbitrary moduli $m$ as just stated.
For the reader's convenience, we provide a complete proof of Theorem~\ref{thm:ajtai-iteration}
in Appendix~\ref{sec:ajtai}. The theorem's purpose is to reduce
the construction of a sparse set with small $m$-discrepancy to the
construction of sparse sets with small $p$-discrepancy, for primes
$p\ll m.$ By applying Theorem~\ref{thm:ajtai-iteration} in a recursive
manner, one reaches smaller and smaller primes. The authors of~\cite{AIKPS90aperiodic-set}~continue
this recursive process until they reach primes $p$ so small that
the trivial construction $\{1,2,3,\ldots,p-1\}$ can be considered
sparse. We proceed differently and terminate the recursion after just
two stages, at which point the input size is small enough for brute
force search based on Corollary~\ref{cor:small-fourier-set-existence}.
The final set that we construct has size logarithmic in $m$ and $m$-discrepancy
a small constant, as opposed to the superlogarithmic size and $o(1)$
discrepancy in the work of Ajtai et al.~\cite{AIKPS90aperiodic-set}.
A detailed exposition of our algorithm follows.
\begin{thm}
\label{thm:explicit-set-small-Fourier-coeffs}Let $0<\epsilon\leq1$
be given. Then there is an algorithm that takes as input an integer
$m\geq2,$ runs in time polynomial in $\log m,$ and outputs a nonempty
set $Z\subseteq\{0,1,2,\ldots,m-1\}$ with
\begin{align*}
 & \disc(Z,m)\leq\epsilon,\\
 & |Z|\leq C_{\epsilon}\log m,
\end{align*}
where $C_{\epsilon}\geq1$ is a constant. Moreover, the constant $C_{\epsilon}$
and the algorithm are given explicitly.
\end{thm}

\begin{proof}
Set $\delta=\epsilon/(11c),$ where $c\geq1$ is the explicit constant
from Theorem~\ref{thm:ajtai-iteration}. Define
\begin{align*}
P' & =\frac{1}{\delta}\ln\left(\frac{1}{\delta}\ln m\right),\\
P'' & =\frac{1}{\delta}\ln m.
\end{align*}
We may assume that
\begin{align}
 & P'\geq\frac{1}{\delta^{2}},\label{eq:P'-large}\\
 & P'>4\left\lceil \frac{8\ln8P'}{\delta^{2}}\right\rceil ^{2},\label{eq:stage1-prereq}\\
 & P''\geq2P'^{2}\left\lceil \frac{1}{\delta^{2}}+1\right\rceil ,\label{eq:stage2-prereq}\\
 & m\geq P''^{2}\left\lceil \frac{1}{\delta^{2}}+1\right\rceil ,\label{eq:stage3-prereq}\\
 & \pi(P')>\pi\left(\frac{P'}{2}\right),\label{eq:stage2-enough-primes}\\
 & \pi(P'')-\pi\left(\frac{P''}{2}\right)>\nu(m),\label{eq:stage3-enough-primes}
\end{align}
where $\pi$ is the prime counting function and $\nu$ is the number
of distinct prime divisors function. Indeed, if any of (\ref{eq:P'-large})\textendash (\ref{eq:stage3-prereq})
is violated, then by elementary calculus $m$ is bounded in terms
of $1/\delta=O(1)$ and therefore the trivial set $Z=\{0,1,2,\ldots,m-1\}$
satisfies $\disc(Z,m)=0$ and $|Z|=O(1).$ Analogously, the explicit
bounds for $\pi$ and $\nu$ in Facts~\ref{fact:PNT} and~\ref{fact:num-prime-factors}
ensure that (\ref{eq:stage2-enough-primes}) and~(\ref{eq:stage3-enough-primes})
can fail only if $m$ is bounded in terms of $1/\delta=O(1),$ so
that we may again output $Z=\{0,1,2,\ldots,m-1\}$.

Assuming~(\ref{eq:P'-large})\textendash (\ref{eq:stage3-enough-primes}),
our construction of $Z$ has three stages. In the first and second
stages, we construct sparse sets $S_{p}\subseteq\{1,2,\ldots,p-1\}$
with small $p$-discrepancy for all primes $p\in(P'/2,P']$ and $p\in(P''/2,P''],$
respectively. In the final stage, we construct the set $Z$ in the
theorem statement. We ensure that each stage runs in time polynomial
in $\ln m.$\bigskip{}

\emph{Stage 1.} For every prime $p'\in(P'/2,P'],$ Corollary~\ref{cor:small-fourier-set-existence}
along with (\ref{eq:stage1-prereq}) guarantees the existence of a
set $S_{p'}\subseteq\{1,2,\ldots,p'-1\}$ with
\begin{align}
 & |S_{p'}|=\left\lceil \frac{8\ln8P'}{\delta^{2}}\right\rceil , &  & \text{prime }p'\in(P'/2,P'],\label{eq:stage1-size}\\
 & \disc(S_{p'},p')\leq\delta, &  & \text{prime }p'\in(P'/2,P'].\label{eq:stage1-norm}
\end{align}
The primes in $(P'/2,P']$ can be identified by the trivial algorithm
in time polynomial in $P'=O(\ln\ln m).$ For each such prime $p',$
we can find a set $S_{p'}$ with the above properties in time $P'^{O(|S_{p'}|)}=o(\ln m)$
by trying out all candidate sets.\bigskip{}

\emph{Stage 2.} Apply the construction of Theorem~\ref{thm:ajtai-iteration}
with parameters $P=P'$ and $R=\lceil1/\delta^{2}\rceil$ to the sets
constructed in Stage~$1$ to obtain a set $S_{p''}\subseteq\{1,2,\ldots,p''-1\}$
for each prime $p''\in(P''/2,P''].$ This choice of parameters is
legitimate by~(\ref{eq:stage2-prereq}). By~(\ref{eq:stage1-size}),
the new sets have the same cardinality, namely,
\begin{align*}
|S_{p''}| & =R\left\lceil \frac{8\ln8P'}{\delta^{2}}\right\rceil \left(\pi(P')-\pi\left(\frac{P'}{2}\right)\right), &  & \qquad\text{prime }p''\in(P''/2,P''].
\end{align*}
The prime number theorem (Fact~\ref{fact:PNT}) implies that $|S_{p''}|=O(P')=O(\ln\ln m)$.
In view of~(\ref{eq:P'-large}), (\ref{eq:stage1-norm}), and $P''=\exp(\delta P'),$
the new sets have
\begin{align}
\disc(S_{p''},p'') & \leq6c\delta, &  & \text{prime }p''\in(P''/2,P''].\label{eq:stage2-norm}
\end{align}

We now show that Stage~$2$ runs in time polynomial in $\ln m.$
To start with, the primes in $(P''/2,P'']$ can be identified by the
trivial algorithm in time polynomial in $P''=O(\ln m).$ For any such
prime $p'',$ the construction of the corresponding set $S_{p''}$
in Theorem~\ref{thm:ajtai-iteration} amounts to $O(|S_{p''}|)=O(\ln\ln m)$
arithmetic operations in the field $\FF_{p''}$ of size $|\FF_{p''}|=O(\ln m),$
and therefore can be carried out in time polynomial in $\ln\ln m.$

\bigskip{}

\emph{Stage 3.} Apply the construction of Theorem~\ref{thm:ajtai-iteration}
with parameters $P=P''$ and $R=\lceil1/\delta^{2}\rceil$ to the
sets constructed in Stage~$2$ to obtain a set $S_{m}\subseteq\{1,2,\ldots,m-1\}$.
This choice of parameters is legitimate by~(\ref{eq:stage3-prereq}).
This new set has cardinality
\begin{multline*}
|S_{m}|=R^{2}\left\lceil \frac{8\ln8P'}{\delta^{2}}\right\rceil \left(\pi(P')-\pi\left(\frac{P'}{2}\right)\right)\\
\times\left|\left\{ p''\text{ prime}:p''\in\left(\frac{P''}{2},P''\right]\text{ and }p''\nmid m\right\} \right|,
\end{multline*}
which in view of (\ref{eq:stage2-enough-primes}) and~(\ref{eq:stage3-enough-primes})
guarantees that $S_{m}$ is nonempty. Simplifying,
\begin{align*}
|S_{m}| & \leq\left\lceil \frac{1}{\delta^{2}}\right\rceil ^{2}\left\lceil \frac{8\ln8P'}{\delta^{2}}\right\rceil \cdot\pi(P')\cdot\pi(P'')\\
 & =O\left(\ln P'\cdot\frac{P'}{\ln P'}\cdot\frac{P''}{\ln P''}\right)\\
 & =O(\ln m),
\end{align*}
where the second step applies the prime number theorem (Fact~\ref{fact:PNT}).
The multiplicative constant in this asymptotic bound on $|S_{m}|$
can be easily recovered from the explicit bounds in Fact~\ref{fact:PNT}.
Using~(\ref{eq:stage2-prereq}), (\ref{eq:stage2-norm}), and $m=\exp(\delta P''),$
we further obtain
\[
\disc(S_{m},m)\leq11c\delta.
\]
Since $\delta=\epsilon/(11c)$, the set $Z=S_{m}$ satisfies the requirements
of the theorem. Finally, the construction of $S_{m}$ in Stage~3
amounts to $O(|S_{m}|)=O(\ln m)$ arithmetic operations in the ring
$\ZZ_{m}$ and therefore can be carried out in time polynomial in
$\ln m.$
\end{proof}

\section{\label{sec:Univariatization}Univariatization}

Consider a halfspace $h_{n}(x)=\sign(\sum z_{i}x_{i}-\theta)$ in
Boolean variables $x_{1},x_{2},\ldots,x_{n}\in\zoo,$ where the coefficients
can be assumed without loss of generality to be integers. Then the
linear form $\sum z_{i}x_{i}-\theta$ ranges in the discrete set $\{\pm1,\pm2,\ldots,\pm N\}$,
for some integer $N$ proportionate to the magnitude of the coefficients.
As a result, one can approximate $h_{n}$ to any given error $\epsilon$
by approximating the sign function to $\epsilon$ on $\{\pm1,\pm2,\ldots,\pm N\}.$
This approach works for both rational approximation and polynomial
approximation. Needless to say, there is no reason to expect that
the degree of the approximant in this na\"{i}ve construction is anywhere
close to optimal. Perhaps the most dramatic example is the \emph{odd-max-bit
function}, defined by $\OMB_{n}(x)=\sign(1+\sum_{i=1}^{n}(-2)^{i}x_{i})$.
A moment's thought reveals that $\OMB_{n}$ can be approximated to
any given error $\epsilon>0$ by a rational function of degree~$1,$
whereas the na\"{i}ve construction produces an approximant of degree
$\Omega(n).$

Surprisingly, we are able to construct a halfspace $h_{n}(x)=\sign(\sum z_{i}x_{i}-\theta)$
with exponentially large coefficients for which the na\"{i}ve construction
is essentially optimal. Specifically, we show that a rational approximant
for $h_{n}$ with given error and given numerator and denominator
degrees implies an analogous \emph{univariate} rational approximant
for the sign function on $\{\pm1,\pm2,\pm3,\ldots,\pm2^{\Theta(n)}\}.$
As a result, tight lower bounds for the rational and polynomial approximation
of $h_{n}$ follow immediately from the univariate lower bounds for
the sign function. The construction of $h_{n}$, carried out in this
section, is the centerpiece of our paper. The role of $h_{n}$ is
to reduce the multivariate problem taken up in this work to a well-understood
univariate question, whence the title of this section. We have broken
down the proof into four steps, corresponding to subsections~\ref{subsec:Eigenvalue-bounds-for}\textendash \ref{subsec:The-master-theorem}
below.

\subsection{\label{subsec:Eigenvalue-bounds-for}Distribution of a linear form
modulo \emph{m}}

We start by studying the probability distribution of the weighted
sum $z_{1}X_{1}+z_{2}X_{2}+\cdots+z_{n}X_{n}$ modulo $m$, where
$z_{1},z_{2},\ldots,z_{n}$ are given integers and $X_{1},X_{2},\ldots,X_{n}\in\zoo$
are chosen uniformly at random. We will show that the distribution
is close to uniform whenever the multiset $\{z_{1},z_{2},\ldots,z_{n}\}$
has small $m$-discrepancy.  This result uses the following classical
fact on linear forms modulo $m$. 
\begin{fact}[cf.~Gould~\cite{gould72combinatorial-identities}; Thathachar~\cite{Thathachar98BP-hierarchy}]
\label{fact:number-solutions-linear-form}Fix a natural number $m\geq2$
and a multiset $Z=\{z_{1},z_{2},\ldots,z_{n}\}$ of integers. Let
$\omega$ be a primitive $m$-th root of unity. Then
\begin{multline}
\left|\Prob_{X\in\zoon}\left[\sum_{j=1}^{n}z_{j}X_{j}\equiv s\pmod m\right]-\frac{1}{m}\right|\\
\leq\frac{1}{m}\sum_{k=1}^{m-1}\left|\prod_{j=1}^{n}\frac{1+\omega^{kz_{j}}}{2}\right|,\qquad s\in\ZZ.\label{eq:weighted-sums-uniform-1}
\end{multline}
 
\end{fact}

\begin{proof}[{Proof \emph{(adapted from~\cite[Lemma~13]{Thathachar98BP-hierarchy})}}]
The fraction of vectors $X\in\zoon$ that satisfy the equation $\sum_{j=1}^{n}z_{j}X_{j}\equiv s\pmod m$
can be computed directly, as follows: 

\begin{align*}
\Prob_{X\in\zoon}\left[\sum_{j=1}^{n}z_{j}X_{j}\equiv s\pmod m\right] & =\Exp_{X\in\zoon}\;\I\left[\sum_{j=1}^{n}z_{j}X_{j}\equiv s\pmod m\right]\\
 & =\Exp_{X\in\zoon}\;\frac{1}{m}\sum_{k=0}^{m-1}\omega^{k(\sum_{j=1}^{n}z_{j}X_{j}-s)}\\
 & =\Exp_{X\in\zoon}\;\frac{1}{m}\sum_{k=0}^{m-1}\omega^{-ks}\prod_{j=1}^{n}\omega^{kz_{j}X_{j}}\\
 & =\frac{1}{m}\sum_{k=0}^{m-1}\omega^{-ks}\Exp_{X\in\zoon}\prod_{j=1}^{n}\omega^{kz_{j}X_{j}}\\
 & =\frac{1}{m}\sum_{k=0}^{m-1}\omega^{-ks}\prod_{j=1}^{n}\frac{1+\omega^{kz_{j}}}{2}\\
 & =\frac{1}{m}+\frac{1}{m}\sum_{k=1}^{m-1}\omega^{-ks}\prod_{j=1}^{n}\frac{1+\omega^{kz_{j}}}{2}.
\end{align*}
This implies~(\ref{eq:weighted-sums-uniform-1}) because $|\omega^{-ks}|=1$
for all $k,s\in\ZZ$. 
\end{proof}
\noindent In the original version of this manuscript, we proved~(\ref{eq:weighted-sums-uniform-1})
using a different, matrix-analytic argument, which we include as Appendix~\ref{app:number-solutions-linear-form}.
The short and elegant proof above was pointed out to us by T.~S.~Jayram,
who kindly allowed us to include it. 

We now simplify the right-hand side of~(\ref{eq:weighted-sums-uniform-1})
and relate it to $m$-discrepancy.
\begin{lem}
\label{lem:sum-equidistributed-mod-m}Fix a natural number $m\geq2$
and a multiset $Z=\{z_{1},z_{2},\ldots,z_{n}\}$ of integers. Then
for all $s\in\ZZ,$
\[
\left|\Prob_{X\in\zoon}\left[\sum_{j=1}^{n}z_{j}X_{j}\equiv s\pmod m\right]-\frac{1}{m}\right|\leq\left(\frac{1+\disc(Z,m)}{2}\right)^{n/2}.
\]
\end{lem}

\begin{proof}
Let $\omega$ be a primitive $m$-th root of unity. For $k=1,2,\ldots,m-1,$
we have
\begin{align*}
\left|\prod_{j=1}^{n}\frac{1+\omega^{kz_{j}}}{2}\right| & =\left(\prod_{j=1}^{n}\frac{(1+\omega^{kz_{j}})(\overline{1+\omega^{kz_{j}}})}{4}\right)^{1/2}\\
 & =\left(\prod_{j=1}^{n}\frac{1+\operatorname{Re}(\omega^{kz_{j}})}{2}\right)^{1/2}\\
 & \leq\left(\frac{1}{n}\sum_{j=1}^{n}\frac{1+\operatorname{Re}(\omega^{kz_{j}})}{2}\right)^{n/2}\\
 & =\left(\frac{1}{2}+\frac{1}{2}\operatorname{Re}\left(\frac{1}{n}\sum_{j=1}^{n}\omega^{kz_{j}}\right)\right)^{n/2}\\
 & \leq\left(\frac{1}{2}+\frac{1}{2}\left|\frac{1}{n}\sum_{j=1}^{n}\omega^{kz_{j}}\right|\right)^{n/2},
\end{align*}
where the second step uses $|\omega|=1$, and the third step follows
by convexity since $1+\operatorname{Re}(\omega^{kz_{j}})\geq0.$ Maximizing
over $k,$ we arrive at
\begin{align*}
\max_{k=1,2,\ldots,m-1}\left|\prod_{j=1}^{n}\frac{1+\omega^{kz_{j}}}{2}\right| & \leq\left(\frac{1}{2}+\frac{1}{2}\max_{k=1,2,\ldots,m-1}\left|\frac{1}{n}\sum_{j=1}^{n}\omega^{kz_{j}}\right|\right)^{n/2}\\
 & =\left(\frac{1+\disc(Z,m)}{2}\right)^{n/2}.
\end{align*}
In view of Fact~\ref{fact:number-solutions-linear-form}, the proof
is complete.
\end{proof}

\subsection{Fooling distributions}

Let $Z=\{z_{1},z_{2},\ldots,z_{n}\}$ be a multiset with $m$-discrepancy
bounded away from $1.$ Consider the linear map $L\colon\zoon\to\ZZ_{m}$
given by $L(x)=\sum z_{i}x_{i}.$ We have shown that for uniformly
random $X\in\zoon$, the probability distribution of $L(X)$ is exponentially
close to uniform. This implies, for some constant $c>0$, that the
sets $L^{-1}(0),L^{-1}(1),\ldots,L^{-1}(m-1)$ cannot be reliably
distinguished by a real polynomial of degree up to $cn$. More precisely,
the characteristic functions of $L^{-1}(0),L^{-1}(1),\ldots,L^{-1}(m-1)$
have approximately the same Fourier spectrum up to degree $cn$. We
will now substantially strengthen this conclusion by proving that
there are probability distributions $\mu_{0},\mu_{1},\ldots,\mu_{m-1}$,
supported on $L^{-1}(0),L^{-1}(1),\ldots,L^{-1}(m-1)$, respectively,
such that the Fourier spectra of $\mu_{0},\mu_{1},\ldots,\mu_{m-1}$
are \emph{exactly} the same up to degree $cn.$ To use a technical
term, these distributions \emph{fool} any polynomial $p$ of degree
up to $cn$, in that $\Exp_{\mu_{0}}p=\Exp_{\mu_{1}}p=\cdots=\Exp_{\mu_{m-1}}p.$
Our proof relies on the following technical result~\cite[Theorem~4.1]{sherstov09opthshs}.
\begin{thm}[Sherstov]
\label{thm:distribution-by-inversion} Let $f,\chi_{1},\dots,\chi_{k}\colon\Xcal\to\moo$
be given functions on a finite set $\Xcal.$ Suppose that 
\begin{align}
 & \sum_{\substack{i=1}
}^{k}|\langle f,\chi_{i}\rangle_{\Xcal}|<\frac{1}{2},\label{eqn:f-chi-bounded}\\
 & \sum_{\substack{j=1\\
j\ne i
}
}^{k}|\langle\chi_{i},\chi_{j}\rangle_{\Xcal}|\leq\frac{1}{2}, &  & i=1,2,\dots,k.\label{eqn:diag-dominance}\\
\intertext{\text{{\it Then there exists a probability distribution \ensuremath{\mu} on \ensuremath{\Xcal} such that}}} & \Exp_{\mu}\,[f(x)\chi_{i}(x)]=0, &  & i=1,2,\dots,k.\nonumber 
\end{align}
\end{thm}

\noindent By way of notation, we remind the reader that $\langle f,g\rangle_{\Xcal}=\frac{1}{|\Xcal|}\sum_{x\in\Xcal}f(x)g(x)$
for any real-valued functions $f$ and $g$ and a nonempty subset
$\Xcal$ of their domain. In words, Theorem~\ref{thm:distribution-by-inversion}
states that if $\chi_{1},\chi_{2},\dots,\chi_{k}$ each have small
correlation with $f$ and, in addition, have small pairwise correlations,
then a distribution exists with respect to which $f$ is completely
uncorrelated with $\chi_{1},\chi_{2},\dots,\chi_{k}.$ We are now
in a position to prove the existence of the promised fooling distributions.
In the statement that follows, recall that $H(p)=-p\log p-(1-p)\log(1-p)$
is the binary entropy function.
\begin{lem}
\label{lem:fooling-distributions}Fix $\delta\in[0,1/2)$ and a nonempty
multiset $Z=\{z_{1},z_{2},\ldots,z_{n}\}$ of integers. Let $m$ be
an integer with
\begin{equation}
2\leq m\leq\left(\frac{2(1-2\delta)}{1+\disc(Z,m)}\right)^{\left(\frac{1}{2}-\delta\right)n}2^{-H(\delta)n-2}.\label{eq:Z-delta-m}
\end{equation}
Define
\begin{align}
\Xcal_{s} & =\left\{ x\in\zoon:\sum_{j=1}^{n}z_{j}x_{j}\equiv s\pmod m\right\} , &  & s\in\ZZ.\label{eq:Xcal-s-defined}
\end{align}
Then each $\Xcal_{s}$ is nonempty. Moreover, there is a probability
distribution $\mu_{s}$ on $\Xcal_{s}$ $($for each $s)$ such that
\begin{equation}
\Exp_{X\sim\mu_{s}}p(X)=\Exp_{X\sim\mu_{s'}}p(X)\label{eq:fooling-p}
\end{equation}
for all $s,s'\in\ZZ$ and all real polynomials $p\colon\zoon\to\Re$
of degree at most $\delta n.$
\end{lem}

\begin{proof}
For a subset $A\subseteq\{1,2,\ldots,n\}$, define $\chi_{A}\colon\zoon\to\moo$
by $\chi_{A}(x)=(-1)^{\sum_{i\in A}x_{i}}.$ The centerpiece of the
proof is the following claim.
\begin{claim}
\label{claim:chi_A-orthog-to-Xs}For every $s\in\ZZ$ and every nonempty
proper subset $A\subset\{1,2,\ldots,n\},$
\begin{align}
 & \Xcal_{s}\ne\varnothing,\label{eq:Xs-nonempty}\\
 & |\langle\chi_{A},1\rangle_{\Xcal_{s}}|\leq2m\left(\frac{1+\disc(Z,m)}{2}\cdot\frac{n}{n-|A|}\right)^{\frac{n-|A|}{2}}.\label{eq:character-correl}
\end{align}
\end{claim}

We will proceed with the main proof and settle the claim after we
are finished. Fix $s\in\ZZ$ arbitrarily. Let $\Acal$ denote the
family of nonempty subsets of $\{1,2,\ldots,n\}$ of cardinality at
most $\delta n.$ Recall from~(\ref{eqn:entropy-bound}) that
\begin{equation}
|\Acal|\leq2^{H(\delta)n}-1.\label{eq:Acal-bound}
\end{equation}
As a result,
\begin{align}
\sum_{A\in\Acal}|\langle\chi_{A},1\rangle_{\Xcal_{s}}| & \leq|\Acal|\cdot\max_{1\leq|A|\leq\delta n}|\langle\chi_{A},1\rangle_{\Xcal_{s}}|\nonumber \\
 & \leq(2^{H(\delta)n}-1)\cdot2m\max_{1\leq k\leq\delta n}\left(\frac{1+\disc(Z,m)}{2}\cdot\frac{n}{n-k}\right)^{\frac{n-k}{2}}\nonumber \\
 & =(2^{H(\delta)n}-1)\cdot2m\left(\frac{1+\disc(Z,m)}{2(1-\delta)}\right)^{\frac{(1-\delta)n}{2}}\nonumber \\
 & <\frac{1}{2},\label{eq:correl-with-1}
\end{align}
where the second step uses~(\ref{eq:Acal-bound}) and Claim~\ref{claim:chi_A-orthog-to-Xs};
the third step is valid because $1+\disc(Z,m)<2(1-\delta)$ by (\ref{eq:Z-delta-m});
and the final step is immediate from~(\ref{eq:Z-delta-m}). An analogous
calculation shows that for every $A\in\Acal,$
\begin{align}
\sum_{A'\in\Acal\setminus\{A\}}|\langle\chi_{A},\chi_{A'}\rangle_{\Xcal_{s}}| & =\sum_{\substack{\substack{A'\in\Acal}
\\
A'\ne A
}
}|\langle\chi_{A\oplus A'},1\rangle_{\Xcal_{s}}|\nonumber \\
 & \leq(2^{H(\delta)n}-1)\cdot2m\left(\frac{1+\disc(Z,m)}{2(1-2\delta)}\right)^{\frac{(1-2\delta)n}{2}}\nonumber \\
 & <\frac{1}{2},\label{eq:cross-correls}
\end{align}
where the second step follows from~(\ref{eq:Acal-bound}) and Claim~\ref{claim:chi_A-orthog-to-Xs},
and the last step uses~(\ref{eq:Z-delta-m}). 

Recall from Claim~\ref{claim:chi_A-orthog-to-Xs} that each $\Xcal_{s}$
is nonempty. Applying Theorem~\ref{thm:distribution-by-inversion}
with (\ref{eq:correl-with-1}) and (\ref{eq:cross-correls}) to the
functions $\chi_{A}$ $(A\in\Acal)$ and $f=1$, we infer the existence
of a probability distribution $\mu_{s}$ on $\Xcal_{s}$ such that
\begin{align}
\Exp_{X\sim\mu_{s}}\chi_{A}(X) & =0, &  & A\in\Acal.\label{eq:distribution-orthogonalizes}
\end{align}
Now that the probability distributions $\mu_{s}$ have been constructed
for each $s\in\ZZ,$ consider an arbitrary polynomial $p\colon\zoon\to\Re$
of degree at most $\delta n.$ Then $p=\sum_{|A|\leq\delta n}p_{A}\chi_{A}$
for some reals $p_{A}$. As a result,~(\ref{eq:distribution-orthogonalizes})
implies that $\Exp_{\mu_{s}}p=p_{\varnothing}$ for all $s\in\ZZ,$
thereby settling~(\ref{eq:fooling-p}).
\end{proof}
\begin{proof}[Proof of Claim~\emph{\ref{claim:chi_A-orthog-to-Xs}}.]
 By symmetry, we may assume that $A=\{1,2,\ldots,k\}$ for some $0<k<n.$
Let $X=(X_{1},X_{2},\ldots,X_{n})$ be a random variable with uniform
distribution on $\zoon$. Then
\begin{align}
\frac{|\Xcal_{s}|}{2^{n}} & \geq\frac{1}{m}-\left|\frac{|\Xcal_{s}|}{2^{n}}-\frac{1}{m}\right|\nonumber \\
 & =\frac{1}{m}-\left|\Prob_{X}[X\in\Xcal_{s}]-\frac{1}{m}\right|\nonumber \\
 & \geq\frac{1}{m}-\left(\frac{1+\disc(Z,m)}{2}\right)^{n/2}\nonumber \\
 & \geq\frac{1}{2m},\label{eq:unconditioned-set-size}
\end{align}
where the last two steps follow from Lemma~\ref{lem:sum-equidistributed-mod-m}
and~(\ref{eq:Z-delta-m}), respectively. This settles~(\ref{eq:Xs-nonempty}).
Moreover,
\begin{align}
\frac{|\Xcal_{s}|}{2^{n}} & |\langle\chi_{A},1\rangle_{\Xcal_{s}}|\nonumber \\
 & =\left|\Exp_{X}\;\chi_{\{1,2,\ldots,k\}}(X)\cdot\I[X\in\Xcal_{s}]\right|\nonumber \\
 & =\left|\sum_{x\in\zook}\frac{(-1)^{x_{1}+\cdots+x_{k}}}{2^{k}}\Prob[x_{1}\ldots x_{k}X_{k+1}\ldots X_{n}\in\Xcal_{s}]\right|\nonumber \\
 & =\left|\sum_{x\in\zook}\frac{(-1)^{x_{1}+\cdots+x_{k}}}{2^{k}}\left(\Prob[x_{1}\ldots x_{k}X_{k+1}\ldots X_{n}\in\Xcal_{s}]-\frac{1}{m}\right)\right|\nonumber \\
 & \leq\frac{1}{2^{k}}\sum_{x\in\zook}\left|\Prob[x_{1}\ldots x_{k}X_{k+1}\ldots X_{n}\in\Xcal_{s}]-\frac{1}{m}\right|\nonumber \\
 & =\frac{1}{2^{k}}\sum_{x\in\zook}\left|\Prob\left[\sum_{j=k+1}^{n}z_{j}X_{j}\equiv s-\sum_{j=1}^{k}z_{j}x_{j}\pmod m\right]-\frac{1}{m}\right|\nonumber \\
 & \leq\left(\frac{1+\disc(\{z_{k+1},z_{k+2},\ldots,z_{n}\},m)}{2}\right)^{(n-k)/2}\nonumber \\
 & \leq\left(\frac{1+\disc(Z,m)}{2}\cdot\frac{n}{n-k}\right)^{(n-k)/2},\label{eq:correl-intermediate}
\end{align}
where the third step uses $k\geq1$; the next-to-last step is legitimate
by Lemma~\ref{lem:sum-equidistributed-mod-m}; and the last step
applies Proposition~\ref{prop:norm-Z-subset-Z}. Now~(\ref{eq:character-correl})
is immediate from~(\ref{eq:unconditioned-set-size}) and~(\ref{eq:correl-intermediate}).
\end{proof}

\subsection{The univariate reduction}

At last, we present a generic construction of a halfspace whose approximation
by rational functions and polynomials gives corresponding approximants
for the sign function on the discrete set $\{\pm1,\pm2,\ldots,\pm m\}$.
In more detail, let $z_{1},z_{2},\ldots,z_{n}$ be given integers.
For any such $n$-tuple, we define an associated halfspace and prove
a lower bound on $m$ in terms of the discrepancy of the multiset
$\{z_{1},z_{2},\ldots,z_{n}\}.$ The following first-principles calculation
will be helpful.
\begin{prop}
\label{prop:averaging-num-denom}Let $a_{1},a_{2},\ldots,a_{k}\in\Re$
and $b_{1},b_{2},\ldots,b_{k}>0$. Then
\begin{equation}
\min\frac{a_{i}}{b_{i}}\leq\frac{\Exp a_{i}}{\Exp b_{i}}\leq\max\frac{a_{i}}{b_{i}}.\label{eq:averaging-num-denom}
\end{equation}
\end{prop}

\begin{proof}
Abbreviate $m=\min a_{i}/b_{i}$ and $M=\max a_{i}/b_{i}.$ Since
each $b_{i}$ is positive, we obtain $mb_{i}\leq a_{i}\leq Mb_{i}$.
Taking a weighted sum of these inequalities, we arrive at $m\Exp b_{i}\leq\Exp a_{i}\leq M\Exp b_{i},$
which is equivalent to~(\ref{eq:averaging-num-denom}).
\end{proof}
We have:
\begin{thm}
\label{thm:master}Fix $\delta\in[0,1/2)$ and a nonempty multiset
$Z=\{z_{1},z_{2},\ldots,z_{n}\}$ of integers. Let $m$ be an integer
with
\begin{equation}
2\leq m\leq\left(\frac{2(1-2\delta)}{1+\disc(Z,m)}\right)^{\left(\frac{1}{2}-\delta\right)n}2^{-H(\delta)n-2}.\label{eq:m-range}
\end{equation}
Define $f\colon\zoo^{n}\times\zoon\to\moo$ by 
\[
f(x,y)=\sign\left(\frac{1}{2}+\sum_{j=1}^{n}(z_{j}\bmod m)x_{j}-m\sum_{j=1}^{n}y_{j}\right).
\]
Then
\begin{align*}
R(f,d_{0},d_{1}) & \geq R(\sign|_{\{\pm1,\pm2,\ldots,\pm m\}},2d_{0},2d_{1})
\end{align*}
for all $d_{0},d_{1}=0,1,2,\ldots,\lfloor\delta n/2\rfloor.$
\end{thm}

\begin{proof}
Fix $0<\epsilon<1$ arbitrarily for the remainder of the proof, and
suppose that $R(f,d_{0},d_{1})<\epsilon$ for some $d_{0},d_{1}\leq\delta n/2.$
Our goal is to show that
\begin{equation}
R(\sign|_{\{\pm1,\pm2,\ldots,\pm m\}},2d_{0},2d_{1})<\epsilon.\label{eq:master-goal}
\end{equation}
The proof is algorithmic and involves three steps. Given any approximant
for $f$, we will first manipulate it to control the sign behavior
in the numerator and denominator, then symmetrize it with respect
to $y,$ and finally\textemdash the arduous part of the proof\textemdash symmetrize
it with respect to $x.$ The result of these manipulations will be
a univariate approximant for the sign function.

\medskip{}

\emph{Step 1: Original approximant.} Since $R(f,d_{0},d_{1})<\epsilon,$
there are polynomials $p$ and $q$ of degree at most $d_{0}$ and
$d_{1},$ respectively, with
\begin{align*}
\left|f(x,y)-\frac{p(x,y)}{q(x,y)}\right| & <\epsilon
\end{align*}
for all $x,y\in\zoon.$ This inequality is equivalent to
\begin{align}
1-\epsilon & <\frac{p(x,y)}{q(x,y)}f(x,y)<1+\epsilon.\label{eq:orig-ineq}
\end{align}
Observe that for all $x,y\in\zoon,$ we have $p(x,y)\ne0$ and $q(x,y)\ne0,$
where the former is a consequence of $\epsilon<1$ and the latter
follows from the definition of a rational approximant. As a result,
(\ref{eq:orig-ineq}) gives
\begin{align}
1-\epsilon & <\frac{p(x,y)q(x,y)f(x,y)}{q(x,y)^{2}}<1+\epsilon,\label{eq:squared-ineq1}\\
1-\epsilon & <\frac{p(x,y)^{2}}{p(x,y)q(x,y)f(x,y)}<1+\epsilon.\label{eq:squared-ineq2}
\end{align}
\medskip{}

\emph{Step 2: Symmetrization on $y.$} The fractions in~(\ref{eq:squared-ineq1})
and~(\ref{eq:squared-ineq2}) have positive numerators and denominators.
Therefore, Proposition~\ref{prop:averaging-num-denom} implies that
\begin{align}
1-\epsilon<\frac{\Exp_{\sigma\in S_{n}}[p(x,\sigma y)q(x,\sigma y)f(x,\sigma y)]}{\Exp_{\sigma\in S_{n}}[q(x,\sigma y)^{2}]} & <1+\epsilon,\label{eq:symmetrizing1}\\
\rule{0mm}{6mm}1-\epsilon<\frac{\Exp_{\sigma\in S_{n}}[p(x,\sigma y)^{2}]}{\Exp_{\sigma\in S_{n}}[p(x,\sigma y)q(x,\sigma y)f(x,\sigma y)]} & <1+\epsilon.\label{eq:symmetrizing2}
\end{align}
Minsky and Papert's symmetrization technique (Proposition~\ref{prop:minsky-papert})
ensures the existence of polynomials $p^{*},q^{*},r^{*}$ of degree
at most $2d_{0},$ $2d_{1},$ and $d_{0}+d_{1}$, respectively, such
that for all $x,y\in\zoon,$
\begin{align*}
\Exp_{\sigma\in S_{n}}[p(x,\sigma y)^{2}] & \equiv p^{*}(x,|y|),\\
\Exp_{\sigma\in S_{n}}[q(x,\sigma y)^{2}] & \equiv q^{*}(x,|y|),\\
\Exp_{\sigma\in S_{n}}[p(x,\sigma y)q(x,\sigma y)] & \equiv r^{*}(x,|y|).
\end{align*}
Moreover, 
\[
f(x,\sigma y)\equiv f^{*}(x,|y|)
\]
for all $\sigma\in S_{n},$ where $f^{*}\colon\zoon\times\{0,1,2,\dots,n\}\to\moo$
is given by
\[
f^{*}(x,t)=\sign\left(\frac{1}{2}+\sum_{j=1}^{n}(z_{j}\bmod m)x_{j}-mt\right).
\]
Now~(\ref{eq:symmetrizing1}) and~(\ref{eq:symmetrizing2}) simplify
to
\begin{align}
1-\epsilon & <\frac{r^{*}(x,t)f^{*}(x,t)}{q^{*}(x,t)}<1+\epsilon,\label{eq:f-star1}\\
1-\epsilon & <\frac{p^{*}(x,t)}{r^{*}(x,t)f^{*}(x,t)}<1+\epsilon\label{eq:f-star2}
\end{align}
for all $x\in\zoon$ and $t=0,1,2,\ldots n.$ The numerators and denominators
of these fractions are again positive, being averages of positive
numbers.\medskip{}

\emph{Step 3: Symmetrization on $x.$} We have reached the most demanding
part of the proof, where we symmetrize the approximants obtained so
far with respect to $x.$ For $s\in\ZZ,$ let $\Xcal_{s}\subseteq\zoon$
be given by~(\ref{eq:Xcal-s-defined}). Then Lemma~\ref{lem:fooling-distributions}
guarantees that each $\Xcal_{s}$ is nonempty, and additionally provides
a probability distribution $\mu_{s}$ on $\Xcal_{s}$ (for each $s\in\ZZ$)
such that for every polynomial $P\colon\zoon\to\Re,$
\begin{equation}
\deg P\leq\delta n\quad\implies\qquad\qquad\Exp_{\mu_{s}}P(x)=\Exp_{\mu_{s'}}P(x)\qquad\forall s,s'\in\ZZ.\label{eq:master-orthogonalization}
\end{equation}
Now fix an integer $s\in[-m-1,m-1].$ On the support of $\mu_{s},$
we have 
\begin{align*}
\sum_{j=1}^{n}(z_{j}\bmod m)x_{j}-s & \in[0\cdot n-m+1,(m-1)\cdot n+m+1]\cap m\ZZ\\
 & \subseteq(-m,(n+1)m)\cap m\ZZ\\
 & =\{0,m,2m,\ldots,nm\},
\end{align*}
where the second step is valid because $n\geq2$ by~(\ref{eq:m-range}).
It follows that on the support of $\mu_{s},$ the linear form
\[
\ell(x,s)=\frac{1}{m}\left(\sum_{j=1}^{n}(z_{j}\bmod m)x_{j}-s\right)
\]
ranges in $\{0,1,2,\ldots,n\},$ forcing $f^{*}(x,\ell(x,s))=\sign(s+\frac{1}{2})$.
Now~(\ref{eq:f-star1}) and~(\ref{eq:f-star2}) imply that
\begin{align*}
1-\epsilon & <\frac{r^{*}(x,\ell(x,s))\sign(s+\frac{1}{2})}{q^{*}(x,\ell(x,s))}<1+\epsilon,\\
1-\epsilon & <\frac{p^{*}(x,\ell(x,s))}{r^{*}(x,\ell(x,s))\sign(s+\frac{1}{2})}<1+\epsilon
\end{align*}
for all integers $s\in[-m-1,m-1]$ and all $x$ in the support of
$\mu_{s}.$ Since the numerators and denominators of these fractions
are positive, Proposition~\ref{prop:averaging-num-denom} allows
us to pass to expectations with respect to $x\sim\mu_{s}$ to obtain
\begin{align*}
1-\epsilon & <\frac{\Exp_{x\sim\mu_{s}}[r^{*}(x,\ell(x,s))]\sign(s+\frac{1}{2})}{\Exp_{x\sim\mu_{s}}[q^{*}(x,\ell(x,s))]}<1+\epsilon,\\
1-\epsilon & <\frac{\Exp_{x\sim\mu_{s}}[p^{*}(x,\ell(x,s))]}{\Exp_{x\sim\mu_{s}}[r^{*}(x,\ell(x,s))]\sign(s+\frac{1}{2})}<1+\epsilon,
\end{align*}
or equivalently 
\begin{align}
\left|\frac{\Exp_{x\sim\mu_{s}}[r^{*}(x,\ell(x,s))]}{\Exp_{x\sim\mu_{s}}[q^{*}(x,\ell(x,s))]}-\sign\left(s+\frac{1}{2}\right)\right| & <\epsilon,\label{eq:exp-final1}\\
\rule{0mm}{6mm}\left|\frac{\Exp_{x\sim\mu_{s}}[p^{*}(x,\ell(x,s))]}{\Exp_{x\sim\mu_{s}}[r^{*}(x,\ell(x,s))]}-\sign\left(s+\frac{1}{2}\right)\right| & <\epsilon,\label{eq:exp-final2}
\end{align}
for all integers $s\in[-m-1,m-1].$

Consider the univariate polynomials
\begin{align*}
p^{**}(s) & =\Exp_{x\sim\mu_{s}}[p^{*}(x,\ell(x,s))],\\
q^{**}(s) & =\Exp_{x\sim\mu_{s}}[q^{*}(x,\ell(x,s))],\\
r^{**}(s) & =\Exp_{x\sim\mu_{s}}[r^{*}(x,\ell(x,s))].
\end{align*}
Equations~(\ref{eq:exp-final1}) and~(\ref{eq:exp-final2}) show
that $r^{**}(s-1)/q^{**}(s-1)$ and $p^{**}(s-1)/r^{**}(s-1)$ approximate
$\sign s$ pointwise on $\{\pm1,\pm2,\ldots,\pm m\}$ to error less
than $\epsilon.$ Moreover, (\ref{eq:master-orthogonalization}) ensures
that the degrees of $p^{**},q^{**},r^{**}$ are at most the degrees
of $p^{*},q^{*},r^{*},$ respectively. We conclude that 
\begin{align*}
R(\sign|_{\{\pm1,\pm2,\ldots,\pm m\}},d_{0}+d_{1},2d_{1}) & <\epsilon,\\
R(\sign|_{\{\pm1,\pm2,\ldots,\pm m\}},2d_{0},d_{0}+d_{1}) & <\epsilon.
\end{align*}
These complementary bounds force~(\ref{eq:master-goal}) and thereby
complete the proof.
\end{proof}

\subsection{\label{subsec:The-master-theorem}The master theorem}

We now combine Theorem~\ref{thm:master} with the efficient construction,
in Theorem~\ref{thm:explicit-set-small-Fourier-coeffs}, of an integer
set with small $m$-discrepancy for $m=2^{\Theta(n)}$. The result
is an explicit halfspace $h_{n}\colon\zoon\to\moo$ whose approximation
by polynomials and rational functions is asymptotically equivalent
to the univariate approximation of the sign function on $\{\pm1,\pm2,\pm3,\ldots,\pm2^{\Theta(n)}\}$.
We refer to this result as our \emph{master theorem} since all our
main theorems are derived from it.
\begin{thm}
\label{thm:rational-approx-h-REDUCTION-TO-SGN}For some constant $c'>0,$
there is an algorithm that takes as input an integer $n\geq1,$ runs
in time polynomial in $n,$ and outputs a halfspace $h_{n}\colon\zoon\to\moo$
with
\begin{equation}
R(h_{n},d_{0},d_{1})\geq R\left(\sign|_{\{\pm1,\pm2,\pm3,\ldots,\pm2^{\lfloor c'n\rfloor}\}},2d_{0},2d_{1}\right)\label{eq:rational-approx-h-REDUCTION-TO-SGN}
\end{equation}
for all $d_{0},d_{1}=0,1,2,\ldots,\lfloor c'n\rfloor.$ Moreover,
the constant $c'$ and the algorithm are given explicitly.
\end{thm}

\begin{proof}
Let
\begin{equation}
c'=\min\left\{ \frac{1}{200},\frac{1}{2C_{1/10}}\right\} ,\label{eq:def-c}
\end{equation}
where $C_{1/10}\geq1$ is the constant defined in Theorem~\ref{thm:explicit-set-small-Fourier-coeffs}.
On input $n$, the construction of $h_{n}$ is as follows. For $n<1/c',$
the sought property~(\ref{eq:rational-approx-h-REDUCTION-TO-SGN})
amounts to $R(h_{n},0,0)\geq R(\sign|_{\{-1,1\}},0,0)$, which is
in turn equivalent to $R(h_{n},0,0)\geq1$ and holds trivially for
the halfspace $h_{n}(x)=(-1)^{x_{1}}.$

We now turn to the nontrivial case, $n\geq1/c'.$ Abbreviate $m=2^{\lfloor c'n\rfloor}.$
Then the algorithm of Theorem~\ref{thm:explicit-set-small-Fourier-coeffs}
constructs, in time polynomial in $n,$ a nonempty multiset $Z$ with
$m$-discrepancy
\begin{align}
\disc(Z,m) & \leq\frac{1}{10}\label{eq:Z-norm-small}
\end{align}
and cardinality $|Z|\leq n/2.$ Observe that for any integer $k\geq1,$
the union of $k$ copies of $Z$ is a multiset with $m$-discrepancy
$\disc(Z,m)$ and cardinality $k|Z|$. Therefore, we may assume without
loss of generality that
\begin{equation}
\frac{n}{4}\leq|Z|\leq\frac{n}{2}.\label{eq:Z-not-too-small}
\end{equation}
We let
\[
h_{n}(x)=\sign\left(\frac{1}{2}+\sum_{j=1}^{|Z|}(z_{j}\bmod m)x_{j}-m\sum_{j=|Z|+1}^{2|Z|}x_{j}\right),
\]
where $z_{1},z_{2},z_{3},\ldots,z_{|Z|}$ denote the elements of the
multiset $Z.$ Taking $\delta=1/25,$ we have from~(\ref{eq:def-c})
and~(\ref{eq:Z-not-too-small}) that
\begin{equation}
c'n\leq\frac{\delta|Z|}{2}.\label{eq:cn-delta-Z}
\end{equation}
Moreover,
\begin{align*}
m & \in[2,2^{c'n}]\\
 & \subseteq[2,2^{n/200}]\\
 & \subseteq\left[2,\left(\frac{2(1-2\delta)}{1+\disc(Z,m)}\right)^{\left(\frac{1}{2}-\delta\right)\cdot|Z|}2^{-H(\delta)\cdot|Z|-2}\right],
\end{align*}
where the second step applies~(\ref{eq:def-c}), and the third step
uses~(\ref{eq:Z-norm-small}), (\ref{eq:Z-not-too-small}), and $n\geq1/c'\geq200.$
As a result, Theorem~\ref{thm:master} implies~(\ref{eq:rational-approx-h-REDUCTION-TO-SGN})
for all $d_{0},d_{1}\leq\delta|Z|/2$. In view of~(\ref{eq:cn-delta-Z}),
the proof is complete.
\end{proof}

\section{\label{sec:Main-results}Main results}

Using the halfspace $h_{n}$ constructed in our master theorem, we
will now establish the main results of this paper.

\subsection{\label{subsec:Polynomial-approximation}Polynomial approximation}

Prior to our work, the strongest lower bound for the approximation
of an explicit halfspace $f_{n}\colon\zoon\to\moo$ by polynomials
was $E(f_{n},c\sqrt{n})\geq1-2^{-c\sqrt{n}}$ for an absolute constant
$c>0$, proved in~\cite{sherstov09hshs,sherstov09opthshs}. The result
that we are about to prove is a quadratic improvement on previous
work, with respect to both degree and error. As we will discuss shortly,
this new result is essentially the best possible.
\begin{thm}[Polynomial approximation]
\label{thm:polynomial-approx-hs-LOWER} Let $h_{n}\colon\zoon\to\moo$
be the halfspace constructed in Theorem~\emph{\ref{thm:rational-approx-h-REDUCTION-TO-SGN}}.
Then for some constant $c>0$ and all $n,$
\begin{equation}
E(h_{n},cn)>1-2^{-cn}.\label{eq:approx-poly-lower}
\end{equation}
\end{thm}

\begin{proof}
Let $c'>0$ be the constant in Theorem~\ref{thm:rational-approx-h-REDUCTION-TO-SGN}.
Then
\begin{align*}
E(h_{n},c'n) & \geq E(\sign|_{\{\pm1,\pm2,\pm3,\ldots,\pm2^{\lfloor c'n\rfloor}\}},2\lfloor c'n\rfloor)\\
 & \geq1-O\left(\frac{n}{2^{c'n}}\right)^{1/2},
\end{align*}
where the first step corresponds to taking $d_{0}=\lfloor c'n\rfloor$
and $d_{1}=0$ in Theorem~\ref{thm:rational-approx-h-REDUCTION-TO-SGN},
and the second step is immediate from Proposition~\ref{prop:polynomial-approx-SGN-lower}.
This implies~(\ref{eq:approx-poly-lower}) for $c>0$ small enough.
\end{proof}
\noindent Theorem~\ref{thm:polynomial-approx-hs-LOWER} is essentially
as strong as one could hope for. First of all, any function in $n$
Boolean variables can be approximated to zero error by a polynomial
of degree at most $n,$ i.e., at most a constant factor larger than
what is assumed in~(\ref{eq:approx-poly-lower}). Moreover, a classic
result due to Muroga~\cite{muroga71threshold} implies that for every
halfspace, the error bound in~(\ref{eq:approx-poly-lower}) is almost
achieved by polynomials of degree $1$:
\begin{fact}
\label{fact:muroga}There is an absolute constant $c>0$ such that
for every $n$ and every halfspace $h\colon\zoon\to\moo,$
\begin{align*}
E(h,1) & \leq1-n^{-cn}.
\end{align*}
\end{fact}

\begin{proof}
Muroga~\cite{muroga71threshold} showed that every halfspace $h\colon\zoon\to\moo$
can be represented as $h(x)=\sign(\sum_{j=1}^{n}z_{j}x_{j}-\theta)$
for some integers $z_{1},z_{2},\ldots,z_{n},\theta$ whose absolute
values sum to $n^{O(n)}.$ It follows that
\begin{align*}
E(h,1) & \leq\max_{x\in\zoon}\left|h(x)-\frac{1}{|\theta|+\sum_{j=1}^{n}|z_{j}|}\left(\sum_{j=1}^{n}z_{j}x_{j}-\theta\right)\right|\\
 & \leq1-\frac{1}{|\theta|+\sum_{j=1}^{n}|z_{j}|}\\
 & \leq1-n^{-O(n)}.\qedhere
\end{align*}
\end{proof}

\subsection{\label{subsec:Rational-approximation}Rational approximation}

We now show that the halfspace $h_{n}$ constructed in our master
theorem cannot be approximated pointwise to any small constant except
by rational functions of degree $\Omega(n)$. This degree lower bound
matches the trivial upper bound and is a quadratic improvement on
the previous best construction~\cite{sherstov09hshs,sherstov09opthshs}.
More generally, we derive a lower bound on the approximation of $h_{n}$
by rational functions of any given degree $d$, and this lower bound
too is essentially the best possible for any halfspace. Details follow.
\begin{thm}[Rational approximation]
\label{thm:rational-approx-to-hs-LOWER} Let $h_{n}\colon\zoon\to\moo$
be the halfspace constructed in Theorem~\emph{\ref{thm:rational-approx-h-REDUCTION-TO-SGN}}.
Then for some constant $c>0$ and all $n,$
\begin{align}
R(h_{n},d) & \geq1-\exp\left(-\frac{cn}{d}\right), &  & d=1,2,\ldots,\lfloor cn\rfloor.\label{eq:rational-approx-to-hs-LOWER}
\end{align}
\end{thm}

\begin{proof}
Let $c'>0$ be the constant in Theorem~\ref{thm:rational-approx-h-REDUCTION-TO-SGN}.
Then for $d=1,2,\ldots,\lfloor c'n\rfloor,$ we have
\begin{align*}
R(h_{n},d) & \geq R(\sign|_{\{\pm1,\pm2,\pm3,\ldots,\pm2^{\lfloor c'n\rfloor}\}},2d)\\
 & \geq1-\exp\left(-\Theta\left(\frac{n}{d}\right)\right),
\end{align*}
where the first step corresponds to taking $d_{0}=d_{1}=d$ in Theorem~\ref{thm:rational-approx-h-REDUCTION-TO-SGN},
and the second step is immediate from Theorem~\ref{thm:rational-approx-SGN}.
This implies~(\ref{eq:rational-approx-to-hs-LOWER}) for $c>0$ small
enough. 
\end{proof}
\noindent We now show that the lower bounds on the approximation error
in Theorem~\ref{thm:rational-approx-to-hs-LOWER} are essentially
the best possible for any halfspace.
\begin{fact}
\label{fact:rational-approx-to-hs-UPPER}There exists an absolute
constant $c>0$ such that for every $n$ and every halfspace $h\colon\zoon\to\moo,$
\begin{align*}
R(h,d) & \leq1-\exp\left(-\frac{cn\log n}{d}\right), &  & d=1,2,\ldots,n.
\end{align*}
\end{fact}

\begin{proof}
As already mentioned, Muroga~\cite{muroga71threshold} showed that
$h(x)\equiv\sign p(x)$ for some linear polynomial $p(x)$ that ranges
in $[-N,-1]\cup[1,N],$ where $N=\exp(cn\log n)$ for some absolute
constant $c>0$. This makes it possible to obtain a rational approximant
for $h(x)$ by taking any rational approximant for the sign function
on $[-N,-1]\cup[1,N]$ and composing it with $p(x)$. We conclude
that for any integer $d$,
\begin{align*}
R(h,d) & \leq R(\sign|_{[-N,-1]\cup[1,N]},d)\\
 & \leq1-\frac{1}{N^{1/d}}\\
 & =1-\exp\left(-\frac{cn\log n}{d}\right),
\end{align*}
where the second step uses Newman's rational approximation (Fact~\ref{fact:newman}).
\end{proof}

\subsection{Threshold degree}

Here, we use the halfspace $h_{n}$ constructed in our master theorem
to study the degree required to sign-represent intersections of halfspaces.
Our result is a lower bound of $\Omega(n)$ for the intersection $h_{n}\wedge h_{n}$
of two independent copies of $h_{n}.$ This result improves quadratically
on the previous best construction~\cite{sherstov09hshs,sherstov09opthshs}
and matches the trivial upper bound of $O(n)$ for sign-representing
any Boolean function in $n$ variables.
\begin{thm}
\label{thm:degthr-h-h} Let $h_{n}\colon\zoon\to\moo$ be the halfspace
constructed in Theorem~\emph{\ref{thm:rational-approx-h-REDUCTION-TO-SGN}}.
Then 
\[
\degthr(h_{n}\wedge h_{n})=\Omega(n).
\]
\end{thm}

\begin{proof}
Abbreviate $D_{n}=\degthr(h_{n}\wedge h_{n}).$ Taking $f=g=h_{n}$
in Theorem~\ref{thm:sherstov-degthr-R} shows that $R(h_{n},4D_{n})<1/2,$
which by Theorem~\ref{thm:rational-approx-to-hs-LOWER} forces $D_{n}=\Omega(n).$
\end{proof}
\noindent Theorem~\ref{thm:degthr-h-h} should be contrasted with
the result of Beigel et al.~\cite{beigel91rational} that the conjunction
of any constant number of majority functions on $\zoon$ has threshold
degree $O(\log n).$ We now derive a lower bound of $\Omega(\sqrt{n\log n})$
on the threshold degree of the intersection of an explicitly given
halfspace and a majority function, improving quadratically on the
previous best construction~\cite{sherstov09hshs,sherstov09opthshs}.
As we discuss shortly, the new construction is optimal up to a logarithmic
factor.
\begin{thm}
\label{thm:degthr-h-maj} Let $h_{n}\colon\zoon\to\moo$ be the halfspace
constructed in Theorem~\emph{\ref{thm:rational-approx-h-REDUCTION-TO-SGN}}.
Then
\begin{align}
\degthr(h_{n}\wedge\MAJ_{n})=\Omega(\sqrt{n\log n}).\label{eqn:degthr-h-maj}
\end{align}
\end{thm}

\begin{proof}
Abbreviate $D_{n}=\degthr(h_{n}\wedge\MAJ_{n}).$ Then $R(h_{n},4D_{n})+R(\MAJ_{n},2D_{n})<1$
by Theorem~\ref{thm:sherstov-degthr-R}. The lower bounds for the
rational approximation of $h_{n}$ and $\MAJ_{n}$ in Theorems~\ref{thm:rational-approx-MAJ}
and~\ref{thm:rational-approx-to-hs-LOWER} now imply that $D_{n}=\Omega(\sqrt{n\log n}).$
\end{proof}
\begin{rem}
\label{rem:degthr-h-MAJ-upper}The construction of Theorem~\ref{thm:degthr-h-maj}
is essentially the best possible, in that 
\begin{equation}
\degthr(h\wedge\MAJ_{n})=O(\sqrt{n}\log n)\label{eq:h-MAJ-upper}
\end{equation}
for every halfspace $h\colon\zoon\to\moo.$ Indeed, taking $d=C\sqrt{n}\log n$
in Theorem~\ref{thm:rational-approx-MAJ} and Fact~\ref{fact:rational-approx-to-hs-UPPER}
for a large enough constant $C\geq1$ yields $R(h,C\sqrt{n}\log n)+R(\MAJ_{n},C\sqrt{n}\log n)<1,$
which in turn implies~(\ref{eq:h-MAJ-upper}) in view of Theorem~\ref{thm:beigel-degthr-rational}.
\end{rem}

\subsection{Threshold density}

In addition to threshold degree, several other complexity measures
are of interest when sign-representing Boolean functions by real polynomials.
One such complexity measure is \emph{threshold density}, defined as
the least $k$ for which a given function can be sign-represented
by a linear combination of $k$ parity functions. Formally, for a
given function $f\colon\zoon\to\moo,$ its threshold density $\dns(f)$
is the minimum size $|\Scal|$ of a family $\Scal\subseteq\Pcal(\oneton)$
such that 
\begin{align*}
f(x)\equiv\sign\left(\sum_{S\in\Scal}w_{S}(-1)^{\sum_{j\in S}x_{j}}\right)
\end{align*}
for some reals $w_{S}.$ It is clear from the definition that $\dns(f)\leq2^{n}$
for all functions $f\colon\zoon\to\moo,$ and we will now construct
a pair of halfspaces whose intersection has threshold density $2^{\Theta(n)}.$
Prior to our work, the best construction~\cite{sherstov09hshs} had
threshold density $2^{\Theta(\sqrt{n})}.$

\global\long\def\op{\text{{\rm KP}}}
 To proceed, we recall a technique due to Krause and Pudlák~\cite{krause94depth2mod}
that transforms Boolean functions with high threshold degree into
Boolean functions with high threshold density. Their transformation
works in a black-box manner and sends a function $f\colon\zoon\to\moo$
to the function $f^{\op}\colon(\zoon)^{3}\to\moo$ defined by 
\begin{align*}
f^{\op}(x,y,z) & =f(\dots,(\overline{z_{i}}\wedge x_{i})\vee(z_{i}\wedge y_{i}),\dots).
\end{align*}
The threshold degree of $f$ and the threshold density of $f^{\op}$
are related as follows~\cite[Proposition~2.1]{krause94depth2mod}.
\begin{thm}[Krause and Pudlák]
\label{thm:degree-length} For every function $f\colon\zoon\to\moo,$
\begin{align*}
\dns(f^{\op})\geq2^{\degthr(f)}.
\end{align*}
\end{thm}

\noindent We are now in a position to obtain the claimed density results.
\begin{thm}
\label{thm:dns} There is an $($explicit$)$ algorithm that takes
as input an integer $n\geq1,$ runs in time polynomial in $n,$ and
outputs a halfspace $H_{n}\colon\zoon\to\moo$ such that
\begin{align}
\dns(H_{n}\wedge H_{n}) & =2^{\Omega(n)},\label{eqn:hh}\\
\dns(H_{n}\wedge\MAJ_{n}) & =2^{\Omega(\sqrt{n\log n})}.\label{eqn:hmajn}
\end{align}
\end{thm}

\begin{proof}
For any function $f\colon\zoon\to\zoo,$ standard arithmetization
gives
\begin{equation}
f^{\op}(x,y,z)=f\left(\ldots,\frac{1}{2}(x_{i}+y_{i}+x_{i}\oplus z_{i}-y_{i}\oplus z_{i}),\ldots\right),\label{eq:KP-arithmetic}
\end{equation}
where $a\oplus b\in\zoo$ denotes as usual the XOR of $a$ and $b$.
Similarly, one has
\begin{equation}
\MAJ_{n}^{\op}(x,y,z)=\MAJ_{4n}(x,y,x\oplus z,\overline{y\oplus z}),\label{eq:KP-MAJ}
\end{equation}
where the XOR and complement operations are applied bitwise. 

Let $h_{n}\colon\zoon\to\moo$ be the halfspace from Theorem~\ref{thm:degthr-h-h},
so that $h_{n}\wedge h_{n}$ has threshold degree $\Omega(n).$ By
Theorem~\ref{thm:degree-length}, the function $(h_{n}\wedge h_{n})^{\op}=h_{n}^{\op}\wedge h_{n}^{\op}$
has threshold density $2^{\Omega(n)}.$ Observe from~(\ref{eq:KP-arithmetic})
that $h_{n}^{\op}\wedge h_{n}^{\op}$ is the result of starting with
the intersection $H_{4n}\wedge H_{4n}$ of two explicitly given halfspaces
in $4n$ variables each, and replacing their input variables with
appropriately chosen parity functions. This replacement cannot increase
the threshold density because the parity of several parity functions
is another parity function. We conclude that $\dns(H_{4n}\wedge H_{4n})=2^{\Omega(n)}.$
This completes the proof of (\ref{eqn:hh}).

The proof of (\ref{eqn:hmajn}) is closely analogous. Specifically,
recall from Theorem~\ref{thm:degthr-h-maj} that $h_{n}\wedge\MAJ_{n}$
has threshold degree $\Omega(\sqrt{n\log n}).$ By Theorem~\ref{thm:degree-length},
the function $(h_{n}\wedge\MAJ_{n})^{\op}=h_{n}^{\op}\wedge\MAJ_{n}^{\op}$
has threshold density $\exp(\Omega(\sqrt{n\log n})).$ It follows
from~(\ref{eq:KP-arithmetic}) and~(\ref{eq:KP-MAJ}) that $h_{n}^{\op}\wedge\MAJ_{n}^{\op}$
is the result of starting with the intersection $H_{4n}\wedge\MAJ_{4n}$
for an explicit halfspace $H_{4n}$ in $4n$ variables, and replacing
the input variables with appropriately chosen parity functions or
their negations. This replacement cannot increase the threshold density
because the parity of several parity functions is another parity function.
We conclude that $\dns(H_{4n}\wedge\MAJ_{4n})=\exp(\Omega(\sqrt{n\log n})).$
This completes the proof of (\ref{eqn:hmajn}).
\end{proof}
\noindent Both lower bounds in Theorem~\ref{thm:dns} are essentially
the best possible for any halfspace $H_{n}\colon\zoon\to\moo$. Indeed,
the first lower bound is tight by definition, while the second lower
bound nearly matches the upper bound of $\exp(O(\sqrt{n}\log^{2}n))$
that follows from Remark~\ref{rem:degthr-h-MAJ-upper}. 

\subsection{Communication complexity}

Using the pattern matrix method, we will now ``lift'' the approximation
lower bound of Theorem~\ref{thm:polynomial-approx-hs-LOWER} to communication
complexity. As a result, we will obtain an explicit separation of
$k$-party communication complexity with unbounded and weakly unbounded
error (which for $k=2$ is equivalent to a separation of sign-rank
and discrepancy). Our application of the pattern matrix method is
based on the fact that the unique set disjointness function $\UDISJ_{m,k}$
has an exact representation on its domain as a polynomial with a small
number of monomials; cf.~\cite[Section~10]{sherstov07quantum}, \cite[Section~4.2.3]{thaler14omb},
and~\cite[Section~3.1]{sherstov16multiparty-pp-upp}. Specifically,
define $\UDISJ_{m,k}^{*}\colon(\zoo^{m})^{k}\to\Re$ by 
\[
\UDISJ_{m,k}^{*}(x)=-1+2\sum_{i=1}^{m}x_{1,i}x_{2,i}\cdots x_{k,i}\,.
\]
Then 
\begin{align}
\UDISJ_{m,k}(x)=\UDISJ_{m,k}^{*}(x), &  & x\in\dom\UDISJ_{m,k}.\label{eq:udisj-linear-poly}
\end{align}
\begin{thm}
\label{thm:pp-upp}For some constant $C>1$ and all positive integers
$n$ and $k,$ there is an $($explicitly given$)$ $k$-party communication
problem $F_{n,k}\colon(\zoon)^{k}\to\moo$ such that 
\begin{align}
\upp(F_{n,k}) & \leq\log n+4,\label{eq:Fnk-upp}\\
\pp(F_{n,k}) & \geq\left\lfloor \frac{n}{C\cdot4^{k}}\right\rfloor ,\label{eq:Fnk-pp}\\
\disc(F_{n,k}) & \leq\exp\left(-\left\lfloor \frac{n}{C\cdot4^{k}}\right\rfloor \right).\label{eq:Fnk-disc}
\end{align}
Moreover, 
\begin{align}
F_{n,k}(x_{1},x_{2},\ldots,x_{k})=\sign\left(w_{0}+\sum_{i=1}^{n}w_{i}x_{1,i}x_{2,i}\cdots x_{k,i}\right)\label{eq:Fnk-form}
\end{align}
for some explicitly given reals $w_{0},w_{1},\dots,w_{n}$. 
\end{thm}

\begin{proof}
Let $h_{n}\colon\zoon\to\moo$ be the halfspace constructed in Theorem~\ref{thm:rational-approx-h-REDUCTION-TO-SGN}.
Then by definition, $h_{n}(x)=\sign p_{n}(x)$ for a linear polynomial
$p_{n}\colon\Re^{n}\to\Re$. Moreover, Theorem~\ref{thm:polynomial-approx-hs-LOWER}
ensures that 
\begin{align}
\deg_{1-2^{-cn}}(h_{n})\geq cn\label{eq:poly-approx-error-h}
\end{align}
for some constant $c>0$ independent of $n$. Abbreviate $m=\lceil2^{k+1}\e/c\rceil^{2}$
and consider the $k$-party communication problem $F'_{n,k}\colon(\zoo^{nm})^{k}\to\moo$
given by 
\begin{equation}
F'_{n,k}=\Sgn\,p_{n}\!\left(\frac{1-\UDISJ_{m,k}^{*}}{2},\frac{1-\UDISJ_{m,k}^{*}}{2},\ldots,\frac{1-\UDISJ_{m,k}^{*}}{2}\right),\label{eq:Fnk-prime-def}
\end{equation}
where the right-hand side features the coordinatewise composition
of the polynomial $p_{n}$ with $n$ independent copies of the polynomial
$(1-\UDISJ_{m,k}^{*})/2$. The identity (\ref{eq:udisj-linear-poly})
implies that $F'_{n,k}$ coincides with $h_{n}\circ\UDISJ_{m,k}$
on the domain of the latter. Therefore, 
\begin{align}
\disc(F'_{n,k}) & \leq\disc(h_{n}\circ\UDISJ_{m,k})\nonumber \\
 & \leq2^{-cn}+2^{-cn}\nonumber \\
 & =2\cdot2^{-cn},\label{eq:Fnk-prime-disc}
\end{align}
where the second step uses~(\ref{eq:poly-approx-error-h}) and the
pattern matrix method (Theorem~\ref{thm:pm-large-adeg}). Applying
the discrepancy method (Corollary~\ref{cor:dm}), we obtain
\begin{align}
\pp(F'_{n,k}) & \geq\log\frac{2}{\disc(F_{n,k}')}\nonumber \\
 & \geq cn.\label{eq:Fnk-prime-pp}
\end{align}

To complete the proof, define the functions $F_{n,k}$ for any positive
integers $n$ and $k$ by 
\[
F_{n,k}=\begin{cases}
F'_{\lfloor n/\lceil2^{k+1}\e/c\rceil^{2}\rfloor,k} & \text{if }n\geq\lceil2^{k+1}\e/c\rceil^{2},\\
0 & \text{otherwise.}
\end{cases}
\]
Then (\ref{eq:Fnk-pp})\textendash (\ref{eq:Fnk-form}) are immediate
from~(\ref{eq:Fnk-prime-def})\textendash (\ref{eq:Fnk-prime-pp}),
whereas~(\ref{eq:Fnk-upp}) is a consequence of~(\ref{eq:Fnk-form})
and Fact~\ref{fact:upp-upper-bound}.
\end{proof}
\noindent Theorem~\ref{thm:pp-upp} gives an explicit separation
$\PP_{k}\subsetneq\UPP_{k}$ for up to $k\leq(0.5-\epsilon)\log n$
parties, where $\epsilon>0$ is an arbitrary constant. The special
case $k=2$ can be equivalently stated as an explicit separation of
sign-rank and discrepancy:
\begin{cor}
\label{cor:pp-upp-2-party}There is an $($explicitly given$)$ family
$\{F_{n}\}_{n=1}^{\infty}$ of communication problems $F_{n}\colon\zoon\times\zoon\to\moo$
with
\begin{align}
\srank(F_{n}) & \leq n+1,\label{eq:Fn-srank}\\
\disc(F_{n}) & =2^{-\Omega(n)},\label{eq:Fn-disc}\\
\upp(F_{n}) & \leq\log n+4,\label{eq:Fn-upp}\\
\pp(F_{n}) & =\Omega(n).\label{eq:Fn-pp}
\end{align}
Moreover, 
\begin{align}
F_{n}(x,y)=\sign\left(w_{0}+\sum_{i=1}^{n}w_{i}x_{i}y_{i}\right)\label{eq:Fn-form}
\end{align}
for some explicitly given reals $w_{0},w_{1},\dots,w_{n}$.
\end{cor}

\begin{proof}
Equations~(\ref{eq:Fn-disc})\textendash (\ref{eq:Fn-form}) result
from setting $k=2$ in Theorem~\ref{thm:pp-upp}. The new item,~(\ref{eq:Fn-srank}),
is immediate from~(\ref{eq:Fn-form}).
\end{proof}
\noindent Theorem~\ref{thm:pp-upp} and Corollary~\ref{cor:pp-upp-2-party}
settle Theorems~\ref{thm:MAIN-pp-upp-multiparty} and~\ref{thm:MAIN-pp-upp},
respectively, from the introduction.

\subsection{\label{subsec:A-circulant-expander}A circulant expander}

Consider a $d$-regular undirected graph $G$ on $n$ vertices, with
adjacency matrix $A.$ Since $A$ is symmetric, it has $n$ real eigenvalues
(counting multiplicities). We denote these eigenvalues by $\lambda_{1}(G)\geq\lambda_{2}(G)\geq\cdots\geq\lambda_{n}(G)$
and define $\lambda(G)=\max\{|\lambda_{2}(G)|,|\lambda_{3}(G)|,\ldots,|\lambda_{n}(G)|\}$.
It is well known and straightforward to verify that $\lambda_{1}(G)=d$
and $|\lambda_{i}(G)|\leq d$ for $i=2,3,\ldots,n.$ We say that $G$
is an \emph{$\epsilon$-expander} if $\lambda(G)\leq\epsilon d.$
This spectral notion is intimately related to key graph-theoretic
and stochastic properties of $G$, such as vertex expansion and the
convergence rate of a random walk on $G$ to the uniform distribution.
One is typically interested in $\epsilon$-expanders that are $d$-regular
for $d$ as small as possible, where $0<\epsilon<1$ is a constant.
The existence of expanders with strong parameters can be verified
using the probabilistic method~\cite{alon-spencer08probab-method},
and explicit constructions are known as well.

In this section, we study the problem of constructing \emph{circulant
}expanders. Formally, a graph is \emph{circulant} if its adjacency
matrix is circulant. It is clear that a circulant graph is $d$-regular
for some $d$, meaning that every vertex has out-degree $d$ and in-degree
$d.$ We focus on circulant graphs that are undirected and have no
self-loops, which corresponds to adjacency matrices that are symmetric
and have zeroes on the diagonal. It is well known~\cite{alon-roichman94rando-cayley-graphs-and-expanders}
that for any $0<\epsilon<1$ and all large enough $n$, there exists
a circulant $\epsilon$-expander on $n$ vertices of degree $O(\log n)$.
This degree bound is asymptotically optimal~\cite{alon-roichman94rando-cayley-graphs-and-expanders,fmt06spectral-estimates-for-cayley-graphs,lns11nonexistence-circular-expander},
and the problem of constructing such circulant expanders explicitly
has been studied by several authors~\cite{alon86,AIKPS90aperiodic-set,alon-roichman94rando-cayley-graphs-and-expanders}.
The best construction prior to our work, due to Ajtai et al.~\cite{AIKPS90aperiodic-set},
achieves degree $(\log^{*}n)^{O(\log^{*}n)}\log n$. In this section,
we construct a circulant $\epsilon$-expander of optimal degree, $O(\log n)$,
for any constant $0<\epsilon<1$. By way of terminology, recall that
the adjacency matrix of a circulant graph on $n$ vertices is $\circulant(\1_{S})$
for some subset $S\subseteq\{0,1,2,\ldots,n-1\}.$ With this in mind,
we say that an algorithm \emph{constructs a circulant graph on $n$
vertices in time $T(n)$} if the algorithm outputs in time $T(n)$
the elements of the associated subset $S$. The formal statement of
our result follows.
\begin{thm}
\label{thm:circulant-expander}Let $0<\epsilon<1$ be given. Then
there is an $($explicitly given$)$ algorithm that takes as input
an integer $n\geq2$ and constructs in time polynomial in $\log n$
an undirected simple $d$-regular circulant graph $G_{n}$ on $n$
vertices, where
\begin{align}
 & 1\leq d\leq O(\log n),\label{eq:thm-expander-degree}\\
 & \lambda(G_{n})\leq\max\left\{ \epsilon,\frac{1}{n-1}\right\} d.\label{eq:thm-expander-gap}
\end{align}
\end{thm}

\begin{proof}
Let $C_{\epsilon}$ be the constant from Theorem~\ref{thm:explicit-set-small-Fourier-coeffs}.
We first consider the trivial case when $2(C_{\epsilon}\log n)^{2}\geq n,$
which means that $n$ is bounded by an explicit constant. In this
case, we take $G_{n}$ to be the complete graph on $n$ vertices.
It is clear that $G_{n}$ is a $d$-regular circulant graph for $d=n-1.$
The adjacency matrix of $G_{n}$ is $\circulant(0,1,1,\ldots,1)$,
whose eigenvalues by Corollary~\ref{cor:circulant-diagonalization}
are $n-1,-1,-1,\ldots,-1$. In particular, $\lambda(G_{n})=1=d/(n-1).$
This settles~(\ref{eq:thm-expander-gap}), whereas~(\ref{eq:thm-expander-degree})
holds trivially because $d$ and $n$ are bounded by a constant.

We now turn to the nontrivial case when $2(C_{\epsilon}\log n)^{2}<n.$
The algorithm of Theorem~\ref{thm:explicit-set-small-Fourier-coeffs}
constructs, in time polynomial in $\log n,$ a set $Z\subseteq\{0,1,2,\ldots,n-1\}$
with
\begin{align}
 & \disc(Z,n)\leq\epsilon,\label{eq:disc-Z-expander}\\
 & 1\leq|Z|\leq C_{\epsilon}\log n.\label{eq:Z-cardinality-expander}
\end{align}
For any $z,z'\in Z,$ the linear congruence $z+\Delta\equiv-(z'+\Delta)\pmod n$
has at most two solutions $\Delta\in\{0,1,2,\ldots,n-1\}.$ Recalling
that $2|Z|^{2}<n$ in the case under consideration, we conclude that
there exists $\Delta\in\{0,1,2,\ldots,2|Z|^{2}\}$ with 
\begin{equation}
z+\Delta\not\equiv-(z'+\Delta)\pmod n,\qquad\qquad\qquad\qquad z,z'\in Z.\label{eq:magic-translate}
\end{equation}
Moreover, such $\Delta$ can clearly be found by brute force search
in time polynomial in $|Z|=O(\log n).$ Equation~(\ref{eq:magic-translate})
now implies that no two elements of the multiset $(Z\cup\Delta)\cup(-Z-\Delta)$
are congruent modulo $n,$ and in particular no element of $Z\cup\Delta$
is congruent to $0$ modulo $n.$ 

We define $G_{n}$ to be the undirected graph with vertex set $\{0,1,2,\ldots,n-1\}$
in which $(i,j)$ is an edge if and only if $i-j$ is congruent modulo
$n$ to an element of $(Z+\Delta)\cup(-Z-\Delta).$ The roles of $i$
and $j$ in this definition are symmetric, making $G_{n}$ an undirected
graph. It is obvious that the adjacency matrix of $G_{n}$ is circulant.
Furthermore, $G_{n}$ has no self-loops because by construction no
element of $Z\cup\Delta$ is congruent to $0$ modulo $n$. Since
the elements of $(Z+\Delta)\cup(-Z-\Delta)$ are pairwise distinct
modulo $n,$ the degree of $G_{n}$ is $|(Z+\Delta)\cup(-Z-\Delta)|=2|Z|.$
Now~(\ref{eq:thm-expander-degree}) follows from~(\ref{eq:Z-cardinality-expander}).
To settle the remaining property~(\ref{eq:thm-expander-gap}), observe
that the first row of the adjacency matrix of $G_{n}$ is the characteristic
vector of the set $((Z+\Delta)\cup(-Z-\Delta))\bmod n$. As a result,
Corollary~\ref{cor:circulant-diagonalization} implies that the eigenvalues
of the adjacency matrix of $G_{n}$ are
\begin{align*}
\sum_{z\in Z+\Delta}\omega^{kz}+\sum_{z\in-Z-\Delta}\omega^{kz}, &  & k=0,1,2,\ldots,n-1,
\end{align*}
where $\omega$ is a primitive $n$-th root of unity. Setting $k=0$
yields the largest eigenvalue, $2|Z|$. The other eigenvalues are
bounded by
\begin{align*}
\lambda(G_{n}) & =\max_{k=1,2,\ldots,n-1}\left|\sum_{z\in Z+\Delta}\omega^{kz}+\sum_{z\in-Z-\Delta}\omega^{kz}\right|\\
 & \leq\max_{k=1,2,\ldots,n-1}\left|\sum_{z\in Z+\Delta}\omega^{kz}\right|+\max_{k=1,2,\ldots,n-1}\left|\sum_{z\in-Z-\Delta}\omega^{kz}\right|\\
 & =2|Z|\disc(Z,n).
\end{align*}
Along with~(\ref{eq:disc-Z-expander}), this proves~(\ref{eq:thm-expander-gap}).
\end{proof}

\section*{Acknowledgments}

I am thankful to Mark Bun, T.~S.~Jayram, Ryan O'Donnell, Rocco Servedio,
and Justin Thaler for valuable comments on this work. Special thanks
to T.~S.~Jayram for allowing me to include his short and elegant
proof of Fact~\ref{fact:number-solutions-linear-form}. 

\bibliographystyle{siamplain}
\bibliography{refs}

\appendix

\section{\label{sec:ajtai}The iteration lemma of Ajtai et al.}

The purpose of this appendix is to provide a detailed and self-contained
proof of Theorem~\ref{thm:ajtai-iteration}, which we restate below
for the reader's convenience.
\begin{thm*}
Fix an integer $R\geq1$ and a real number $P\geq2$. Let $m$ be
an integer with
\[
m\geq P^{2}(R+1).
\]
Fix a set $S_{p}\subseteq\{1,2,\ldots,p-1\}$ for each prime $p\in(P/2,P]$
with $p\nmid m,$ such that all $S_{p}$ have the same cardinality.
Consider the multiset
\begin{multline*}
S=\{(r+s\cdot(p^{-1})_{m})\bmod m:\\
\qquad r=1,\ldots,R;\quad p\in(P/2,P]\text{ prime with }p\nmid m;\quad s\in S_{p}\}.
\end{multline*}
Then the elements of $S$ are pairwise distinct and nonzero. Moreover,
\begin{equation}
\disc(S,m)\leq\frac{c}{\sqrt{R}}+\frac{c\log m}{\log\log m}\cdot\frac{\log P}{P}+\max_{p}\{\disc(S_{p},p)\}\label{eq:disc-S}
\end{equation}
for some $($explicitly given$)$ constant $c\geq1$ independent of
$P,R,m.$
\end{thm*}
\noindent This result is a slight generalization of the \emph{iteration
lemma }of Ajtai et al.~\cite{AIKPS90aperiodic-set}, which corresponds
to the special case for $m$ prime. We closely follow their proof
but provide ample detail to make it more accessible. We have structured
the presentation around five key milestones, corresponding to Sections~\ref{subsec:Shorthand-notation}\textendash \ref{subsec:Finishing-the-proof}
below. Before proceeding, the reader may wish to review the number-theoretic
preliminaries in Section~\ref{subsec:Number-theoretic-preliminaries}.

\subsection{Shorthand notation\label{subsec:Shorthand-notation}}

In the remainder of this manuscript, we adopt the shorthand
\[
e(x)=\exp(2\pi x\iu),
\]
where $\iu$ is the imaginary unit. We will need the following bounds,
illustrated in Figure~\ref{fig:graph-e}:
\begin{align}
|1-e(x)| & \leq2\pi x, &  & 0\leq x\leq1,\label{eq:e-close-to-1}\\
|1-e(x)| & \geq4\min(x,1-x), &  & 0\leq x\leq1.\label{eq:e-far-from-1}
\end{align}
To verify these bounds, write $|1-e(x)|=|1-\exp(2\pi x\iu)|=\sqrt{2-2\cos(2\pi x)}$
and apply elementary calculus.

We let $\Pcal$ denote the set of prime numbers $p\in(P/2,P]$ with
$p\nmid m.$ In this notation, the multiset $S$ is given by
\[
S=\{(r+s\cdot(p^{-1})_{m})\bmod m:p\in\Pcal,\;s\in S_{p},\;r=1,2,\ldots,R\}.
\]
There are precisely $\pi(P)-\pi(P/2)$ primes in $(P/2,P],$ of which
at most $\nu(m)$ are prime divisors of $m.$ Therefore,
\begin{equation}
|\Pcal|\geq\pi(P)-\pi\left(\frac{P}{2}\right)-\nu(m).\label{eq:Pcal-lower}
\end{equation}

\begin{figure}
\includegraphics[width=7cm]{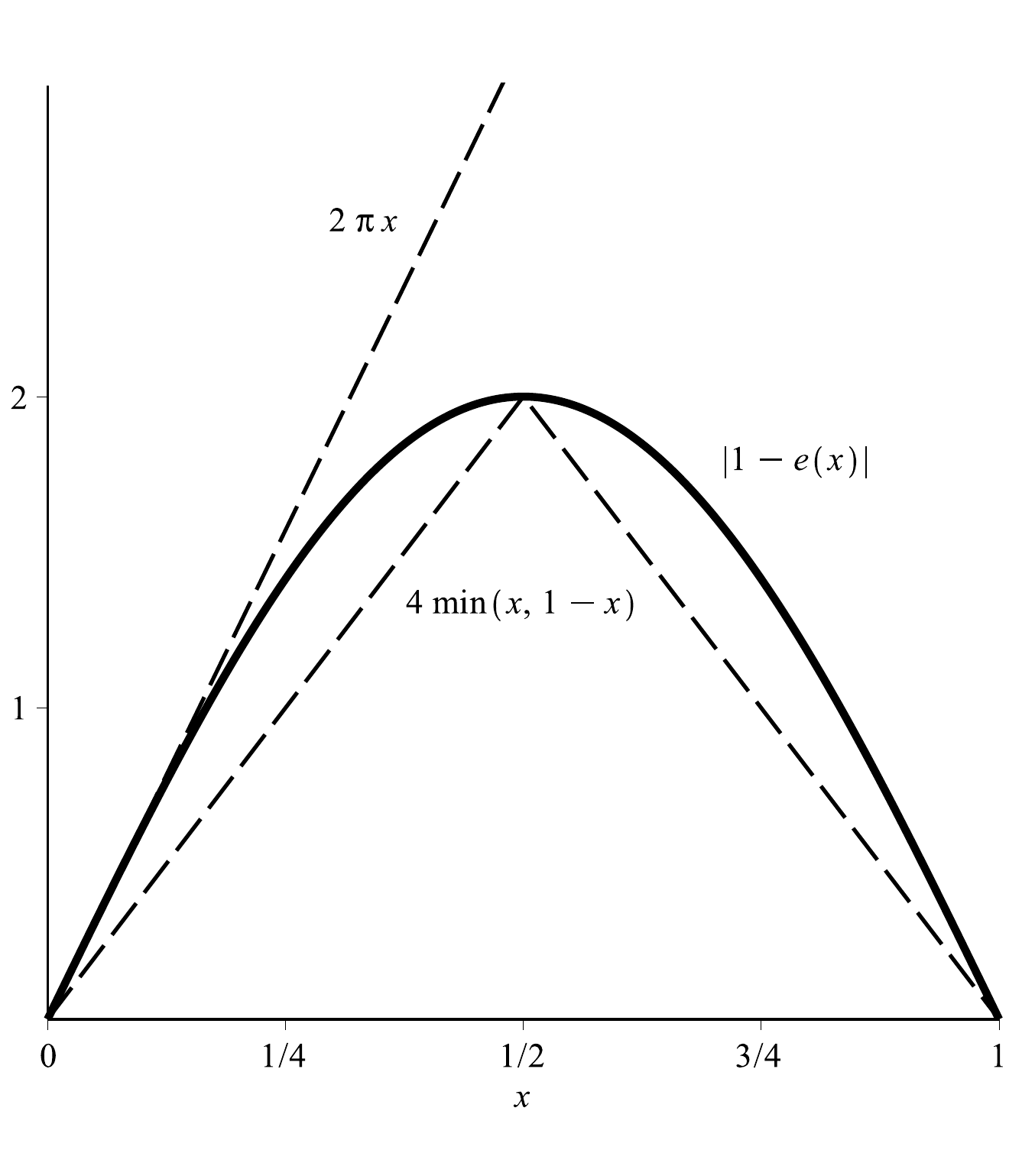}

\caption{\label{fig:graph-e}A graph of $|1-e(x)|$ and its approximations
by piecewise linear functions.}
\end{figure}

\subsection{\label{subsec:Elements-of-Sm-are-nonzero}Elements of \emph{S} are
nonzero and distinct}

As our first step, we verify that the elements of $S$ are nonzero
modulo $m$. Consider any $r\in\{1,2,\ldots,R\}$, any prime $p\in(P/2,P]$
with $p\nmid m,$ and any $s\in S_{p}.$ Then $pr+s\in[1,PR+P-1]\subseteq[1,m).$
This means that $pr+s\not\equiv0\pmod m,$ which in turn implies that
$r+s\cdot(p^{-1})_{m}\not\equiv0\pmod m$.

We now show that the multiset $S$ contains no repeated elements.
For this, consider any $r,r'\in\{1,2,\ldots,R\},$ any primes $p,p'\in\Pcal,$
and any $s\in S_{p}$ and $s'\in S_{p'}$ such that 
\begin{equation}
r+s\cdot(p^{-1})_{m}\equiv r'+s'\cdot(p'^{-1})_{m}\pmod m.\label{eq:modular-equality}
\end{equation}
Our goal is to show that $p=p',r=r',s=s'.$ To this end, multiply
(\ref{eq:modular-equality}) through by $pp'$ to obtain 
\begin{equation}
r\cdot pp'+s\cdot p'\equiv r'\cdot pp'+s'\cdot p\pmod m.\label{eq:modular-equality-pp}
\end{equation}
The left-hand side and right-hand side of (\ref{eq:modular-equality-pp})
are integers in $[1,RP^{2}+(P-1)P]\subseteq[1,m),$ whence
\begin{equation}
r\cdot pp'+s\cdot p'=r'\cdot pp'+s'\cdot p.\label{eq:straight-equality}
\end{equation}
This implies that $p\mid s\cdot p',$ which in view of $s<p$ and
the primality of $p$ and $p'$ forces $p=p'.$ Now (\ref{eq:straight-equality})
simplifies to 
\begin{equation}
r\cdot p+s=r'\cdot p+s',\label{eq:straight-equality-simplified}
\end{equation}
which in turn yields $s\equiv s'\pmod p$. Recalling that $s,s'\in\{1,2,\ldots,p-1\},$
we arrive at $s=s'.$ Finally, substituting $s=s'$ in~(\ref{eq:straight-equality-simplified})
gives $r=r'.$

\subsection{\label{subsec:Correlation-bound-for-k-small}Correlation for \emph{k}
small}

So far, we have shown that the elements of $S$ are distinct and nonzero.
Recall that our objective is to bound the $m$-discrepancy of this
set. Put another way, we must bound the exponential sum 
\begin{equation}
\left|\sum_{s\in S}e\left(\frac{k}{m}\cdot s\right)\right|\label{eq:exp-sum}
\end{equation}
for all $k=1,2,\ldots,m-1.$ This subsection and the next provide
two complementary bounds on~(\ref{eq:exp-sum}). The first bound,
presented below, is preferable when $k$ is close to zero modulo $m.$
\begin{claim}
\label{claim:k-small}Let $k\in\{1,2,\ldots,m-1\}$ be given. Then
\begin{multline*}
\left|\sum_{s\in S}e\left(\frac{k}{m}\cdot s\right)\right|\\
\leq\left(\frac{2\pi\min(k,m-k)}{m}+\max_{p\in\Pcal}\{\disc(S_{p},p)\}+\frac{\nu(k)+\nu(m-k)}{|\Pcal|}\right)|S|.
\end{multline*}
\begin{proof}
Let $\Pcal'$ be the set of those primes in $\Pcal$ that do not divide
$k$ or $m-k.$ Then clearly 
\begin{equation}
|\Pcal\setminus\Pcal'|\leq\nu(k)+\nu(m-k).\label{eq:P-minus-P'}
\end{equation}
We have
\begin{align}
 & \left|\sum_{s\in S}e\left(\frac{k}{m}\cdot s\right)\right|\nonumber \\
 & \qquad=\left|\sum_{r=1}^{R}\sum_{p\in\Pcal}\sum_{s\in S_{p}}e\left(\frac{k}{m}\cdot(r+s\cdot(p^{-1})_{m}\right)\right|\nonumber \\
 & \qquad\leq\sum_{r=1}^{R}\sum_{p\in\Pcal}\left|\sum_{s\in S_{p}}e\left(\frac{k}{m}\cdot(r+s\cdot(p^{-1})_{m}\right)\right|\nonumber \\
 & \qquad=R\sum_{p\in\Pcal}\left|\sum_{s\in S_{p}}e\left(\frac{ks\cdot(p^{-1})_{m}}{m}\right)\right|\nonumber \\
 & \qquad\le R\sum_{p\in\Pcal'}\left|\sum_{s\in S_{p}}e\left(\frac{ks\cdot(p^{-1})_{m}}{m}\right)\right|+R\sum_{p\in\Pcal\setminus\Pcal'}\left|\sum_{s\in S_{p}}e\left(\frac{ks\cdot(p^{-1})_{m}}{m}\right)\right|\nonumber \\
 & \qquad\leq R\sum_{p\in\Pcal'}\left|\sum_{s\in S_{p}}e\left(\frac{ks\cdot(p^{-1})_{m}}{m}\right)\right|+R\sum_{p\in\Pcal\setminus\Pcal'}|S_{p}|.\label{eq:k-near-zero-decomposition}
\end{align}

We proceed to bound the two summations in~(\ref{eq:k-near-zero-decomposition}).
Bounding the second summation is straightforward:

\begin{align}
R\sum_{p\in\Pcal\setminus\Pcal'}|S_{p}| & =R\cdot\frac{|\Pcal\setminus\Pcal'|}{|\Pcal|}\sum_{p\in\Pcal}|S_{p}|\nonumber \\
 & =\frac{|\Pcal\setminus\Pcal'|}{|\Pcal|}\cdot|S|\nonumber \\
 & \leq\frac{\nu(k)+\nu(m-k)}{|\Pcal|}\cdot|S|,\label{eq:k-near-zero-easy-sum}
\end{align}
where the first step is valid because all sets $S_{p}$ have the same
cardinality, and the last step uses~(\ref{eq:P-minus-P'}). 

The other summation in~(\ref{eq:k-near-zero-decomposition}) requires
more work. For $p\in\Pcal'$ and $K\in\{k,k-m\},$ we have
\begin{align*}
 & \hspace{-7mm}\left|\sum_{s\in S_{p}}e\left(\frac{ks\cdot(p^{-1})_{m}}{m}\right)\right|\\
 & =\left|\sum_{s\in S_{p}}e\left(\frac{Ks\cdot(p^{-1})_{m}}{m}\right)\right|\\
 & =\left|\sum_{s\in S_{p}}e\left(-\frac{Ks\cdot(m^{-1})_{p}}{p}\right)e\left(\frac{Ks}{pm}\right)\right|\\
 & \leq\left|\sum_{s\in S_{p}}e\left(-\frac{Ks\cdot(m^{-1})_{p}}{p}\right)\left(e\left(\frac{Ks}{pm}\right)-1\right)\right|+\left|\sum_{s\in S_{p}}e\left(-\frac{Ks\cdot(m^{-1})_{p}}{p}\right)\right|\\
 & \leq\left|\sum_{s\in S_{p}}e\left(-\frac{Ks\cdot(m^{-1})_{p}}{p}\right)\left(e\left(\frac{Ks}{pm}\right)-1\right)\right|+\disc(S_{p},p)\cdot|S_{p}|\\
 & \leq\sum_{s\in S_{p}}\left|e\left(\frac{Ks}{pm}\right)-1\right|+\disc(S_{p},p)\cdot|S_{p}|\\
 & =\sum_{s\in S_{p}}\left|e\left(\frac{|K|s}{pm}\right)-1\right|+\disc(S_{p},p)\cdot|S_{p}|\\
 & \leq|S_{p}|\cdot\frac{2\pi|K|}{m}+\disc(S_{p},p)\cdot|S_{p}|,
\end{align*}
where the second step uses Fact~\ref{fact:rel-prime} and the relative
primality of $p$ and $m$; the third step applies the triangle inequality;
the fourth step follows from $p\nmid|K|$, and the last step is valid
by~(\ref{eq:e-close-to-1}) and $s<p$. We have shown that
\begin{align*}
\left|\sum_{s\in S_{p}}e\left(\frac{ks\cdot(p^{-1})_{m}}{m}\right)\right| & \leq\frac{2\pi\min(k,m-k)}{m}\cdot|S_{p}|+\disc(S_{p},p)\cdot|S_{p}|
\end{align*}
for $p\in\Pcal'.$ Summing over $\Pcal',$
\begin{align}
R\sum_{p\in\Pcal'} & \left|\sum_{s\in S_{p}}e\left(\frac{ks\cdot(p^{-1})_{m}}{m}\right)\right|\nonumber \\
 & \qquad\leq R\sum_{p\in\Pcal'}\left(\frac{2\pi\min(k,m-k)}{m}\cdot|S_{p}|+\disc(S_{p},p)\cdot|S_{p}|\right)\nonumber \\
 & \qquad\leq R\sum_{p\in\Pcal}\left(\frac{2\pi\min(k,m-k)}{m}\cdot|S_{p}|+\disc(S_{p},p)\cdot|S_{p}|\right)\nonumber \\
 & \qquad\leq\left(\frac{2\pi\min(k,m-k)}{m}+\max_{p\in\Pcal}\{\disc(S_{p},p)\}\right)R\sum_{p\in\Pcal}|S_{p}|\nonumber \\
 & \qquad=\left(\frac{2\pi\min(k,m-k)}{m}+\max_{p\in\Pcal}\{\disc(S_{p},p)\}\right)|S|.\label{eq:k-near-zero-hard-sum}
\end{align}
By~(\ref{eq:k-near-zero-decomposition})\textendash (\ref{eq:k-near-zero-hard-sum}),
the proof of the claim is complete.
\end{proof}
\end{claim}

\subsection{\label{subsec:Correlation-bound-for-k-large}Correlation  for $k$
large}

We now present an alternative bound on the exponential sum~(\ref{eq:exp-sum}),
which is preferable to the bound of Claim~\ref{claim:k-small} when
$k$ is far from zero modulo $m.$
\begin{claim}
\label{claim:k-large}Let $k\in\{1,2,\ldots,m-1\}$ be given. Then
\[
\left|\sum_{s\in S}e\left(\frac{k}{m}\cdot s\right)\right|\leq\frac{m}{2R\min(k,m-k)}\cdot|S|.
\]
\end{claim}

\begin{proof}[Proof:]
\begin{align*}
\left|\sum_{s\in S}e\left(\frac{k}{m}\cdot s\right)\right| & =\left|\sum_{p\in\Pcal}\sum_{s\in S_{p}}\sum_{r=1}^{R}e\left(\frac{k}{m}\cdot(r+s\cdot(p^{-1})_{m})\right)\right|\\
 & \leq\sum_{p\in\Pcal}\sum_{s\in S_{p}}\left|\sum_{r=1}^{R}e\left(\frac{k}{m}\cdot(r+s\cdot(p^{-1})_{m})\right)\right|\\
 & =\sum_{p\in\Pcal}\sum_{s\in S_{p}}\left|\sum_{r=1}^{R}e\left(\frac{kr}{m}\right)\right|\\
 & =\sum_{p\in\Pcal}\sum_{s\in S_{p}}\frac{|1-e(kR/m)|}{|1-e(k/m)|}\\
 & \leq\sum_{p\in\Pcal}\sum_{s\in S_{p}}\frac{2}{|1-e(k/m)|}\\
 & \leq\sum_{p\in\Pcal}\sum_{s\in S_{p}}\frac{m}{2\min(k,m-k)}\\
 & =\frac{m}{2R\min(k,m-k)}\cdot|S|,
\end{align*}
where the last two steps use~(\ref{eq:e-far-from-1}) and $|S|=R\sum_{p\in\Pcal}|S_{p}|$,
respectively.
\end{proof}

\subsection{\label{subsec:Finishing-the-proof}Finishing the proof}

Facts~\ref{fact:PNT} and~\ref{fact:num-prime-factors} imply that
\begin{align}
\pi(P)-\pi\left(\frac{P}{2}\right) & \geq\frac{P}{C\log P}\qquad\qquad(P\geq C),\label{eq:many-primes}\\
\max_{k=1,2,\ldots,m}\nu(k) & \leq\frac{C\log m}{\log\log m},\label{eq:few-prime-factors}
\end{align}
where $C\geq1$ is a constant independent of $R,P,m.$ Moreover, $C$
can be easily calculated from the explicit bounds in Facts~\ref{fact:PNT}
and~\ref{fact:num-prime-factors}. We will show that the theorem
conclusion~(\ref{eq:disc-S}) holds with $c=4C^{2}.$ We may assume
that
\begin{align}
 & P\geq C,\label{eq:P-geq-C}\\
 & \frac{C\log m}{\log\log m}\leq\frac{P}{2C\log P},\label{eq:more-primes-than-prime-factors}
\end{align}
since otherwise the right-hand side of~(\ref{eq:disc-S}) exceeds~$1$
and the theorem is trivially true. By~(\ref{eq:Pcal-lower}), (\ref{eq:many-primes}),
(\ref{eq:few-prime-factors}), and~(\ref{eq:more-primes-than-prime-factors}),
we obtain
\[
|\Pcal|\geq\frac{P}{2C\log P},
\]
which along with~(\ref{eq:few-prime-factors}) gives
\begin{align}
\max_{k=1,2,\ldots,m-1}\frac{\nu(k)+\nu(m-k)}{|\Pcal|} & \leq\frac{2C\log m}{\log\log m}\cdot\frac{2C\log P}{P}\nonumber \\
 & =\frac{c\log m}{\log\log m}\cdot\frac{\log P}{P}.\label{eq:ratio-bound}
\end{align}
Claims~\ref{claim:k-small} and~\ref{claim:k-large} ensure that
for every $k=1,2,\ldots,m-1,$
\begin{align*}
\left|\sum_{s\in S}e\left(\frac{k}{m}\cdot s\right)\right| & \leq\left(\min\left(\frac{2\pi\min(k,m-k)}{m},\frac{m}{2R\min(k,m-k)}\right)\right.\\
 & \qquad\qquad\left.+\max_{p\in\Pcal}\{\disc(S_{p},p)\}+\frac{\nu(k)+\nu(m-k)}{|\Pcal|}\right)|S|\\
 & \leq\left(\sqrt{\frac{\pi}{R}}+\max_{p\in\Pcal}\{\disc(S_{p},p)\}+\frac{\nu(k)+\nu(m-k)}{|\Pcal|}\right)|S|\\
 & \leq\left(\frac{c}{\sqrt{R}}+\max_{p\in\Pcal}\{\disc(S_{p},p)\}+\frac{\nu(k)+\nu(m-k)}{|\Pcal|}\right)|S|.
\end{align*}
Substituting the estimate from~(\ref{eq:ratio-bound}), we conclude
that
\begin{multline*}
\max_{k=1,2,\ldots,m-1}\left|\sum_{s\in S}e\left(\frac{k}{m}\cdot s\right)\right|\\
\leq\left(\frac{c}{\sqrt{R}}+\max_{p\in\Pcal}\{\disc(S_{p},p)\}+\frac{c\log m}{\log\log m}\cdot\frac{\log P}{P}\right)|S|.\qquad\qquad
\end{multline*}
This conclusion is equivalent to~(\ref{eq:disc-S}). The proof of
Theorem~\ref{thm:ajtai-iteration} is complete.

\section{\label{app:number-solutions-linear-form}An alternate proof of Fact~\ref{fact:number-solutions-linear-form}}

The purpose of this appendix is to give an alternate, matrix-analytic
proof of Fact~\ref{fact:number-solutions-linear-form}.
\begin{fact*}[restatement of Fact~\ref{fact:number-solutions-linear-form}]
Fix a natural number $m\geq2$ and a multiset $Z=\{z_{1},z_{2},\ldots,z_{n}\}$
of integers. Let $\omega$ be a primitive $m$-th root of unity. Then
\begin{multline}
\left|\Prob_{X\in\zoon}\left[\sum_{j=1}^{n}z_{j}X_{j}\equiv s\pmod m\right]-\frac{1}{m}\right|\\
\leq\frac{1}{m}\sum_{k=1}^{m-1}\left|\prod_{j=1}^{n}\frac{1+\omega^{kz_{j}}}{2}\right|,\qquad s\in\ZZ.\label{eq:weighted-sums-uniform-1-1}
\end{multline}
\end{fact*}
\begin{proof}
For any integer $z,$ consider the circulant matrix
\[
T_{z}=\frac{1}{2}\circulant(\underbrace{1,0,\ldots,0}_{m})+\frac{1}{2}\circulant(\underbrace{\overbrace{0,\ldots,0}^{z\bmod m},1,0,\ldots,0}_{m}).
\]
By Corollary~\ref{cor:circulant-diagonalization}, the matrix $W=[\omega^{jk}/\sqrt{m}]_{j,k=0,1,\ldots,m-1}$
obeys
\begin{align}
 & WW^{*}=I,\label{eq:W-unitary-again}\\
 & W^{*}T_{z}W=\diag\left(1,\frac{1+\omega^{z}}{2},\frac{1+\omega^{2z}}{2},\ldots,\frac{1+\omega^{(m-1)z}}{2}\right), &  & z\in\ZZ.\label{eq:W-diagonalizable}
\end{align}
In particular,
\begin{align*}
W^{*} & T_{-z_{n}}T_{-z_{n-1}}\cdots T_{-z_{1}}W\\
 & \qquad=(W^{*}T_{-z_{n}}W)(W^{*}T_{-z_{n-1}}W)\cdots(W^{*}T_{-z_{1}}W)\\
 & \qquad=\prod_{j=1}^{n}\diag\left(1,\frac{1+\omega^{-z_{j}}}{2},\frac{1+\omega^{-2z_{j}}}{2},\ldots,\frac{1+\omega^{-(m-1)z_{j}}}{2}\right)\\
 & \qquad=\diag\left(1,\prod_{j=1}^{n}\frac{1+\omega^{-z_{j}}}{2},\prod_{j=1}^{n}\frac{1+\omega^{-2z_{j}}}{2},\ldots,\prod_{j=1}^{n}\frac{1+\omega^{-(m-1)z_{j}}}{2}\right),
\end{align*}
where the first two steps use~(\ref{eq:W-unitary-again}) and~(\ref{eq:W-diagonalizable}),
respectively. Applying~(\ref{eq:W-unitary-again}) yet again, we
arrive at
\begin{align*}
 & T_{-z_{n}}T_{-z_{n-1}}\cdots T_{-z_{1}}\\
 & \quad=W\diag\left(1,\prod_{j=1}^{n}\frac{1+\omega^{-z_{j}}}{2},\prod_{j=1}^{n}\frac{1+\omega^{-2z_{j}}}{2},\ldots,\prod_{j=1}^{n}\frac{1+\omega^{-(m-1)z_{j}}}{2}\right)W^{*}\\
 & \quad=\frac{1}{m}J+\sum_{k=1}^{m-1}\prod_{j=1}^{n}\frac{1+\omega^{-kz_{j}}}{2}W_{k}W_{k}^{*},
\end{align*}
where $W_{1},W_{2},\ldots,W_{m-1}$ denote the last $m-1$ columns
of $W.$ Since the components of each $W_{k}$ are bounded in absolute
value by $1/\sqrt{m},$ we conclude that
\begin{equation}
\left\Vert T_{-z_{n}}T_{-z_{n-1}}\cdots T_{-z_{1}}-\frac{1}{m}J\right\Vert _{\infty}\leq\frac{1}{m}\sum_{k=1}^{m-1}\left|\prod_{j=1}^{n}\frac{1+\omega^{-kz_{j}}}{2}\right|.\label{eq:infty-distance-in-terms-of-eigens}
\end{equation}

We are now in a position to prove~(\ref{eq:weighted-sums-uniform-1-1}).
Let $X=(X_{1},X_{2},\ldots,X_{n})$ be a random variable distributed
uniformly in $\zoon.$ Consider the random variables $Y_{0},Y_{1},Y_{2},\ldots,Y_{n}$
given by $Y_{k}=(z_{1}X_{1}+z_{2}X_{2}+\cdots+z_{k}X_{k})\bmod m.$
The sequence $Y_{0},Y_{1},Y_{2},\ldots,Y_{n}$ has a natural interpretation
in terms of an $n$-step random walk in $\ZZ_{m}.$ Specifically,
the random walk starts at $Y_{0}=0$ and evolves according to
\[
Y_{k}=\begin{cases}
Y_{k-1} & \text{with probability \ensuremath{1/2,}}\\
(Y_{k-1}+z_{k})\bmod m & \text{with probability \ensuremath{1/2.}}
\end{cases}
\]
In particular, the $k$-th step of the random walk has transition
probability matrix
\[
\begin{matrix}\\
{\displaystyle \frac{1}{2}}
\end{matrix}\begin{matrix}\phantom{\overbrace{\quad\qquad\qquad}^{z_{k}\bmod m}}\\
\begin{bmatrix}1\\
 & 1\\
 &  & 1\\
 &  &  & 1\\
 &  &  &  & \ddots\\
 &  &  &  &  & 1\\
 &  &  &  &  &  & 1
\end{bmatrix}
\end{matrix}\begin{matrix}\\
{\displaystyle \;+\;\frac{1}{2}}
\end{matrix}\begin{matrix}\overbrace{\quad\qquad\qquad}^{-z_{k}\bmod m}\qquad\qquad\qquad\\
\begin{bmatrix} &  &  & 1\\
 &  &  &  & 1\\
 &  &  &  &  & \ddots\\
 &  &  &  &  &  & 1\\
1\\
 & \ddots\\
 &  & 1
\end{bmatrix}
\end{matrix}\begin{matrix}\\
,{\displaystyle \phantom{\frac{1}{2}}}
\end{matrix}
\]
where the unspecified entries are zero, and the rows and columns correspond
in the usual manner to the values $0,1,\ldots,m-1$. In the notation
of the opening paragraph of the proof, this matrix is precisely $T_{-z_{k}}.$
Letting $p_{0},p_{1},\ldots,p_{n}$ be the $m$-dimensional vectors
that represent the probability distributions of $Y_{0},Y_{1},\ldots,Y_{n},$
respectively, we obtain the recursive relations $p_{k}=T_{-z_{k}}p_{k-1}$.
Therefore,
\[
p_{n}=T_{-z_{n}}T_{-z_{n-1}}\cdots T_{-z_{1}}p_{0}.
\]
 Now
\begin{align*}
\left\Vert p_{n}-\begin{bmatrix}{\displaystyle \frac{1}{m}} & {\displaystyle \frac{1}{m}} & \cdots & {\displaystyle \frac{1}{m}}\end{bmatrix}^{T}\right\Vert _{\infty} & =\left\Vert T_{-z_{n}}T_{-z_{n-1}}\cdots T_{-z_{1}}p_{0}-\frac{1}{m}Jp_{0}\right\Vert _{\infty}\\
 & \leq\left\Vert T_{-z_{n}}T_{-z_{n-1}}\cdots T_{-z_{1}}-\frac{1}{m}J\right\Vert _{\infty}\|p_{0}\|_{1}\\
 & =\left\Vert T_{-z_{n}}T_{-z_{n-1}}\cdots T_{-z_{1}}-\frac{1}{m}J\right\Vert _{\infty}\\
 & \leq\frac{1}{m}\sum_{k=1}^{m-1}\left|\prod_{j=1}^{n}\frac{1+\omega^{-kz_{j}}}{2}\right|\\
 & =\frac{1}{m}\sum_{k=1}^{m-1}\left|\prod_{j=1}^{n}\frac{1+\omega^{kz_{j}}}{2}\right|,
\end{align*}
where the next-to-last step uses~(\ref{eq:infty-distance-in-terms-of-eigens}).
This conclusion is obviously equivalent to~(\ref{eq:weighted-sums-uniform-1-1}).
\end{proof}

\end{document}